\documentclass[11pt,a4paper]{amsart}
\usepackage{amsthm}
\usepackage{amssymb}
\usepackage{graphicx}
\usepackage{enumerate}
\usepackage{xcolor}
\usepackage{verbatim}
\usepackage[hyperindex,breaklinks]{hyperref}
\usepackage{braket}
\usepackage[all,cmtip]{xy}
\usepackage{tikz}
\usepackage{tikz-cd}
\usetikzlibrary{positioning, 
                quotes}
    \tikzset{
    path/.style={line width=2.2pt, red}
    }
    \usetikzlibrary{calc}
\usetikzlibrary{decorations.pathmorphing,shapes}

\usepackage{subcaption}
\usepackage{float}  
\usepackage{adjustbox} 
                
\newcommand{\RR}{\mathbb{R}}
\newcommand{\mr}{\mathrm}
\newcommand{\ra}{\rightarrow}
\newcommand{\mc}{\mathcal}
\newcommand{\HH}{\mathcal{H}}
\newcommand{\xra}{\xrightarrow}
\newcommand{\xla}{\xleftarrow}
\newcommand{\CC}{\mathbb{C}}
\newcommand{\ii}{\mathrm{in}}
\newcommand{\oo}{\mathrm{out}}
\newcommand{\dd}{\partial}
\newcommand{\DN}{\mathrm{DN}}
\newcommand{\wh}{\widehat}
\newcommand{\til}{\widetilde}
\newcommand{\val}{\mathrm{val}}
\newcommand{\ol}{\overline}
\newcommand{\DNI}{\mathrm{DN}_\mathrm{int}}

\newcommand{\GCob}{\mathrm{GraphCob}}
\newcommand{\T}{\mathsf{T}}
\newcommand{\dvol}{d\mathrm{vol}}
\newcommand{\s}{\mathsf}

\newcommand{\ggamma}{\pmb{\gamma}}
\newcommand{\xxi}{\pmb{\xi}}
\newcommand{\ul}{\underline}
\newcommand{\D}{\mathcal{D}}

\theoremstyle{remark}
\newtheorem{remark}{Remark}[section]
\theoremstyle{plain}
\newtheorem{lemma}[remark]{Lemma}
\newtheorem{proposition}[remark]{Proposition}
\newtheorem{thm}[remark]{Theorem}
\newtheorem{corollary}[remark]{Corollary}
\theoremstyle{definition}
\newtheorem{definition}[remark]{Definition}
\newtheorem{example}[remark]{Example}

\usepackage{fancyhdr}
\fancypagestyle{fancyams}{
   \fancyhead[RE,LO]{}
    \fancyhead[RO,LE]{\scriptsize\thepage}
    \fancyhead[CO]{\scriptsize\MakeUppercase{
    Combinatorial QFT on graphs: first quantization formalism
    }
    } 
    \fancyhead[CE]{\scriptsize\MakeUppercase{I. Contreras, S. Kandel,  P. Mnev, and K. Wernli}}
    
    \fancyfoot{}
}

\usepackage[automake,toc = false]{glossaries-extra}
\glsxtrRevertMarks
\newglossarystyle{myNoHeaderStyle}{
   \setglossarystyle{longheader} 
}

\makeglossaries 
 \newglossaryentry{X}
 {
   name={$X$},
  sort={},
  description={Spacetime graph}
}
\newglossaryentry{Y}
 {
   name={$Y$, or $Y_i$},
  sort={},
  description={A subgraph $Y \subset X$  (or several subgraphs $Y_i \subset X$) of the spacetime graph }
}
\newglossaryentry{Laplacian}
{
 name={$\Delta_{X}$},
 sort={},
 description={Laplacian on $X$}
}

\newglossaryentry{Kinetic Operator}
{
  name={$K_X$},
  sort={},
  description={Kinetic operator on $X$, $K_X = \Delta_X + m^2$}
}

\newglossaryentry{Gamma}
{
  name={$\Gamma$},
  sort={},
  description={Feynman graph}
}

\newglossaryentry{Propagator}
{
  name={$G(x,y)$},
  sort={},
  description={The propagator or Green's function of the kinetic operator $K_X$, integral kernel (matrix) of $G = K_X^{-1}$}
}

\newglossaryentry{Partition function}
{
  name={$Z_X$},
  sort={},
  description={The partition function on $X$}
}

\newglossaryentry{Perturbative partition function}
{
  name={$Z^{\mr{pert}}_X$},
  sort={},
  description={The perturbative partition function on $X$}
}

\newglossaryentry{Extension operator}
{
  name={$E_{Y,X}$},
  sort={},
  description={The extension operator (also known as Poisson operator): extends a field $\phi_Y$ into bulk $X$ as a solution of Dirichlet problem}
}

\newglossaryentry{Dirichlet-to-Neumann operator}
{
  name={$\DN_{Y,X}$},
  sort={},
  description={Dirichlet-to-Neumann operator}
}

\newglossaryentry{Action functional}
{
  name={$S_X$},
  sort={},
  description={Action functional on the space of fields on $X$}
}

\newglossaryentry{First quantized action functional}
{
  name={$S_X^{1q}$},
  sort={},
  description={First quantized action functional}
}

\newglossaryentry{Zeta-regularized determinant}
{
  name={${\det}^\zeta A$},
  sort={},
  description={Zeta-regularized determinant of operator $A$}
}

\newglossaryentry{Path space}
{
  name={$P_X(u,v)$},
  sort={},
  description={The set of paths in $X$ joining $u$ to $v$}
}

\newglossaryentry{h-path space}
{
  name={$\Pi_X(u,v)$},
  sort={},
  description={The set of h-paths in $X$ joining $u$ to $v$}
}

\newglossaryentry{edge-to-path space}
{
  name={$P_{X}^{\Gamma}$},
  sort={},
  description={The set of edge-to-path maps from $\Gamma$ to $X$}
}

\newglossaryentry{edge-to-h-path space}
{
  name={$\Pi_{X}^{\Gamma}$},
  sort={},
  description={The set of edge-to-h-path maps from $\Gamma$ to $X$}
}

\newglossaryentry{degree of a h-path}
{
  name={$\mr{deg}(\tilde\gamma)$},
  sort={},
  description={The number of jumps in an h-path $\tilde\gamma$}
}

\newglossaryentry{length of a path}
{
  name={$l(\gamma)$},
  sort={},
  description={The length of a path (or h-path) $\gamma$ }
}

\newglossaryentry{h}{
name={$h(\tilde\gamma)$},
sort={},
description={The number of hesitations of an h-path $\tilde\gamma$}
}

\newglossaryentry{s}{
name={$s(\tilde\gamma)$},
sort={},
description={Weight of an h-path, $s(\tilde\gamma)=m^{-2l(\tilde\gamma)}(-1)^{h(\tilde\gamma)}$}
}

\newglossaryentry{w}
{
  name={$w(\gamma)$},
  sort={},
  description={Weight of a path, $w(\gamma)=\prod_{v\in V(\gamma)}\dfrac{1}{m^2+\mr{val}(v)}$}
}

\allowdisplaybreaks

\begin{document}

\title{
Combinatorial QFT on graphs: first quantization formalism
}
\author[I. Contreras]{Ivan Contreras}
\address{Department of Mathematics, Amherst College, 31 Quadrangle Drive, Amherst, MA 01002, USA}
\email{icontreraspalacios@amherst.edu}

\author[S. Kandel]{Santosh Kandel}
\address{California State University, Sacramento, Department of Mathematics and Statistics, CA 95819, USA}
\email{kandel@csus.edu}

\author[P. Mnev]{Pavel Mnev}
\address{University of Notre Dame, Notre Dame, IN 46556, USA}
\address{St. Petersburg Department of V. A. Steklov Institute of Mathematics of the Russian Academy of Sciences, 
27 Fontanka, St. Petersburg, Russia, 191023}
\email{pmnev@nd.edu}

\author[K. Wernli]{Konstantin Wernli}
\address{Centre for Quantum Mathematics, IMADA, University of Southern Denmark, Campusvej 55, 5230 Odense M, Denmark}
\email{kwernli@imada.sdu.dk}

\thanks{
The work of K. Wernli was supported by the ERCSyG project, Recursive and Exact New Quantum Theory (ReNewQuantum) which received
funding from the European Research Council (ERC) under the European Union’s Horizon 2020 research and innovation programme
under grant agreement No. 810573.}

\begin{abstract}
    We study a
    combinatorial model of the quantum scalar field with polynomial potential
    on a graph. In the first quantization formalism, the value of a Feynman graph is given by a sum over maps from the Feynman graph to the spacetime graph (mapping edges to paths). This picture interacts naturally with Atiyah-Segal-like cutting-gluing of spacetime graphs. In particular, one has combinatorial counterparts of the known gluing formulae for Green's functions and (zeta-regularized) determinants of Laplacians.
\end{abstract}
\maketitle

\setcounter{tocdepth}{3}
\tableofcontents

\section{Introduction}
%
In this paper we study a combinatorial model of the quantum massive scalar field with polynomial potential on a spacetime given a by a graph $X$. Our motivation to do so was the study of the first quantization formalism, that we recall in Section \ref{sec: first quantization} below, and in particular its interplay with locality, i.e. cutting and gluing of the spacetime manifold. At the origin is the Feynman-Kac formula \eqref{G 1st quant formula} for the Green's function of the kinetic operator. In case the spacetime is a graph, this formula has a combinatorial analog given by summing over paths with certain weights (see Section \ref{s: path sum formulae for the propagator and determinant}). These path sums interact very naturally with cutting and gluing, in a mathematically rigorous way, see Theorem \ref{thm: gluing prop and det} and its proof from path sum formulae \ref{sec: gluing proof path sums}. 

A second motivation to study this model was the notion of (extended) functorial QFTs with source a Riemannian cobordism. Few examples of functorial QFTs out of Riemannian cobordism categories exist, for instance  \cite{K}, \cite{KMW},\cite{Pickrell}. In this paper, we define a graph cobordism category and show that the combinatorial model defines a functor to the category of Hilbert spaces (Section \ref{ss: functorial picture}). We also propose an extended cobordism (partial) $n$-category of graphs  and a functor to a target $n$-category of commutative algebras defined by the combinatorial QFT we are studying (Section \ref{ss: extended graph QFT}). 

Finally, one can use this discrete toy model to approximate the continuum theory, which in this paper we do only in easy one-dimensional examples (see Section \ref{ss: gluing in Gaussian theory: comparison to continuum}). We think that the results derived in this paper will be helpful to study the interplay between renormalization and locality in higher dimensions (the two-dimensional case was discussed in detail in \cite{KMW}). 

\subsection{Motivation: first quantization formalism}\label{sec: first quantization}
We outline 
the idea of the first quantization picture in QFT 
in the example of the interacting scalar field.\footnote{We refer the reader to the inspiring exposition of this idea in \cite[Section 3.2]{Dijkgraaf97}.}

Consider the scalar field theory on a Riemannian $n$-manifold $M$ perturbed by  a polynomial potential $p(\phi)=\sum_{k\geq 3}\frac{p_k}{k!}\phi^k$, defined by the action functional
\begin{equation}\label{S (intro)}
S(\phi)=\int_M \left(\frac12\phi (\Delta+m^2) \phi+p(\phi)\right) d^nx. 
\end{equation}
Here $\phi\in C^\infty(M)$ is the field, $\Delta$ is the Laplacian determined by the metric, $m>0$ is the mass parameter and $d^nx$ denotes the metric volume element.

The partition function is formally given by a (mathematically ill-defined) functional integral understood perturbatively as a sum over Feynman graphs $\Gamma$,\footnote{In this discussion we will ignore the issue of divergencies and renormalization.}
\begin{equation} \label{Z pert expansion (intro)}
    Z_M=``\int_{C^\infty(M)} \mc{D}\phi \; e^{-\frac{1}{\hbar} S(\phi)}"\; = \left({\det}^\zeta(\Delta+m^2)\right)^{-\frac12}\cdot \sum_\Gamma \Phi_\Gamma.
\end{equation}
Here ${\det}^\zeta$ is the functional determinant in zeta function regularization.
The weight $\Phi_\Gamma$ of a Feynman graph is the product of Green's functions $G(x,y)$ of the kinetic operator $\Delta+m^2$ associated with the edges of $\Gamma$, integrated over ways to position vertices of $\Gamma$ at points of $M$ (times the vertex factors, a symmetry factor and a loop-counting factor):
\begin{equation}\label{Feynman weight (intro)}
    \Phi_\Gamma= \frac{\hbar^{|E|-|V|}}{|\mr{Aut}(\Gamma)|} \int_{M^{\times V}} d^n x_1 \cdots d^n x_{|V|} \; \prod_{(u,v)\in E} G(x_u,x_v) \cdot \prod_{v\in V} (-p_{\mr{val}(v)}).
\end{equation}
Here $V,E$ are the set of vertices and the set of edges of $\Gamma$, respectively.

Next, one can understand the kinetic operator $\Delta+m^2=\colon\wh{H}$ as a quantum Hamiltonian of an auxiliary quantum mechanical system with Hilbert space $L^2(M)$. Then, one can write the Green's function $G(x,y)$ as the evolution operator of this auxiliary system integrated over the time of evolution:
$$
 G(x,y)= \int_0^\infty dt\, \langle x|  e^{-t \wh{H}} |y\rangle.
$$
Replacing the evolution operator (a.k.a. heat kernel) with its Feynman-Kac path integral representation, one has
\begin{equation}\label{G 1st quant formula}
G(x,y)=\int_0^\infty dt \int_{\gamma\colon [0,t]\ra M,\; \gamma(0)=x,\, \gamma(t)=y} 
\mc{D}\gamma\; e^{-
S^{1q}(\gamma)
}.
\end{equation}
Here the inner integral is over paths $\gamma$ on $M$ parameterized by the interval $[0,t]$, starting at $y$ and ending at $x$; the auxiliary (``first quantization'') action in the exponent is
\begin{equation}\label{S^1q}
    S^{1q}(\gamma)=\int_0^t d\tau \left(\frac{\dot{\gamma}^2}{4}+m^2-\frac16 R(\gamma)\right),
\end{equation}
where $R$ is the scalar curvature of the metric on $M$; $\dot\gamma^2\colon=g_{\gamma(\tau)}(\dot\gamma,\dot \gamma)$ is the square norm of the velocity vector $\dot\gamma\in T_{\gamma(\tau)}M$ of the path w.r.t. the metric $g$ on $M$.\footnote{The action (\ref{S^1q}) can be obtained from the short-time asymptotics (Seeley-DeWitt expansion) of the heat kernel $\kappa(x,y;t)=\langle x|  e^{-t \wh{H}} |y\rangle \underset{t\ra 0}{\sim} (4\pi t)^{-\frac{n}{2}} e^{-\frac{d(x,y)^2}{4t}} (1+b_2(x,y)t+b_4(x,y)t^2+\cdots)$, with $b_{2k}$ smooth functions on $M\times M$ (in particular, on the diagonal), with $b_2(x,x)=-m^2+\frac16 R(x)$; $d(x,y)$ is the geodesic distance on $M$, see e.g. \cite{Vassilevich}.
One then has $\displaystyle \kappa(x,y;t)=\lim_{N\ra \infty} \int_{M^{\times(N-1)}}d^nx_{N-1} \cdots d^nx_1\kappa(x_N=x,x_{N-1};\delta t)\cdots
\kappa(x_2,x_1;\delta t) \kappa(x_1,x_0=y;\delta t)= \lim_{N\ra \infty} 
\int \prod_{j=1}^{N-1}\left((4\pi \delta t)^{-\frac{n}{2}}dx_j\right) \, e^{-\sum_{j=1}^N \left(\frac{d(x_{j},x_{j-1})^2}{4(\delta t)^2}+(m^2-\frac16 R(x_j))\right) \delta t}$.
In the r.h.s. one recognizes the path integral of (\ref{G 1st quant formula}) written as a limit of finite-dimensional integrals (cf. \cite{Feynman-Hibbs}). We denoted $\delta t=t/N$.
}

Plugging the integral representation (\ref{G 1st quant formula}) of the Green's function into (\ref{Feynman weight (intro)}), one obtains the following integral formula for the weight of a Feynman graph:
\begin{multline}\label{Feynman weight 1q formalism}
    \Phi_\Gamma = \frac{\hbar^{|E|-|V|}}{|\mr{Aut}(\Gamma)|}\int_{0<t_1,\ldots,t_{|E|}<\infty} dt_1\cdots dt_{|E|} \int_{\ggamma\colon \Gamma_{t_1,\ldots,t_{|E|}}\ra M} \mc{D}\ggamma\, e^{-S^{1q}(\ggamma)} \cdot \\
    \cdot \prod_{v\in V} \left(-p_{\mr{val}(v)}\right).
\end{multline}
Here $\Gamma_{t_1,\ldots,t_{|E|}}$ is the graph $\Gamma$ seen as a metric graph with $t_e$ the length of edge $e$. The outer integral is over metrics on $\Gamma$, the inner (path) integral is over maps $\ggamma$ of $\Gamma$ to $M$, sending vertices to points of $M$ and edges to paths connecting those points; $S^{1q}(\ggamma)$ is understood as a sum of expressions in the r.h.s. of (\ref{S^1q}) over edges of $\Gamma$.

We refer to the formula (\ref{Feynman weight 1q formalism}), representing the weight of a Feynman graph via an integral over maps $\Gamma\ra M$ (or, equivalently, as a partition function of an auxiliary 1d sigma model on the graph $\Gamma$ with target $M$), as the ``first quantization formula.''\footnote{As opposed to the functional integral (\ref{Z pert expansion (intro)}) -- the ``second quantization formula.''}

\begin{remark}
   It is known that $\frac{1}{6}R$ appears in the quantum Hamiltonian of the quantum mechanical system of a free particle on a closed Riemannian manifold $M$, see for example \cite{Anderson-Driver, Woodhouse}. Here, the difference is that  $\frac{1}{6}R$  is introduced  in the classical action (\ref{S^1q}) so that $\Delta + m^2$ is the quantum Hamiltonian.  
\end{remark}

\begin{remark}
    One can absorb the determinant factor in the r.h.s. of (\ref{Z pert expansion (intro)}) into the sum over graphs, if we extend the set of graphs $\Gamma$ to allow them to have circle connected components (with no vertices), with the rule
    \begin{multline}\label{Feynman weight of S^1}
        \Phi_{S^1}=-\frac12 \log{\det}^\zeta (\Delta+m^2) =
        \frac12 \int_0^\infty \frac{dt}{t} \,\operatorname{tr}\left( e^{-t\wh{H}}\right) \\=
        \frac12 \int_0^\infty \frac{dt}{t} \int_{\gamma\colon S^1_t\ra M} \mc{D}\gamma \,e^{-S^{1q}(\gamma)},
    \end{multline}
where the integral in $t$ is understood in zeta-regularized sense; $S^1_t=\mathbb{R}/t\mathbb{Z}$ is the circle 
of perimeter $t$. 
\end{remark}

\subsubsection{Version with 1d gravity.}
Another way to write the formula (\ref{G 1st quant formula}) is to consider paths $\gamma$ parameterized by the standard interval $I=[0,1]$ (with coordinate $\sigma$) 
and  introduce an extra field -- the metric $\xi=\ul{\xi}(\sigma)(d\sigma)^2$ on $I$:
\begin{equation}\label{G 1q with 1d gravity}
    G(x,y)= \int_{(\mr{Met}(I)\times \mr{Map}(I,M)_{x,y})/\mr{Diff(I)}}\mc{D}\xi\, 
    \mc{D}\gamma \, e^{- \bar{S}^{1q} (\gamma,\xi)}.
\end{equation}
Here $\mr{Map}(I,M)_{x,y}$ is the space of paths $\gamma\colon I\ra M$ from $x$ to $y$; the exponent in the integrand is
\begin{multline}\label{S^1q bar}
    \bar{S}^{1q} (\gamma,\xi)= \int_I \left(\frac14 (\xi^{-1}\otimes \gamma^* g)(d\gamma,d\gamma) + m^2-\frac16 R(\gamma)\right) d\mr{vol}_\xi\\
    = 
    \int_0^1 \left(\frac{\dot{\gamma}^2}{4\ul{\xi}} + m^2-\frac16 R(\gamma)\right) \sqrt{\ul{\xi}}\, d\sigma,
\end{multline}
with $d\mr{vol}_\xi$ the Riemannian volume form of $I$ induced by $\xi$. Note that the action (\ref{S^1q bar}) is invariant under diffeomorphisms of $I$. One can gauge-fix this symmetry by requiring that the metric $\xi$ is constant on $I$, then one is left with integration over the length $t$ of $I$ w.r.t. the constant metric; this reduces the formula (\ref{G 1q with 1d gravity}) back to (\ref{G 1st quant formula}). 

In (\ref{G 1q with 1d gravity}), the Green's function of the original theory on $M$ is understood in terms of a 1d sigma-model on $I$ with target $M$ coupled to 1d gravity.
For a Feynman graph, similarly to (\ref{Feynman weight 1q formalism}), one has
\begin{equation}\label{Feynman graph 1q with 1dgrav}
    \Phi_\Gamma=\hbar^{|E|-|V|}\int_{(\mr{Met}(\Gamma)\times \mr{Map}(\Gamma,M))/\mr{Diff}(\Gamma)}\mc{D}\xxi\,\mc{D}\ggamma \, e^{-\bar{S}^{1q}(\ggamma,\xxi)} \prod_{v\in V} \left(-p_{\mr{val}(v)}\right)
\end{equation}
-- the partition function of 1d sigma model on the Feynman graph $\Gamma$ coupled to 1d gravity on $\Gamma$; $\bar{S}^{1q}(\ggamma,\xxi)$ is understood as a sum of terms (\ref{S^1q bar}) over the edges of $\Gamma$.\footnote{If one thinks of the quotient by $\mr{Diff}(\Gamma)$ in (\ref{Feynman graph 1q with 1dgrav}) as a stack, one can see the symmetry factor $\frac{1}{|\mr{Aut}(\Gamma)|}$ as implicitly contained in the integral. 
To see that, one should think of the quotient by $\mr{Diff}(\Gamma)$ as first a quotient by the connected component of the identity map $\mr{Diff}_0(\Gamma)$ and then a quotient by the mapping class group $\pi_0 \mr{Diff}(\Gamma)=\mr{Aut}(\Gamma)$.

One can interpret the $1/t$ factor in the r.h.s. of (\ref{Feynman weight of S^1}) in a similar fashion: $\mr{Diff}(S^1)$ splits as $\mr{Diff}(S^1,\mr{pt})\times S^1$ --diffeomorphisms preserving a marked point on $S^1$, plus an extra factor $S^1$ corresponding to rigid rotations (moving the marked point). It is that extra $S^1$ factor that leads to the factor $1/t=1/\mr{vol}(S^1)$  in the integration measure over $t$. The factor $1/2$ in (\ref{Feynman weight of S^1}) comes by the previous mechanism from the mapping class group $\mathbb{Z}_2$ of $S^1$ (orientation preserving/reversing diffeomorphisms up to isotopy).
}
\subsubsection{Heuristics on locality in the first quantization formalism}
Suppose that we have a decomposition $M = M_1 \cup_Y M_2$ of $M$ into two Riemannian manifolds $M_i$, with common boundary $Y$. 
Then locality of quantum field theory -- or, a fictional ``Fubini theorem'' for the (also fictional) functional integral -- suggests a gluing formula 
\begin{equation}
Z_M = ``\int_{C^\infty(Y)}\mathcal{D}\phi_Y\; Z_{M_1}(\phi_Y)Z_{M_2}(\phi_Y)," \label{eq: gluing formula introduction}
\end{equation}
where $Z_{M_i}$ is a functional of $C^\infty(Y)$, again formally given by a functional integral understood as a sum over Feynman graphs,\footnote{Again, for the purpose of this motivational section we are not discussing the problem of divergencies and renormalization. For $n = \dim M = 2$, a precise definition of all involved objects and a proof of the gluing formula \eqref{eq: gluing formula introduction} can be found in \cite{KMW}. }
\begin{multline}
    Z_{M_i}(\phi_Y) = ``\int_{\substack{\phi \in C^\infty(M) \\ \phi|_{Y} = \phi_Y}}\mc{D}\phi\; e^{-\frac{1}{\hbar} S(\phi)}"\\ 
    =  \left({\det}^\zeta(\Delta_{M_i,Y}+m^2)\right)^{-\frac12}\cdot \sum_\Gamma \Phi_\Gamma(\phi_Y), \label{eq: rel Z pert intro} 
\end{multline} 
where we are putting Dirichlet boundary conditions on the kinetic operator. Feynman graphs now have 
bulk and boundary vertices, $V = V^{\mr{bulk}} \sqcup V^\partial$, where boundary vertices are required to be univalent. The set of edges then decomposes as $E = \sqcup_{i=0}^2 E_i$ with where edges in $E_i$ have $i$ endpoints in $V^\partial$. The weight of a Feynman graph then is 
\begin{multline}
    \Phi_\Gamma(\phi_Y) = \\ =\frac{\hbar^{|E|-|V|}}{|\mr{Aut}(\Gamma)|}\int_{M_i^{\times V^\mr{bulk}}}d^nx_1\cdots d^nx_{|V^\mr{bulk}|}\int_{Y^{V^\partial}} d^{n-1}y_1\cdots d^{n-1}y_{|V^\partial|} \\
    \prod_{v\in V^\mr{bulk}}(-p_{\val(v)})\prod_{w\in V^\partial}\phi_Y(y_w) \\ 
    \prod_{(u,v) \in E_0}G_{M_i,Y}(x_u,x_v)\prod_{(u,v) \in E_1}E_{Y,M_i}(x_u,y_v)\prod_{(u,v)\in E_2}-\DN_{Y,M_i}(y_u,y_v), \label{eq: intro feynman weight rel}
\end{multline}
where $G_{M_i,Y}$ denotes the Green's function of the operator with Dirichlet boundary conditions, $E_{Y,M_i}(x,y) = \partial_{n_y}G(x,y)$ is the normal derivative of the Green's function at a boundary point $y \in Y$, and $\DN_{Y,M_i}$ is the Dirichlet-to-Neumann operator associated to the kinetic operator (see Section \ref{ss: gluing in Gaussian theory: comparison to continuum} for details). 

Let us sketch an interpretation 
of the gluing formula for the Green's function from the standpoint of the first quantization formalism.
Let $x \in M_1$, $y\in M_2$ and consider a path $\gamma \colon [0,t] \to M$ with $\gamma(0) = x$ and $\gamma(t) = y$. Then the decomposition $M = M_1 \cup_Y M_2$ induces a decomposition $\gamma = \gamma_1 * \gamma_2 * \gamma_3$ as follows (``$*$'' means concatenation of paths). Let $t_0 = 0$, $t_1 = \mr{min} \{t, \gamma(t) \in Y\}$ and $t_2 =\mr{max} \{t, \gamma(t) \in Y\}$ and $t_3 = t$, then $\gamma_i = \gamma|_{[t_{i-1},t_i]}$. This gives a decomposition 
$$P_M(x,y) = \bigsqcup_{u,v \in Y}P'_{M_1}(x,u) \times P_M(u,v) \times P'_{M_2}(v,y),$$ 
where we have introduced the notation $P_M(x,y)$ for the set of all paths from $x$ to $y$ (of arbitrary length) and $P'_{M_i}(x,u)$ for the set of all paths starting at $x \in M_i$ and ending at $u \in Y$ and not intersecting $Y$ in between. See Figure \ref{fig: intro path decomp}.
\begin{figure}[H]
    \centering
    \begin{tikzpicture}
    \draw (0,0) ellipse (5cm and 3cm); 
    \node[left] at (-3,0) {$x$}; 
    \node[right] at (2,1) {$y$}; 
    \draw[thick] (0,-3) -- (0,3); 
    \draw plot [smooth] coordinates {(-3,0)  (-1.5,1.5)  (0,2)  (1,1.5)  (0,1)  (-1,0) (0,-1.5)  (1,-1)  (2,1)};
    \node[above] at (-1.5,1.5) {$\gamma_1$};
    \node[left] at (-1,0) {$\gamma_2$};
    \node[right] at (1,-1) {$\gamma_3$};
    \draw[fill = black] (0,2) circle (2pt);
    \draw[fill = black] (-3,0) circle (2pt);
    \draw[fill = black] (2,1) circle (2pt);
    \draw[fill = black] (0,-1.5) circle (2pt);
    \node[right] at (-5,0) {$M_1$}; 
    \node[left] at (5,0) {$M_2$}; 
    \node[left] at (0,-2.7) {$Y$};
    \end{tikzpicture}
    \caption{Decomposing a path $\gamma = \gamma_1 * \gamma_2 * \gamma_3$.}
    \label{fig: intro path decomp}
\end{figure}
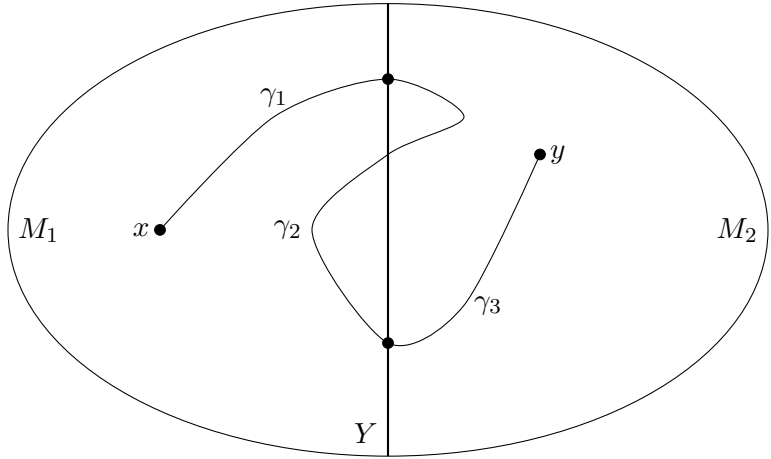

Paths of a specific length $t$ will be denoted $P^t_M(x,y)$, or $(P'_M)^t(x,u)$. Assuming a Fubini theorem for the path measure $\D\gamma$, additivity of the action suggests that we could rewrite \eqref{G 1st quant formula} as 
\begin{multline}
    G(x,y) ``=" \int_{Y \times Y}d^{n-1}u\, d^{n-1}v \int_{0}^\infty dt_1 \int_{\gamma_1\in (P'_{M_1})^{t_1}(x,u)} \D\gamma_1\, e^{-S^{\mr{1q}}(\gamma_1)} \\
    \int_{0}^\infty dt_2 \int_{\gamma_2 \in P_M^{t_2}(u,v)} \D\gamma_2\, e^{-S^{\mr{1q}}(\gamma_2)} 
    \int_{0}^\infty dt_3 \int_{\gamma_3 \in (P_{M_2})^{t_3}(v,y)} \D\gamma_3\, e^{-S^{\mr{1q}}(\gamma_1)}. \label{eq: intro decomp green's function}
\end{multline}  
Comparing with the gluing formula for the Green's function\footnote{See \cite[Proposition 4.2]{KMW} and Section \ref{subsubsec: gluing continuum} of the present paper for details.}
\begin{equation} G(x,y) = \int_{Y\times Y} d^{n-1}u\, d^{n-1}v\; E_{Y,M_1}(x,u)\, \varkappa_{Y,M}(u,v)\,  E_{Y,M_2}(v,y)\label{eq: gluing green intro}, 
\end{equation}
with $\varkappa_{Y,M} = (\DN_{Y,M_1}+\DN_{Y,M_2})^{-1}$ the inverse of the ``total'' Dirichlet-to-Neumann operator, suggests the following path integral formulae for the extension operator and $\varkappa$: 
\begin{align}
    E_{Y,M_i}(x,u) &= \int_{0}^\infty dt \int_{\gamma\in (P'_{M_i})^{t}(x,u)} \D\gamma\, e^{-S^{\mr{1q}}(\gamma)} ,\label{eq: intro path int E} \\ 
    \varkappa_{Y,M}(u,v) &= \int_{0}^\infty dt \int_{\gamma \in P_M^{t}(u,v)} \D\gamma \, e^{-S^{\mr{1q}}(\gamma)}. \label{eq: intro path int chi}
\end{align}
The results of our paper\footnote{Another reason to guess that formula is the fact that the integral kernel of the Dirichlet-to-Neumann operator is given by a symmetric normal derivative of the Green's function $\DN_{Y,M_i}(u,v) = -\partial_{n_u}\partial_{n_v}G_{M_i,Y}$ (in a regularized sense - see \cite[Remark 3.4]{KMW}), and formula \eqref{eq: intro path int E} for the first normal derivative of the Green's function.} actually suggest also the following path integral formula for the Dirichlet-to-Neumann operator: 
\begin{equation}
    \DN_{Y,M_i}(u,v)=
    \int_{0}^\infty dt\int_{\gamma \in (P''_{M_i})^t(u,v)}\D\gamma\,e^{-S^{\mr{1q}}(\gamma)}.
    \label{eq: intro path int D to N}
\end{equation}
Here $(P''_{M_i})^t(u,v)$ is the set of all paths from $\gamma \colon [0,t] \to M_i$ from $u \in Y$ to $v \in Y$ such that $\gamma(\tau) \notin Y$ for all $0 < \tau < t$. 

Assuming these formulae, we have again a ``first quantization formula'' for weights of Feynman graphs with boundary vertices
\begin{multline}\label{eq: relative Feynman weight 1q formalism}
    \Phi_\Gamma(\phi_Y) = \frac{\hbar^{|E|-|V|}}{|\mr{Aut}(\Gamma)|}\int_{0<t_1,\ldots,t_{|E|}<\infty} dt_1\cdots dt_{|E|}\\
    \int_{\ggamma\colon \Gamma_{t_1,\ldots,t_{|E|}}\ra M_i} \mc{D}\ggamma\,
    e^{-S^{1q}(\ggamma)} 
    \prod_{v\in V^\mr{bulk}} \left(-p_{\mr{val}(v)}\right)\prod_{v^\partial\in V^\partial}\phi_Y(\ggamma(v^\partial)).
\end{multline}
Here notation is as in \eqref{Feynman weight 1q formalism}, the only additional condition is that $\ggamma$ respects the type of edges in $\Gamma$, that is, for all $x \in \Gamma_{t_1,\ldots,t_{|E|}}$ we have $\ggamma(x) \in Y$ if and only if $x \in V^\partial$. 
\subsection{QFT on a graph. A guide to the paper.}
In this paper we study a toy (``combinatorial'' or ``lattice'') version of the scalar field theory (\ref{S (intro)}), where the spacetime manifold $M$ is replaced by a graph $X$, the scalar field $\phi$ is a function on the vertices of $X$ and the Laplacian in the kinetic operator is replaced by the graph Laplacian $\Delta_X$. 
I.e., the model is defined by the action
\begin{equation}
    S_X(\phi)=\sum_{v\in V_X} \frac12 \phi(v) \left((\Delta_X+m^2 \mr{Id}) \phi\right)(v)+p(\phi(v)),
\end{equation}
where $V_X$ is the set of vertices of $X$ and $p$ is the interaction potential (a polynomial of $\phi$), as before.

This model has the following properties. 
\begin{enumerate}[(i)]
    \item The ``functional integral'' over the space of fields is a finite-dimensional convergent integral (Section \ref{s: Scalar field theory on a graph}).
    \item The functional integral can be expanded in Feynman graphs, giving an asymptotic expansion of the nonperturbative partition function in powers of $\hbar$ (Section \ref{s: Interacting theory via Feynman diagrams}).
    \item Partition functions are compatible with unions of graphs over a subgraph (``gluing'') -- we see this as a graph counterpart of Atiyah-Segal functorial picture of QFT, with compatibility w.r.t. cutting-gluing $n$-manifolds along closed $(n-1)$-submanifolds. This functorial property of the graph QFT can be proven 
    \begin{enumerate}[(a)]
        \item 
    directly from the functional integral perspective (by a Fubini theorem argument) -- Section \ref{ss: functorial picture}, or 
    \item at the level of Feynman graphs (Section \ref{sec: cutting pert Z}). 
    \end{enumerate}
    The proof of functoriality at the level of Feynman graphs relies on the ``gluing formulae'' describing the behavior of Green's functions and determinants w.r.t. gluing of spacetime graphs (Section \ref{subsec:Gluing-in-Gaussian-theory}). These formulae are a combinatorial analog of known gluing formulae for Green's functions and zeta-regularized functional determinants on manifolds (Section \ref{ss: gluing in Gaussian theory: comparison to continuum}).
    \item The graph QFT admits a higher-categorical extension  which allows cutting-gluing along higher-codimension ``corners''\footnote{
    ``Corners'' in the graph setting are understood \`a la \v{C}ech complex, as multiple overlaps of ``bulk'' graphs.
    } (Section \ref{ss: extended graph QFT}), in the spirit of Baez-Dolan-Lurie extended TQFTs.
    \item The Green's function on a graph $X$ can be written as a sum over paths (Section \ref{s: path sum formulae for the propagator and determinant}, in particular Table \ref{tab: closed path formulas}), giving an analog of the formula (\ref{G 1st quant formula}); similarly, the determinant can be written as a sum over closed paths, giving an analog of (\ref{Feynman weight of S^1}).   This leads to
     a ``first-quantization'' representation of Feynman graphs, as a sum over maps $\Gamma\ra X$, sending vertices of $\Gamma$ to vertices of $X$ and sending edges of $\Gamma$ to paths on $X$ (connecting the images of the incident vertices) -- Section \ref{s: Interacting theory: first quantization formalism}. This yields a graph counterpart of the continuum first quantization formula (\ref{Feynman weight 1q formalism}).
     \item There are path sum formulae for the combinatorial extension (or ``Poisson'') operators and Dirichlet-to-Neumann operators (Section \ref{sec: path sums relative}, see in particular Table \ref{tab:my_label}), analogous to the path integral formulae \eqref{eq: intro path int E} and \eqref{eq: intro path int D to N}. 
     \item First quantization perspective gives a visual interpretation of the gluing formula for Green's functions and determinants on a graph $X=X'\cup_Y X''$ in terms of cutting the path into portions spent in $X'$ or in $X''$ (Section \ref{sec: gluing proof path sums}), and likewise an interpretation of the cutting-gluing of Feynman graphs (Section \ref{ss: 1st quantization cutting-gluing a Feynman graph}).
\end{enumerate}


\begin{remark}
A free (Gaussian) version of the combinatorial model we are studying in this paper was studied 
in \cite{RV}. Our twist on it is the deformation by a polynomial potential, the path-sum (first quantization) formalism, and the gluing formula for propagators (the BFK-like gluing formula for determinants was studied in \cite{RV}). 
\end{remark}

\subsection{Acknowledgements}
We thank Olga Chekeres, Andrey Losev and Donald R. Youmans for inspiring discussions on the first quantization approach in QFT. I.C., P.M. and K.W. would like to thank the Galileo Galilei Institute, where part of the work was completed, for hospitality.

\newpage
\pagestyle{fancyams}
\printunsrtglossary[type=main,title=Notations,nonumberlist,style=myNoHeaderStyle]

\section{Scalar field theory on a graph}\label{s: Scalar field theory on a graph}

Let $X$ be a finite graph. 
Consider the toy field theory on $X$ where fields are real-valued functions $\phi(v)$ on the set of vertices $V_X$, i.e., the space of fields is the space of 0-cochains on $X$ seen as a 1-dimensional CW complex,
$$F_X=C^0(X).$$
We define the action functional as 
\begin{equation}\label{S}
\begin{aligned}
    S_X(\phi)&=& \frac12( d\phi,d\phi ) +  \langle\mu,\frac{m^2}{2}\phi^2+ p(\phi)\rangle\\
    &=&\frac12 (\phi, (\Delta_X+m^2)\phi)+ \langle\mu, p(\phi)\rangle\\
   &=& \sum_{e\in E_X} \frac12 (\phi(v^1_e)-\phi(v^0_e))^2 + \sum_{v\in V_X} \left(\frac{m^2}{2}\phi(v)^2+ p(\phi(v))\right) .
\end{aligned}
\end{equation}
Here:
\begin{itemize}
    \item  $d\colon C^0(X) \ra C^1(X)$ is the cellular coboundary operator (we assume that $0$ cells carry $+$ orientation and $1$-cells carry some orientation -- the model does not depend on this choice).
    \item $(,)\colon \mr{Sym}^2 C^k(X)\ra \RR$ for $k=0,1$ is the standard metric, in which the cell basis is orthonormal.
    \item $\langle,\rangle$ is the canonical pairing of chains and cochains; $\mu$ is the 0-chain given by the sum of all vertices with coefficient $1$.\footnote{The 0-chain $\mu$ is an analog of the volume form on the spacetime manifold in our model. 
    If we want to consider the field theory on $X$ as a lattice approximation of a continuum field theory, we would need to scale the metric $(,)$ and the 0-chain $\mu$ appropriately with the mesh size. Additionally, one would need to add mesh-dependent counterterms to the action in order to have finite limits for the partition function and correlators.
    }
    \item $m>0$ is the fixed ``mass'' parameter.
    \item $p(\phi) 
    $ is a fixed polynomial
    (``potential''),
    \begin{equation}\label{p(phi)}
        p(\phi)=\sum_{k\geq 3} \frac{p_k}{k!}\phi^k.
    \end{equation}
    More generally, $p(\phi)$ can be a real analytic function.
    We will assume that $\frac{m^2}{2}\phi^2+ p(\phi)$ has a unique absolute minimum at $\phi=0$ and that it grows sufficiently fast\footnote{Namely, we want the integral $\int_\RR d\phi\, e^{-\frac{1}{\hbar}(\frac{m^2}{2}\phi^2+ p(\phi))}$ to converge for any $\hbar>0$.} at $\phi\ra \pm \infty$,
    so that the integral (\ref{Z}) converges measure-theoretically.
    \item $\Delta_X=d^T d\colon C^0(X)\ra C^0(X)$ is the graph Laplace operator on $0$-cochains; $d^T\colon C^1(X)\ra C^0(X)$ is the dual (transpose) map to the coboundary operator (in the construction of the dual, one identifies chains and cochains using the standard metric). The matrix elements of $\Delta_X$ in the cell basis, for $X$ a simple graph (i.e. without double edges and short loops), are
    $$ (\Delta_X)_{uv} =\left\{ 
    \begin{array}{ll}
          \mr{val}(v) & \mbox{if}\; u=v,  \\
         -1& \mbox{if  $u\neq v$ and $u$ is connected to $v$ by an edge}, \\
         0 & \mbox{otherwise},
    \end{array}
    \right. 
    $$
    where $\mr{val}(v)$ is the valence of the vertex $v$. 
    More generally, for $X$ not necessarily simple, one has 
    $$ 
    (\Delta_X)_{uv} =\left\{ 
    \begin{array}{ll}
          \mr{val}(v)-2\cdot \#\{\mbox{short loops $v\ra v$}\} & \mbox{if}\; u=v , \\
         - \#\{ \mbox{edges $u\ra v$}\}& \mbox{if  $u\neq v$}.
    \end{array}
    \right. 
    $$
\end{itemize}

We will be interested in the partition function
\begin{equation}\label{Z}
    Z_X=\int_{F_X} D\phi\; e^{-\frac{1}{\hbar}S_X} ,
\end{equation}
where $$D\phi=  
\prod_{v\in V_X} \frac{d\phi(v)}{\sqrt{2\pi\hbar}}$$ 
is the ``functional integral measure'' on the space of fields $F_X$ (in this case, just the Lebesgue measure on a finite-dimensional space); $\hbar>0$ is the parameter of quantization -- the ``Planck constant.''\footnote{Or one can think of $\hbar$ as ``temperature'' if one thinks of (\ref{Z}) as a partition function of statistical mechanics with $S$ the energy of a state $\phi$.}   
The integral in the r.h.s. of (\ref{Z}) is absolutely convergent.  
One can also consider correlation functions
\begin{equation}\label{correlator}
\langle \phi(v_1) \cdots \phi(v_n) \rangle = \frac{1}{Z_X}\int_{F_X} D\phi\; e^{-\frac{1}{\hbar}S_X} \phi(v_1)\cdots \phi(v_n).
\end{equation} 

\begin{remark}
    We stress that in this section we consider the \emph{nonperturbative} partition functions/correlators and $\hbar$ is to be understood as an actual positive number, unlike in the setting of perturbation theory (Section \ref{s: Interacting theory via Feynman diagrams}) where $\hbar$ becomes a formal parameter.
\end{remark}

\begin{remark}
    In this paper we use the Euclidean QFT convention for our (toy) functional integrals, with the integrand $e^{-\frac{1}{\hbar} S}$ instead of $e^{\frac{i}{\hbar}S}$, in order to have a better measure-theoretic convergence situation. The first convention leads to absolutely convergent integrals whereas the second leads to conditionally convergent oscillatory integrals.
\end{remark}

\subsection{Functorial picture}\label{ss: functorial picture}

One can interpret our model in the spirit of Atiyah-Segal functorial picture of QFT, 
as a (symmetric monoidal) functor 
\begin{equation}\label{Z as a functor}
\GCob \xrightarrow{(\HH,Z)} \mr{Hilb}  
\end{equation}
from the spacetime category\footnote{This terminology is taken from \cite{R}.} of graph cobordisms to the category of  
Hilbert spaces and Hilbert-Schmidt operators.

Here in the source category 
$\GCob$  is as follows:
\begin{itemize}
    \item The objects are graphs $Y$.
    \item A morphism from $Y_\ii$  to $Y_\oo$ is a graph $X$ which contains $Y_\ii$ and $Y_\oo$ as disjoint subgraphs. We will write $Y_\ii\xra{X}Y_\oo$ and refer to $Y_\ii,Y_\oo$ as ``ends'' (or ``boundaries'') of $X$, or we will say that $X$ is a ``graph cobordism'' between $Y_\ii$ and $Y_\oo$.
    \item The composition is given by unions of graphs with out-end of one cobordism identified with the in-end of the subsequent one:
    \begin{equation}\label{gluing of cobs}
    (Y_3\xla{X''}Y_2)\circ(Y_2\xla{X'}Y_1) = Y_3 \xla{X} Y_1, 
    \end{equation}
    where 
    $$X=X'\cup_{Y_2} X''.$$
    \item The monoidal structure is given by disjoint unions of graphs.
\end{itemize}
All graphs are assumed to be finite. As defined, $\GCob$ does not have unit morphisms (as usual for spacetime categories in non-topological QFTs); by abuse of language, we still call it a category.

The target category $\mr{Hilb}$ has as its objects  Hilbert spaces $\HH$ over $\CC$;\footnote{
 Alternatively (since we do not put $i$ in the exponent in the functional integral), one can consider Hilbert spaces over $\mathbb{R}$.
} the morphisms are Hilbert-Schmidt operators; 
the composition is composition of operators. The monoidal structure is given by tensor products (of Hilbert spaces and of operators).

The functor (\ref{Z as a functor}) is constructed as follows. For an end-graph $Y\in \mr{Ob}(\GCob)$, the associated vector space is 
\begin{equation}\label{H_Y}
\HH_Y=L^2(C^0(Y))
\end{equation}
-- the space of complex-valued   square-integrable 
functions on the vector space $C^0(Y)=\RR^{V_Y}$. 

For a graph cobordism $Y_\ii\xra{X} Y_\oo$, the associated operator 
$Z_X\colon \HH_{Y_\ii}\ra \HH_{Y_\oo}$ is
\begin{equation}\label{Z of a cob}
Z_X\colon \Psi_\ii \mapsto 
\Big(\Psi_\oo\colon \phi_\oo 
\mapsto \int_{F_{Y_\ii}}D\phi_\ii\;
\langle \phi_\oo | Z_X |\phi_\ii \rangle
\Psi_\ii(\phi_\ii)\Big) 
\end{equation} 
with the integral kernel
\begin{equation}\label{Z int kernel}
    \langle \phi_\oo | Z_X |\phi_\ii \rangle \colon=
\int_{F_X^{\phi_\ii,\phi_\oo}} [D\phi]^{\phi_\ii,\phi_\oo}  \; e^{-\frac{1}{\hbar}(S_X(\phi) -\frac12 S_{Y_\ii}(\phi_\ii)-\frac12 S_{Y_\oo}(\phi_\oo))} .
\end{equation}
Here 
\begin{itemize}
    \item $F_X^{\phi_\ii,\phi_\oo}$ is the space of fields on $X$ subject to boundary conditions $\phi_\ii,\phi_\oo$ imposed on the ends, i.e., it is the fiber of the evaluation-at-the-ends map 
$$F_X \ra F_{Y_\ii}\times F_{Y_\oo}$$
over the pair $(\phi_\ii,\phi_\oo)$. 
\item The measure 
$$[D\phi]^{\phi_\ii,\phi_\oo}=\prod_{v\in V_X\setminus (V_{Y_\ii}\sqcup V_{Y_\oo})} \frac{d\phi(v)}{\sqrt{2\pi\hbar}}$$ stands for the ``conditional functional measure'' on fields subject to boundary conditions. 
\end{itemize}

We will also call the expression (\ref{Z int kernel}) the partition function on the graph $X$ ``relative'' to the ends $Y_\ii$, $Y_\oo$, or just the partition function relative to the ``boundary'' subgraph $Y=Y_\ii\sqcup Y_\oo$, if the distinction between ``in'' and ``out'' is irrelevant. In the latter case we will use the notation $Z_{X,Y}(\phi_Y)$, with $\phi_Y=(\phi_\ii,\phi_\oo)$.

\begin{proposition}\label{prop: functoriality}
The assignment (\ref{H_Y}), (\ref{Z of a cob}) is a functor of monoidal categories.
\end{proposition}
\begin{proof}
The main point to check is that composition is mapped to composition. It follows from Fubini theorem, locality of the integration measure (that it is a product over vertices of local measures) and additivity of the action:
\begin{equation} \label{S additivity}
    S_{X}(\phi)=S_{X'}(\phi|_{X'})+S_{X''}(\phi|_{X''})-S_{Y_2}(\phi|_{Y_2})
\end{equation}
in the notations of (\ref{gluing of cobs}). Indeed, it suffices to prove
\begin{equation}\label{functoriality eq}
    \int_{F_{Y_2}} D\phi_2\;\langle \phi_3 | Z_{X''} |\phi_2 \rangle \langle \phi_2 |Z_{X'}|\phi_1\rangle  \stackrel{!}{=}  
    \langle \phi_3 | Z_X |\phi_1 \rangle
\end{equation}
-- again, we are considering the gluing of graph cobordisms as in (\ref{gluing of cobs}). The l.h.s. is
\begin{multline*}
    \int_{F_{Y_2}} D\phi_2 \int_{F_{X'}^{\phi_1,\phi_2}} [D\phi']^{\phi_1,\phi_2}  \int_{F_{X''}^{\phi_2,\phi_3}} [D\phi'']^{\phi_2,\phi_3} 
    \exp \Big( -\frac{1}{\hbar}\Big(\big(-\frac12 S_{Y_1}(\phi_1)+\\
    +\underbrace{S_{X'}(\phi')-\frac12 S_{Y_2}(\phi_2)\big)+\big(-\frac12 S_{Y_2}(\phi_2)+S_{X''}(\phi'')}_{S_X(\phi)}-\frac12 S_{Y_3}(\phi_3)\big) \Big) \Big)\\
    =\int_{F_X^{\phi_1,\phi_3}} [D\phi]^{\phi_1,\phi_3} e^{-\frac{1}{\hbar}(S_X(\phi)-\frac12 S_{Y_1}(\phi_1)-\frac12 S_{Y_2}(\phi_3))},
\end{multline*}
which proves (\ref{functoriality eq}). Here we understood that $\phi$ is a field on the glued cobordism $X$ restricting to $\phi',\phi''$ on $X',X''$, respectively. Compatibility with disjoint unions is obvious by construction.
\end{proof}

\begin{remark}
One can interpret the correlator (\ref{correlator}) as the partition function of $X$ seen as a cobordism $\{v_1,\ldots,v_n\}\xra{X} \varnothing$  applied to the state $\phi(v_1)\otimes\cdots \otimes \phi(v_n)\in \HH_{\{v_1,\ldots,v_n\}}$.
\end{remark}

\begin{remark}\label{rem: graph-continuum QFT dictionary}
   The combinatorial model we are presenting is intended to be an analog (toy model) of the continuum QFT, according to the dictionary of Table \ref{table 1}. 
\begin{center}
   \begin{table}[H]
    \bgroup
    \def\arraystretch{1.3}
    \begin{adjustbox}{max width=1.1\textwidth,center}
   \begin{tabular}{c|c}
        combinatorial QFT & continuum QFT  \\ \hline
        graph $X$ & closed spacetime $n$-manifold $M$; \\ 
          field $\phi\colon V_X\ra \RR$ & scalar field $\phi\in C^\infty(M)$; \\
        action (\ref{S}) & action $S(\phi)=\int_M \frac12 d\phi\wedge *d\phi +(\frac{m^2}{2}\phi^2+p(\phi))d\mr{vol}$ \\
        & $ = \int_M (\frac12 \phi (\Delta+m^2)\phi+p(\phi) ) d\mr{vol}$;\\
        partition function (\ref{Z}) & functional integral on a closed manifold;\\ \hline
        graph cobordism $Y_\ii\xra{X} Y_\oo$ &  $n$-manifold $M$ with in/out-boundaries \\ & being closed $(n-1)$-manifolds $\gamma_\ii,\gamma_\oo$;\\
        gluing/cutting of graph cobordisms 
        & gluing/cutting of smooth $n$-cobordisms;\\
        matrix element (\ref{Z int kernel}) & functional integral \\ & with boundary conditions $\phi_\ii,\phi_\oo$.
   \end{tabular}
   \end{adjustbox}
   \egroup
   \caption{Comparison between toy model and continuum QFT.}
   \label{table 1}
   \end{table}
\end{center}
\end{remark}

When we want to emphasize that a graph $X$ is not considered as a cobordism (or equivalently $X$ is seen as a cobordism $\varnothing \xra{X} \varnothing$), we will call $X$ a ``closed'' graph (by analogy with closed manifolds).


\subsection{Aside: ``QFT with corners'' (or ``extended QFT'') picture}
\label{ss: extended graph QFT}
Fix any $n\geq 1$. We will describe a (tautological) extension 
of the functorial picture above for our graph model
as an $n$-extended QFT (with gluing/cutting along ``corners'' of codimension up to $n$), in the spirit of Baez-Dolan-Lurie program \cite{Baez-Dolan}, \cite{Lurie} of extended topological quantum field theories.\footnote{\label{footnote: no cob hypothesis}
Scalar field theory is not a topological field theory, so, e.g., we should not expect (and in fact do not have) an analog of the cobordism hypothesis here: the theory is not recovered from its value on highest-codimension corner.} One has a functor of symmetric monoidal $n$-categories
\begin{equation}\label{extended QFT functor}
\GCob^n \xra{(\HH,Z)} \T^n .
\end{equation}
We proceed to describe its ingredients.

\subsubsection{Source $n$-category}
The source $n$-category $\GCob^n$ is as follows.
\begin{itemize}
    \item Objects (a.k.a. $0$-morphisms) are graphs $X^{[0]}$ (the index in brackets is to emphasize that this is a graph at categorical level $0$).
    \item A 1-morphism between objects (graphs) $Y_1^{[0]},Y_2^{[0]}$ is a graph $X^{[1]}$ together with graph embeddings of $Y_1^{[0]},Y_2^{[0]}$ into $X^{[1]}$ with disjoint images.
    \item For $2\leq k\leq n$, a $k$-morphism between two $(k-1)$-morphisms $Y_1^{[k-1]},Y_2^{[k-1]}\in \mr{Mor}_{k-1}(A_1^{[k-2]},A_2^{[k-2]})$ is a graph $X^{[k]}$ equipped with embeddings of $Y_{1,2}^{[k-1]}$ satisfying the following ``maximal disjointness'' assumption: 
    in the resulting
     diagram of graph embeddings
    $$
    \xymatrix{X^{[k]} &  \\
    Y_1^{[k-1]} \ar[u] & Y_2^{[k-1]} \ar[ul] \\
    A_1^{[k-2]} \ar[u] \ar[ur] & A_2^{[k-2]} \ar[ul] \ar[u]
    }
    $$
     the intersection of images of $Y_1^{[k-1]},Y_2^{[k-1]}$ in $X^{[k]}$ is the union of images of $A_1^{[k-2]}$ and $A_2^{[k-2]}$. 
\end{itemize}


\begin{example}
    Figure \ref{fig: 2 mor} below shows  an example of a 2-morphism $\vcenter{\hbox{\includegraphics[scale=0.4]{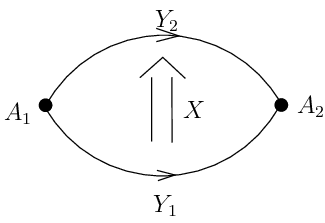}}}$ in $\GCob^2$:
\begin{figure}[H]
\includegraphics[scale=0.8]{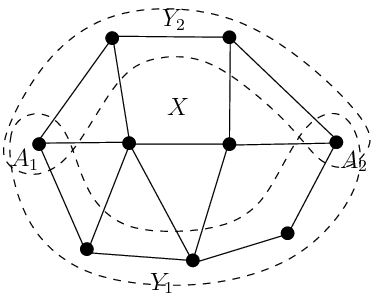}
\caption{An example of a 2-morphism in $\GCob^2$.}
\label{fig: 2 mor}
\end{figure}
\end{example}

The monoidal structure in $\GCob^n$ is given by disjoint unions at each level, and the composition is given by unions (pushouts) $X_1^{[k]}\cup_{Y^{[k-1]}} X_2^{[k]}$.

\begin{remark}
In $\GCob^n$ we only consider ``vertical'' compositions of $k$-morphisms, i.e., one can only glue two level $k$ graphs over a level $k-1$ graph, not over any graph at level $k'<k$ (otherwise, the composition would fail the maximal disjointness assumption). One might call this structure a ``partial'' $n$-category\footnote{
Or a ``Pickwickian $n$-category.'' (Cf. ``He had used the word in its Pickwickian sense\ldots He had merely considered him a humbug in a Pickwickian point of view.'' Ch. Dickens, Pickwick Papers.)
} (but by abuse of language we suppress ``partial''). On a related point, as in Section \ref{ss: functorial picture}, there are no unit $k$-morphisms.
\end{remark}

\subsubsection{Target $n$-category}
The target $n$-category $\T^n$ is as follows.
\begin{itemize}
    \item Objects are commutative unital algebras over $\CC$ (``CUAs'') $\HH^{[0]}$.
    \item For $1\leq k\leq n-1$, a $k$-morphism between CUAs $\HH_1^{[k-1]},\HH_2^{[k-1]}$ is a CUA $\HH^{[k]}$ equipped with injective morphisms of unital algebras of $\HH_1^{[k-1]},\HH_2^{[k-1]}$ into $\HH^{[k]}$.
    \item An $n$-morphism between CUAs $\HH_{1,2}^{[n-1]}$ is a linear map $Z\colon \HH_1^{[n-1]}\ra \HH_2^{[n-1]}$. This map is not required to be an algebra morphism.
\end{itemize}
The monoidal structure is given by tensor product at all levels. The composition of $n$-morphisms is the composition of linear maps. The composition of $k$-morphisms for $k<n$ is given by the balanced tensor product of algebras over a subalgebra,
$$ \HH_1^{[k]}\otimes_{\HH^{[k-1]}} \HH_2^{[k]}. $$

\subsubsection{The QFT functor} The functor (\ref{extended QFT functor}) is defined as follows.
\begin{itemize}
    \item For $0\leq k\leq n-1$, a graph $X^{[k]}\in \mr{Mor}_k(Y_1^{[k-1]},Y_2^{[k-1]})$ is mapped to the commutative unital algebra of functions on $0$-cochains on the graph 
    $$ \HH^{[k]}(X^{[k]}): =C^\infty(C^0(X^{[k]})), $$
    with algebra maps from functions on 0-cochains of $Y_1^{[k-1]},Y_2^{[k-1]}$ induced by graph inclusions $i_{1,2}\colon Y_{1,2}^{[k-1]}\ra X^{[k]}$ :
    \begin{equation}\label{i^**}
        (i_{1,2}^*)^*\colon C^\infty(C^0(Y_{1,2}^{[k-1]}))\ra C^\infty(C^0(X^{[k]})).
    \end{equation}
    Here $i_{1,2}^*$ is the restriction of a 0-cochain (field) from $X^{[k]}$ to $Y_{1,2}^{[k-1]}$, and $(i_{1,2}^*)^*$ is the pullback by this restriction map.
    \item A graph $X^{[n]}\in \mr{Mor}_n(Y_\ii^{[n-1]},Y_\oo^{[n-1]})$
    is mapped to the linear map $Z(X^{[n]})\colon \HH^{[n-1]}(Y_\ii^{[n-1]})\ra \HH^{[n-1]}(Y_\oo^{[n-1]})$  defined by a variant of (\ref{Z of a cob}), (\ref{Z int kernel}) allowing the in- and out-boundaries to intersect:
    \begin{multline} \label{Z with corners}
    Z(X)\colon \Psi_\ii \mapsto \Big(\Psi_\oo \colon \phi_\oo \mapsto \\
    \mapsto 
    \int_{F_X^{\phi_\oo} } [D \phi]^{\phi_\oo} e^{-\frac{1}{\hbar}(S_X(\phi)-\frac12 S_{ Y_\oo}(\phi_\oo)-\frac12 S_{Y_\ii}(\phi|_{Y_\ii}))} \Psi_\ii(\phi|_{Y_\ii})
    \Big).
    \end{multline}
    Here we suppressed the superscripts $[\cdots]$ for $X,Y 
    $ to lighten the notation. 
\end{itemize}
It is a straightforward check (by repeating  the argument of Proposition \ref{prop: functoriality}) that formula (\ref{Z with corners}) is compatible with gluing (vertical composition)  of $n$-morphisms 
in $\GCob^n$, see Figure \ref{fig: 2 mor comp}.
    \begin{figure}[H]
    \includegraphics[scale=0.7]{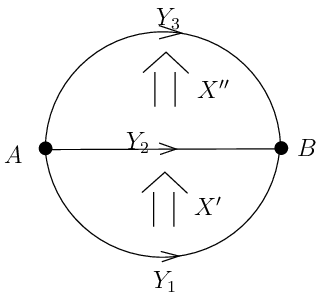} 
    \caption{Composition of 2-morphisms.}
    \label{fig: 2 mor comp}
    \end{figure}

\begin{remark}\label{rem: extended cob}
    The category $\GCob^n$ should be thought of as an (oversimplified) 
    toy model of Baez-Dolan-Lurie fully extended smooth cobordism  $(\infty,n)$-category, 
    where graphs $X^{[k]}$ model $k$-dimensional smooth strata.

    In this language, if we relabel our graphs by categorical co-level, $X^{\{k\}}\colon= X^{[n-k]}$, we should think of graphs $X^{\{0\}}$ as ``bulk,'' graphs $X^{\{1\}}$ as ``boundaries,'' graphs $X^{\{2\}}$ as ``codimension 2 corners,'' etc.
\end{remark}

\begin{remark}
We also remark that one can consider a different (simpler) version of the source category -- iterated cospans of graphs w.r.t. graph inclusions, without any disjointness conditions. While this $n$-category is simpler to define and admits non-vertical compositions, it has less resemblance to the  extended cobordism category, as here the intersection of strata of codimensions  $k$ and $l$ can have codimension less than $k+l$.
\end{remark}

\begin{remark}
Note that the most interesting part of the theory is concentrated in the top component of the functor (partition functions $Z$) -- e.g., the interaction potential $p(\phi)$ only affects it, not the spaces of states $\HH^{[k]}$. This is why we emphasize (cf. footnote \ref{footnote: no cob hypothesis}) that there is no analog of the cobordism hypothesis in our model: one cannot recover the entire functor (in particular, the top component) from its bottom component. 

This situation is similar to another example of an extended geometric (non-topological) QFT -- the 2D Yang-Mills theory. In this case, the area form affects the QFT 2-functor only at the top-dimension stratum (and thus ``obstructs'' the cobordism hypothesis), cf. \cite{IM}.
\end{remark}


\begin{remark}
The case $n=1$ of the formalism of this section is slightly (inconsequentially) different from the non-extended functorial picture of Section \ref{ss: functorial picture}, with target category $\mathsf{T}^1$ instead of $\mr{Hilb}$, with boundaries mapped to algebras of smooth functions on boundary values of fields rather than Hilbert spaces of $L^2$  functions of boundary values of fields. 

In the extended setting, we cannot use $L^2$ functions for two reasons: (a) they don't form an algebra and (b) the pullback (\ref{i^**}) of a function of field values on codim=2 corner vertices to a codim=1 boundary is generally not square integrable. 
(i.e. our QFT functor applied to the inclusion of a corner graph into the boundary graph does not land in the $L^2$ space).
\end{remark}

\section{Gaussian theory} 

\subsection{Gaussian theory on a closed graph}
Consider the free case of the model (\ref{S}), with the interaction $p(\phi)$ set to zero. The action is quadratic
\begin{equation}\label{S Gaussian}
S_X(\phi)=\frac12 ( \phi, K_X \phi ),
\end{equation}
where the kinetic operator is
$$ K_X\colon= \Delta_X+m^2 $$
-- it is a positive self-adjoint operator on $F_X$.
Let us denote its inverse
$$G_X\colon = (K_X)^{-1}$$
-- the ``Green's function'' or ``propagator;'' we will denote matrix elements of $G_X$ in the basis of vertices by $G_X(u,v)$, for $u,v\in V_X$.

The partition function (\ref{Z}) for a closed graph $X$ is the Gaussian integral
\begin{equation}\label{Gaussian Z}
Z_X=\int_{F_X} D\phi\; e^{-\frac{1}{\hbar}S_X(\phi)} = \det(K_X)^{-\frac12}.
\end{equation} 
The correlator (\ref{correlator}) is given by Wick's lemma, as a moment  of the Gaussian measure:
$$ 
\langle \phi(v_1)\cdots \phi(v_{2m}) \rangle = \hbar^m \sum_{\mr{partitions 
}\; \{1,\ldots,2m\}=\cup_{i=1}^m \{a_i,b_i\}} G_X(v_{a_1},v_{b_1})\cdots G_X(v_{a_m},v_{b_m}).
$$

\subsubsection{Examples}
\begin{example}\label{ex: path graph N=3}
Consider the  graph $X$ shown in Figure \ref{fig: line 3} below:
\begin{figure}[H]
    \centering
    \includegraphics{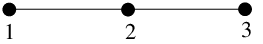}
    \caption{A line graph on 3 vertices.}
    \label{fig: line 3}
\end{figure}

The kinetic operator is
$$
K_X=
\left(
\begin{array}{ccc}
    1+m^2 & -1 & 0  \\
     -1& 2+m^2 & -1 \\
     0 &-1 & 1+m^2
\end{array}
\right).
$$
Its determinant is:
\begin{equation} \label{ex path3 det}
\det K_X=m^2(1+m^2)(3+m^2)
\end{equation}
and the inverse is
\begin{equation} \label{ex path3 G_X}
    G_X=\frac{1}{m^2(1+m^2)(3+m^2)}
\left(
\begin{array}{ccc}
    1+3m^2+m^4 & 1+m^2 & 1 \\
     1+m^2& (1+m^2)^2 & 1+m^2 \\
     1 &1+m^2 & 1+3m^2 + m^4
\end{array}
\right).
\end{equation}

\end{example}

\begin{example}\label{ex: path graph}
    Consider the line graph of length $N$:
    \begin{figure}[H]
    \centering
\includegraphics[scale=0.75]{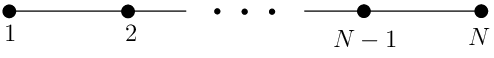}
    \caption{A line graph of length $N$.}
    \label{fig: path graph}
    \end{figure} 
    
The kinetic operator is the tridiagonal matrix
$$ 
K_X=\left(
\begin{array}{ccccc}
     1+m^2& -1 && & \\
     -1& 2+m^2 &-1 & & \\
     &-1 & \ddots & \ddots & \\
     &&\ddots& 2+m^2 & -1 \\
     &&& -1 & 1+m^2
\end{array}
\right).
$$
The matrix elements of its inverse are:\footnote{\label{footnote: path graph -- Neumann bc}
One finds this by solving the finite difference equation $-G(i+1,j)+(2+m^2)G(i,j)-G(i-1,j)=\delta_{ij}$, using the ansatz $G(i,j)=A_+ e^{\beta i} + A_- e^{-\beta i} $ for $i\leq j$ and $G(i,j)=B_+ e^{\beta i} + B_- e^{-\beta i} $ for $i\geq j$, with $A_\pm, B_\pm$ some coefficients depending on $j$. One imposes single-valuedness (``continuity'') at $i=j$ and ``Neumann boundary conditions'' 
$G(0,j)=G(1,j)$, $G(N,j)=G(N+1,j)$,
which -- together with the original equation at $i=j$ -- determines uniquely the solution. 
One can obtain the determinant from the propagator using the property
$\frac{d}{dm^2}\log\det K_X=\operatorname{tr} K_X^{-1}\frac{d}{dm^2} K_X=\sum_{i=1}^N G_X(i,i)$.
}
\begin{equation}\label{path graph G}
G_X(i,j) 
=\frac{\cosh\beta(N-|i-j|)+\cosh\beta (N+1-i-j)}{2\sinh \beta\, \sinh \beta N}
,\quad 1\leq i,j\leq N,
\end{equation}
where 
$\beta$ is related to $m$ by
\begin{equation}\label{q in path graph}
\sinh \frac{\beta}{2}=\frac{m}{2}.
\end{equation}
The determinant is:
\begin{equation} \label{N-N graph det}
\det K_X
=2\tanh \frac{\beta}{2}\, \sinh \beta N.
\end{equation}
\end{example}

\begin{example}\label{ex: circle graph}
    Consider the circle graph with $N$ vertices shown in Figure \ref{fig: Circle N} below:  
\begin{figure}[H]
    \centering
    \includegraphics[scale=0.75]{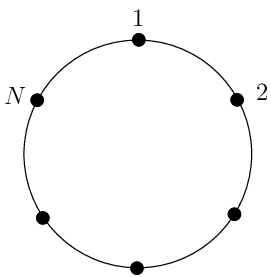}
    \caption{A circle graph with $N$ vertices.}
    \label{fig: Circle N}
\end{figure}    The kinetic operator is:
    $$ 
K_X=\left(
\begin{array}{ccccc}
     2+m^2& -1 && & -1 \\
     -1& 2+m^2 &-1 & & \\
     &-1 & \ddots & \ddots & \\
     &&\ddots& 2+m^2 & -1 \\
     -1 &&& -1 & 2+m^2
\end{array}
\right).
$$
(We are only writing the nonzero entries.) Its inverse is given by 
\begin{equation}\label{G circle graph}
G_X(i,j)  
=\frac{\cosh \beta (\frac{N}{2}-|i-j|)}{2 \sinh \beta\, \sinh \frac{\beta N}{2}}
,\quad 1\leq i,j\leq N. 
\end{equation}
Here 
$\beta$ is as in (\ref{q in path graph}). The determinant is:
\begin{equation}\label{det circle graph}
\det K_X 
= 4 \sinh^2 \frac{\beta N}{2}.
\end{equation}
For instance, for $N=3$ we obtain
\begin{equation}\label{eq: G circle graph N=3}
         G_X = \frac{1}{m^2(m^2+3)}
   \begin{pmatrix} 
    m^2+1 & 1 & 1 \\ 
    1 & m^2 +1 & 1 \\ 
    1 & 1 & m^2 +1 
    \end{pmatrix}
\end{equation}
and 
\begin{equation}
    \det K_X = m^2(m^2+3)^2. \label{eq: det circle graph N=3}
\end{equation}
\end{example}


\subsection{Gaussian theory 
relative to the boundary
} \label{subsec:Gaussian-theory-relative-boundary}
Consider the Gaussian theory on a graph $X$ with ``boundary subgraph'' $Y\subset X$. 

\subsubsection{Dirichlet problem}
Consider the following ``Dirichlet problem.'' For a fixed field configuration on the boundary $\phi_Y\in F_Y$, we are looking for a field configuration on $X$, $\phi \in F_X$ such that
\begin{eqnarray} 
\phi|_Y&=& \phi_Y, \label{Dirichlet eq1}
\\ \label{Dirichlet eq2}
(K_X \phi)(v)&=&0 \;\;\mr{for\;all}\; v\in V_X\setminus V_Y.
\end{eqnarray}
Equivalently, we are minimizing the action (\ref{S Gaussian}) on the fiber $F_X^{\phi_Y}$ of the evaluation-on-$Y$ map $F_X \ra F_Y$ over $\phi_Y$. 
The solution exists and is unique due to convexity and nonnegativity of $S_X$.

Let us write the inverse of $K_X$ as a $2\times 2$ block matrix according to partition of vertices of $X$ into (1) not belonging to $Y$ (``bulk vertices'') or (2) belonging to $Y$ (``boundary vertices''):
\begin{equation}\label{K_X inverse as a block matrix}
    (K_X)^{-1} = \left( 
    \begin{array}{c|c}
         A& B \\ \hline
         C& D
    \end{array}
    \right).
\end{equation}
Note that this matrix is symmetric, so $A$ and $D$ are symmetric and $C=B^T$.

Then, we can write the solution of the Dirichlet problem as follows: (\ref{Dirichlet eq2}) implies
$ K_X\phi = \left(
\begin{array}{c}
     0  \\ 
      \xi
\end{array}
\right)
$ for some $\xi\in F_Y$. Hence, 
$$\phi= (K_X)^{-1}
\left(
\begin{array}{c}
     0  \\ 
      \xi
\end{array}
\right) =
\left(
\begin{array}{c}
     B\xi  \\ 
      D\xi
\end{array}
\right) \underset{(\ref{Dirichlet eq1})}{=} \left(
\begin{array}{c}
     0  \\ 
      \phi_Y
\end{array}
\right).
$$
Therefore, $\xi=D^{-1}\phi_Y$ and the solution of the Dirichlet problem is
\begin{equation}\label{Dirichlet solution}
   \phi= \left(
\begin{array}{c}
     B D^{-1} \phi_Y  \\ 
      \phi_Y
\end{array}
\right).
\end{equation}

\subsubsection{Dirichlet-to-Neumann operator}\label{sec:DtoN}
Note also that the evaluation of the action $S_X$ on the solution of the Dirichlet problem is
\begin{multline}\label{S on solution of Dirichlet problem}
    S_X(\phi)=\frac12 (\phi, K_X\phi) =
    \frac12 \Big(  
    \left(
\begin{array}{c}
     B D^{-1} \phi_Y  \\ 
      \phi_Y
\end{array}
\right)
    ,
    \left(
\begin{array}{c}
     0  \\ 
      \xi
\end{array}
\right)
    \Big) \\
    = \frac12 (\phi_Y,\xi)=\frac12 (\phi_Y, D^{-1} \phi_Y).
\end{multline}

The map sending $\phi_Y$ to the corresponding $\xi$ (i.e. the kinetic operator evaluated on the solution of the Dirichlet problem) is a combinatorial analog of the Dirichlet-to-Neumann operator.\footnote{\label{footnote: DN operator in continuum}
Recall that in the continuum setting, for $X$ a manifold with boundary, the Dirichlet-to-Neumann operator $\DN\colon C^\infty(\dd X)\ra C^\infty(\dd X)$ maps a smooth function $\phi_\dd$ to the normal derivative $\dd_n \phi(x)$ on $\dd X$ of the solution $\phi$ of the Helmholtz equation $(\Delta+m^2)\phi=0$ subject to Dirichlet boundary condition $\phi|_\dd = \phi_\dd$.} 
We will call the operator $\DN_{Y,X}\colon=D^{-1}\colon F_Y\ra F_Y$ the (combinatorial) Dirichlet-to-Neumann operator.\footnote{We put the subscripts in $\DN_{Y,X}$ to emphasize that we are 
extending $\phi_Y$ into $X$ as a solution of (\ref{Dirichlet eq2}).
When we will discuss gluing, the same $Y$ can be a subgraph of two different graphs $X',X''$; then it is important into which graph we are extending $\phi_Y$.} 

Recall (see e.g. \cite{KMW}) that in the continuum setting, the action of the free massive scalar field on a manifold with boundary, evaluated on a classical solution with Dirichlet boundary condition $\phi_\dd$ is $\int_{\dd X}\frac12\phi_\dd \DN(\phi_\dd) $. Comparing with (\ref{S on solution of Dirichlet problem})  
reinforces the idea that
is reasonable to call $D^{-1}$ the Dirichlet-to-Neumann operator.

We will denote the operator $BD^{-1}$ appearing in (\ref{Dirichlet solution}) by
\begin{equation}\label{E}
    E_{Y,X}=BD^{-1}
\end{equation}
-- the ``extension'' operator (extending $\phi_Y$ into the bulk of $X$ as a solution of the Dirichlet problem).\footnote{
In \cite{RV}, this operator is called the \emph{Poisson operator}.
}
\subsubsection{Partition function and correlators (relative to a boundary subgraph)}\label{sec:rel part funct}
Let us introduce a notation for the blocks of the matrix $K_X$ corresponding to splitting of the vertices of $X$ into bulk and boundary vertices, similarly to (\ref{K_X inverse as a block matrix}):
\begin{equation}\label{K_X  as a block matrix}
    K_X = \left( 
    \begin{array}{c|c}
         \wh{A}=K_{X,Y}& \wh{B} \\ \hline
         \wh{C}& \wh{D}
    \end{array}
    \right).
\end{equation}

The partition function relative to $Y$ (cf. (\ref{Z int kernel})) is again given by a Gaussian integral
\begin{multline}\label{Gaussian Z rel to Y}
    Z_{X,Y}(\phi_Y)= \int_{F_X^{\phi_Y}} [D\phi]^{\phi_Y} e^{-\frac{1}{\hbar}(S_X(\phi)-\frac12 S_Y(\phi_Y))}\\
    =\det(K_{X,Y} 
    )^{-\frac12} e^{-\frac{1}{2\hbar}(\phi_Y, (\DN_{Y,X}-\frac12 K_Y)\phi_Y)}.
\end{multline}

The normalized correlators (depending on the boundary field $\phi_Y$) are as follows.
\begin{itemize}
    \item 1-point correlator:\footnote{When specifying that a vertex $v$ is in $V_X\setminus V_Y$ we will use a shorthand and write $v\in X\setminus Y$.}
    \begin{equation}\label{Gaussian 1-point corr}
        \langle \phi(v) \rangle_{\phi_Y} = (E_{Y,X} \phi_Y)(v),\quad v\in X\setminus Y.
    \end{equation}
    \item Centered 
    $2m$-point correlator:
    \begin{multline}\label{Gaussian centered correlator}
        \langle \delta \phi(v_1)\cdots \delta \phi(v_{2m})  \rangle_{\phi_Y}=\\=
        \hbar^m\sum_{\mr{partitions }\; \{1,\ldots,2m\}=\cup_{i=1}^m \{a_i,b_i\}} G_{X,Y}(v_{a_1},v_{b_1})\cdots G_{X,Y}(v_{a_m},v_{b_m}),\\ v_1,\ldots,v_{2m}\in X\setminus Y.
    \end{multline}
    Here: 
    \begin{itemize}
        \item 
    $\delta \phi(v)\colon = \phi(v)- \langle \phi(v) \rangle_{\phi_Y}= \phi(v)-(E_{Y,X} \phi_Y)(v)$ is the fluctuation of the field w.r.t. its average; 
    \item $G_{X,Y}\colon=(K_{X,Y})^{-1}$ is the ``propagator with Dirichlet boundary condition on $Y$'' (or ``propagator relative to $Y$'').
    \end{itemize}
    \item Non-centered correlators follow from (\ref{Gaussian centered correlator}), e.g.
    \begin{equation}\label{Gaussian non-centered 2-point corr}
        \langle \phi(v_1) \phi(v_2) \rangle_{\phi_Y}  = \hbar\, G_{X,Y}(v_1,v_2)+(E_{Y,X}\phi_Y)(v_1)\cdot (E_{Y,X}\phi_Y)(v_2).
    \end{equation}
\end{itemize}

\begin{remark}
In our notations, the subscript $X,Y$ (as in $K_{X,Y}$, $G_{X,Y}$, $Z_{X,Y}$) stands for an object on $X$ relative to $Y$.\footnote{I.e. we think of $(X,Y)$ as a \emph{pair} of 1-dimensional CW complexes, where ``pair'' has the same meaning as in, e.g., the long exact sequence in cohomology of a pair.} On the other hand, the subscript $Y,X$ (as in $\DN_{Y,X}$, $E_{Y,X}$) refers to an object related to extending a field on $Y$ to a classical solution in the ``bulk'' $X$.
\end{remark}

\subsubsection{Examples}
\begin{example}\label{ex: path 2 graph rel bdry}
    Consider the  graph $X$ shown in Figure \ref{fig: line 2} below, relative the subgraph $Y$ consisting solely of the vertex $2$.
    \begin{figure}[H]
    \includegraphics[scale=0.8]{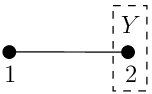}
    \caption{A graph with two vertices relative to one vertex.}
    \label{fig: line 2}
    \end{figure}
     The full kinetic operator is
    $$ 
   K_X= \left(
    \begin{array}{c|c}
        1+m^2 & -1 \\ \hline
        -1 & 1+m^2
    \end{array}
    \right)
    $$
    and the relative version is its top left block, $K_{X,Y}=1+m^2$. The relative propagator is 
    $G_{X,Y}=K_{X,Y}^{-1}=\frac{1}{1+m^2}$.
    The inverse of the full kinetic operator is
     $$ 
   K_X^{-1}= \frac{1}{m^2(2+m^2)}\left(
    \begin{array}{c|c}
        1+m^2 & 1 \\ \hline
        1 & 1+m^2
    \end{array}
    \right).
    $$
    The DN operator is the inverse of the bottom right block: $\DN_{Y,X}=\frac{m^2(2+m^2)}{1+m^2}$ and the extension operator (\ref{E}) is
    $E_{Y,X}=\frac{1}{1+m^2}$.

    In particular, the relative partition function is
    $$ 
    Z_{X,Y}(\phi_Y)=(1+m^2)^{-\frac12} \, e^{-\frac{1}{2\hbar}\left(\frac{m^2(2+m^2)}{1+m^2}-\frac{m^2}{2}\right)\phi_Y^2}.
    $$
\end{example}

\begin{example}\label{ex: path graph rel one end}
    Consider the line graph 
    of length $N$  relative to the subgraph consisting of the right endpoint $Y=\{N\}$ (Figure \ref{fig: line N rel 1}). 
\begin{figure}[H]
    \centering
    \includegraphics[scale=0.75]{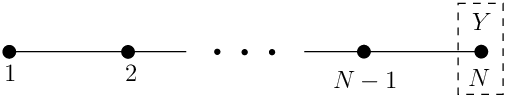}
    \caption{A line graph relative to one endpoint.}
    \label{fig: line N rel 1}
\end{figure}   The relative propagator is
    $$
G_{X,Y}(i,j)=\frac{\sinh\beta (N-\frac12-|i-j|)+\sinh\beta(N+\frac12-i-j)}{2 \sinh \beta\, \cosh \beta (N-\frac12)},\quad 1\leq i,j\leq N-1,
    $$
    with $\beta$ as in (\ref{q in path graph}).
    The DN operator is the inverse of the $N-N$ block (element) of the absolute propagator (\ref{path graph G}):
    $$
    \DN_{Y,X}=\frac{2\sinh \frac{\beta}{2}\, \sinh\beta N}{\cosh \beta (N-\frac12)}.
    $$
    The extension operator is
    $$ E_{Y,X}(i,N)= \frac{\cosh \beta (i-\frac12)}{\cosh \beta (N-\frac12)},\quad 1\leq i\leq N-1 $$
    and the determinant is
    \begin{equation} \label{N-D graph determinant}
    \det K_{X,Y}= \frac{\cosh \beta (N-\frac12)}{\cosh \frac{\beta}{2}}.
    \end{equation}
\end{example}

\begin{example}\label{ex: path graph rel both ends}
    Consider again the line graph, but now relative to both left and right endpoints, see Figure \ref{fig: line N rel 2} below.
    \begin{figure}[H]
    \centering
    \includegraphics[scale = 0.75]{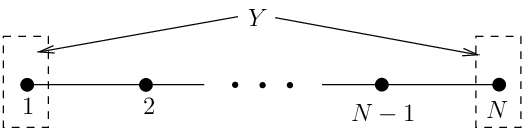}
    \caption{A line graph relative to both endpoints.}
    \label{fig: line N rel 2}
\end{figure}    

Then we have:
    \begin{gather} \label{G for path graph rel both ends}
        G_{X,Y}(i,j)=\frac{\cosh\beta (N-1-|i-j|)-\cosh\beta(N+1-i-j)}{2\sinh \beta\, \sinh \beta(N-1)},\;\; 2\leq i,j\leq N-1,\\
        \label{DN for path graph rel both ends}
        \DN_{Y,X}=\frac{2\sinh \frac{\beta}{2}}{\sinh \beta (N-1)} \left(
        \begin{array}{cc}
            \cosh \beta (N-\frac12) &  -\cosh \frac{\beta}{2} \\ 
            -\cosh \frac{\beta}{2} &  \cosh \beta (N-\frac12)
        \end{array}
        \right), \\
        \label{E for path graph rel both ends} E_{Y,X}(i,1)= \frac{\sinh \beta(N-i)}{\sinh\beta(N-1)} ,\;\;
        E_{Y,X}(i,N)= \frac{\sinh \beta (i-1)}{\sinh\beta(N-1)}
        , \\
        \label{det for path graph rel both ends}
        \det K_{X,Y} = \frac{\sinh \beta (N-1)}{\sinh \beta}.
    \end{gather}
\end{example}

\subsection{Gluing in Gaussian theory. 
Gluing of propagators and determinants.} \label{subsec:Gluing-in-Gaussian-theory}

\subsubsection{Cutting a closed graph} \label{sss Gaussian theory, cutting  closed graph}
Consider a closed 
graph $X=X'\cup_Y X''$ obtained from graphs $X',X''$ by gluing 
along a common subgraph $X'\supset Y \subset X''$. 

\begin{thm}\label{thm: gluing prop and det}
\begin{enumerate}[(a)]
    \item The propagator on $X$ is expressed in terms the data (propagators, DN operators, extension operators) on $X',X''$ relative to $Y$ as follows.
    \begin{itemize}
        \item For both vertices $v_1,v_2 \in X'$:
    \begin{multline}\label{propagator gluing eq1}
        G_X(v_1,v_2)=
     G_{X',Y}(v_1,v_2)+\\
     +\sum_{u_1,u_2\in Y}E_{Y,X'}(v_1,u_1)\DN^{-1}_{Y,X}(u_1,u_2)E_{Y,X'}(v_2,u_2) .
    \end{multline}
    For both vertices in $X''$, the formula is similar. Here the total DN operator is
    \begin{equation}\label{DN total}
        \DN_{Y,X}=\DN_{Y,X'}+\DN_{Y,X''}-K_Y.
    \end{equation}
    Also, we assume by convention that $G_{X',Y}(v_1,v_2)=0$ if either of $v_1,v_2$ is in $Y$. We also set $E_{Y,X'}(u,v)=\delta_{u,v}$ if $u,v\in Y$.
    \item For $v_1\in X'$, $v_2\in X''$,
    \begin{equation}\label{propagator gluing eq2}
        G_X(v_1,v_2)=
        \sum_{u_1,u_2\in Y}E_{Y,X'}(v_1,u_1 )\DN^{-1}_{Y,X}(u_1,u_2)E_{Y,X''}(v_2,u_2) 
    \end{equation} 
    and similarly for $v_1\in X''$, $v_2\in X'$. 
    \end{itemize}
    \item The determinant of $K_X$ is expressed in terms of the data on $X',X''$ relative to $Y$ as follows:
    \begin{equation}\label{determinant gluing}
        \det K_{X}=\det (K_{X',Y}) \det (K_{X'',Y}) \det (\DN_{Y,X}).
    \end{equation}
\end{enumerate}
\end{thm}

We will give three proofs of these gluing formulae:
\begin{enumerate}[(1)]
    \item From Fubini theorem for the ``functional integral'' (QFT/second quantization approach).
    \item From inverting a $2\times 2$ block matrix via Schur complement and Schur's determinant formula.
    \item From path counting (first quantization approach) -- later, in Section \ref{sec: gluing proof path sums}. 
\end{enumerate}

\subsubsection{Proof 1 (``functional integral approach'')}

First, consider the partition function on $X$ relative to $Y$:
\begin{multline}\label{Gaussin gluing proof 1 eq1}
    Z_{X,Y}(\phi_Y)=\int_{F_X^{\phi_Y}} [D\phi]^{\phi_Y} e^{-\frac{1}{\hbar} (S_X(\phi)-\frac12 S_Y(\phi_Y))} \underset{(\ref{S additivity}) }{=} \\=
    \int_{F_{X'}^{\phi_Y}} [D\phi']^{\phi_Y} \int_{F_{X''}^{\phi_Y}} [D\phi'']^{\phi_Y} e^{-\frac{1}{\hbar}(S_{X'}(\phi')+S_{X''}(\phi'')-\frac32 S_Y(\phi_Y))} \\
    = (\det K_{X',Y})^{-\frac12}(\det K_{X'',Y})^{-\frac12}e^{-\frac{1}{2\hbar} (\phi_Y, (\DN_{Y,X'}+\DN_{Y,X''}-\frac32 K_Y)\phi_Y)}.
\end{multline}
Comparing the r.h.s. with (\ref{Gaussian Z rel to Y}) as functions of $\hbar$, we obtain the formula (\ref{DN total}) for the total DN operator and the relation for determinants
\begin{equation*}
    \det K_{X,Y} = \det K_{X',Y}\cdot \det K_{X'',Y}.
\end{equation*}

The partition function on $X$ can be obtained by integrating (\ref{Gaussin gluing proof 1 eq1}) over the field on the ``gluing interface'' $Y$:
\begin{multline*}
    Z_X= \int_{F_X} D\phi\, e^{-\frac{1}{\hbar}S_X(\phi)}= \int_{F_Y} D\phi_Y\, e^{\frac{1}{2\hbar}S_Y(\phi_Y)} Z_{X,Y}(\phi_Y) =\\
    =\int_{F_Y} D\phi_Y (\det K_{X,Y})^{-\frac12} e^{-\frac{1}{2\hbar}(\phi_Y,\DN_{Y,X}\phi_Y)}=(\det K_{X,Y})^{-\frac12} (\det \DN_{Y,X})^{-\frac12}.
\end{multline*}
Comparing the r.h.s. with (\ref{Gaussian Z}), we obtain the gluing formula for determinants (\ref{determinant gluing}).

Next, we prove the gluing formula for propagators thinking of them as 2-point correlation functions. We denote by $\ll\cdots\gg$ correlators not normalized by the partition function.   Consider the case $v_1,v_2\in X'$. We have
\begin{multline*}
    \underbrace{\ll \phi(v_1)\phi(v_2) \gg}_{\hbar\, G_X(v_1,v_2)\cdot Z_X} =\int_{F_X} D\phi\, e^{-\frac{1}{\hbar}S_X(\phi)}\phi(v_1)\phi(v_2) =\\=
    \int_{F_Y} D\phi_Y \underbrace{\int_{F_{X'}^{\phi_Y}} [D\phi']^{\phi_Y} \phi'(v_1) \phi'(v_2) e^{-\frac{1}{\hbar}(S_{X'}(\phi')-\frac12 S_Y(\phi_Y))} }_{ \ll \phi(v_1)\phi(v_2)\gg_{\phi_Y}^{X'} \underset{(\ref{Gaussian non-centered 2-point corr}  )}{=}Z_{X',Y} (\phi_Y)\cdot (\hbar\, G_{X',Y}(v_1,v_2)+(E_{Y,X'}\phi_Y)(v_1)\cdot (E_{Y,X'}\phi_Y)(v_2))} \cdot \\
    \cdot \underbrace{\int_{F_{X''}^{\phi_Y}}  [D\phi'']^{\phi_Y} e^{-\frac{1}{\hbar}(S_{X''}(\phi'')-\frac12 S_Y(\phi_Y))}}_{Z_{X'',Y}(\phi_Y)} \\
    = \int_{F_Y} D\phi_Y (\det K_{X',Y})^{-\frac12} (\det K_{X'',Y})^{-\frac12} e^{-\frac{1}{2\hbar}(\phi_Y,\DN_{Y,X}\phi_Y) } \cdot \\
    \cdot 
    (\hbar\, G_{X',Y}(v_1,v_2)+\sum_{u_1,u_2\in Y} E_{Y,X'}(v_1,u_1) \phi_Y(u_1) \phi_Y(u_2) E_{Y,X'}(v_2,u_2)) \\
    = Z_X\cdot \hbar (G_{X',Y}(v_1,v_2)+\sum_{u_1,u_2\in Y}  E_{Y,X'}(v_1,u_1) \DN^{-1}_{Y,X}(u_1,u_2) E_{Y,X'}(v_2,u_2)).
\end{multline*}
This proves the gluing formula (\ref{propagator gluing eq1}).

Finally, consider the case $v_1\in X'$, $v_2\in X''$. By a similar computation we find
\begin{multline*}
    \underbrace{\ll \phi(v_1) \phi(v_2)  \gg}_{\hbar\, G_X(v_1,v_2)\cdot Z_X} = 
    \int_{F_Y} D\phi_Y \underbrace{\int_{F_{X'}^{\phi_Y}} [D\phi']^{\phi_Y} \phi'(v_1) e^{-\frac{1}{\hbar}(S_{X'}(\phi')-\frac12 S_Y(\phi_Y))} }_{ 
    =\ll \phi(v_1) \gg_{\phi_Y}^{X'}
    \underset{(    \ref{Gaussian 1-point corr}         )}{=}Z_{X',Y}(\phi_Y)\cdot (E_{Y,X'}\phi_Y)(v_1)}\cdot \\ \cdot
    \underbrace{\int_{F_{X''}^{\phi_Y}} [D\phi'']^{\phi_Y} \phi''(v_2) e^{-\frac{1}{\hbar}(S_{X''}(\phi'')-\frac12 S_Y(\phi_Y))}}_{
    = \ll \phi(v_2)\gg^{X''}_{\phi_Y}= Z_{X'',Y}(\phi_Y)\cdot (E_{Y,X''}\phi_Y)(v_2)
    } \\
    = \int_{F_Y} D\phi_Y (\det K_{X',Y})^{-\frac12} (\det K_{X'',Y})^{-\frac12} e^{-\frac{1}{2\hbar}(\phi_Y,\DN_{Y,X}\phi_Y) } \cdot \\ \cdot 
    \sum_{u_1,u_2\in Y} E_{Y,X'}(v_1,u_1) \phi_Y(u_1) \phi_Y(u_2) E_{Y,X''}(v_2,u_2)  \\
    =Z_X \cdot \hbar\sum_{u_1,u_2\in Y} E_{Y,X'}(v_1,u_1) \DN^{-1}_{Y,X}(u_1,u_2) E_{Y,X''}(v_2,u_2).
\end{multline*}
This proves (\ref{propagator gluing eq2}).

\subsubsection{Proof 2 (Schur complement approach)}

Let us introduce the notations
\begin{equation}\label{olG, olE}
    \ol{G}_{X,Y} =
    \left( \begin{array}{c|c}
        G_{X,Y} &  0 \\ \hline
         0& 0 
    \end{array}  
    \right)
    ,\quad
    \ol{E}_{Y,X} = 
    \left(
    \begin{array}{c}
         E_{Y,X}  \\
         \mr{id} 
    \end{array}
    \right)
\end{equation}
for the extension of the propagator on $X$ relative to $Y$ by zero to vertices of $Y$ and the extension of the extension operator by identity to vertices of $Y$ (the blocks correspond to vertices of $X\setminus Y$ and vertices of $Y$, respectively).\footnote{
Note that one can further refine the block decompositions (\ref{olG, olE}) according to partitioning of vertices in $X\setminus Y$ into those in $X'\setminus Y$ and those in $X''\setminus Y$. Then the block $G_{X,Y}$ becomes $\left(
\begin{array}{cc}
    G_{X',Y} & 0  \\
    0 & G_{X'',Y}
\end{array}
\right)$ and the block $E_{Y,X}$ becomes $\left(
\begin{array}{c}
     E_{Y,X'}  \\
      E_{Y,X''}
\end{array}
\right)$.
} 
Using these notations, gluing formulae (\ref{propagator gluing eq1}), (\ref{propagator gluing eq2}) for the propagator can be jointly expressed as
\begin{equation}\label{gluing of propagators, matrix form}
    G_X\stackrel{!}{=}\ol{G}_{X,Y}+ \ol{E}_{Y,X} \DN_{Y,X}^{-1} \ol{E}_{Y,X}^T.
\end{equation}
The r.h.s. here is
\begin{multline*}
    \left( 
    \begin{array}{c|c}
         G_{X,Y} + E\, \DN^{-1} E^T & E\, \DN^{-1}  \\ \hline
         \DN^{-1} E^T& \DN^{-1} 
    \end{array}
    \right) =
    \left(
    \begin{array}{c|c}
       \wh{A}^{-1}+B D^{-1} D D^{-1} C   & BD^{-1} D \\ \hline
       DD^{-1}C  & D
    \end{array}
    \right)
    \\=
     \left(
    \begin{array}{c|c}
       \wh{A}^{-1}+B D^{-1} C   & B \\ \hline
      C  & D
    \end{array}
    \right).
\end{multline*}
Here we are suppressing the subscript $Y,X$ for $E$ and $\DN$; notations for the blocks are as in (\ref{K_X inverse as a block matrix}), (\ref{K_X  as a block matrix}). So, the only part to check is that the 1-1 block above is $A$. It is a consequence of the inversion formula for $2\times 2$ block matrices, which in particular asserts that the 1-1 block $\wh{A}$ of the matrix $K_X$ inverse to $G_X$ is the inverse of the Schur complement of the 2-2 block in $G_X$, i.e.,
$$ \wh{A}^{-1} = A-B D^{-1}C.  $$
This finishes the proof of the gluing formula for propagators (\ref{gluing of propagators, matrix form}). 

Schur's formula for a determinant of a block $2\times 2$ matrix applied to (\ref{K_X inverse as a block matrix}) yields
\begin{eqnarray*}
    \det K_X^{-1}= \det \underbrace{D}_{\DN^{-1}} \cdot \det (\underbrace{A-BD^{-1}C}_{\wh{A}^{-1}=K_{X,Y}^{-1}})
\end{eqnarray*}
and thus
$$
\det K_X = \det \DN\cdot \det K_{X,Y}= \det \DN\cdot \det K_{X',Y} \cdot \det K_{X'',Y} .
$$
In the last equality we used that $K_{X,Y}$ is block-diagonal, with blocks corresponding to $X'\setminus Y$ and $X''\setminus Y$. This proves the gluing formula for determinants.

\subsubsection{Examples}

\begin{example}
    Consider the gluing of two line graphs of length 2, $X',X''$ over a common vertex $Y$ into a line graph $X$ of length 3 as pictured in Figure \ref{fig: gluing N=3 line} below.
\begin{figure}[H]
    \centering
    \includegraphics{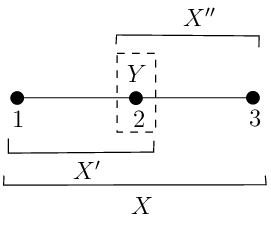}
    \caption{Gluing two line graphs into a longer line graph.}
    \label{fig: gluing N=3 line}
\end{figure} 

The data of the constituent graphs $X',X''$ relative to $Y$ was computed in Example \ref{ex: path 2 graph rel bdry}. We assemble the data on the glued graph $X$ using the gluing formulae of Theorem \ref{thm: gluing prop and det}. We have
    \begin{multline*}
    \underbrace{\DN_{Y,X}}_{=\DN_{Y,X}(2,2)}= \DN_{Y,X'}+\DN_{Y,X''}-K_Y =\\
    =\frac{m^2(2+m^2)}{1+m^2}+\frac{m^2(2+m^2)}{1+m^2}-m^2 = \frac{m^2(3+m^2)}{1+m^2} .
    \end{multline*}
    For the propagator we have, e.g.,
    \begin{multline*} G_X(1,1)=G_{X',Y}(1,1)+ E_{Y,X'}(1,2) \DN^{-1}_{Y,X}(2,2) E_{Y,X'}(1,2) \\
    = \frac{1}{1+m^2}+\frac{1}{1+m^2}\cdot \frac{1+m^2}{m^2(3+m^2)}\cdot \frac{1}{1+m^2} 
    \end{multline*}
    and 
     \begin{multline*} G_X(1,3)=E_{Y,X'}(1,2) \DN^{-1}_{Y,X}(2,2) E_{Y,X''}(3,2) 
    = \frac{1}{1+m^2}\cdot \frac{1+m^2}{m^2(3+m^2)}\cdot \frac{1}{1+m^2} ,
    \end{multline*}
    which agrees with the 1-1 entry and 1-3 entry in (\ref{ex path3 G_X}) respectively.

    For the gluing of determinants, we have
    $$ \det K_{X',Y}\cdot \det K_{X'',Y}\cdot \det\DN_{Y,X} = (1+m^2)\cdot (1+m^2)\cdot \frac{m^2(3+m^2)}{1+m^2} , $$
    which agrees with (\ref{ex path3 det}).
\end{example}

\begin{example}
    Consider the circle graph $X$ with $N$ vertices presented as a gluing by the two endpoints of two line graphs $X'$, $X''$ of lengths $N',N''$ respectively, with $N=N'+N''-2 $, see Figure \ref{fig: circle gluing} below.
    
    \begin{figure}[H]
    \includegraphics{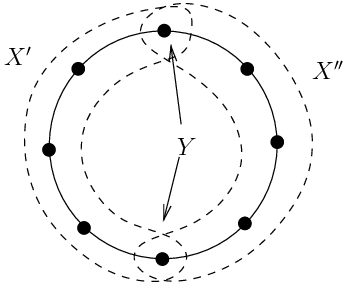}
    \caption{Gluing a circle from two intervals.}
    \label{fig: circle gluing}
    \end{figure}
    
    One can then use the gluing formulae of Theorem \ref{thm: gluing prop and det} to recover the propagator and the determinant on the circle graph (cf. Example \ref{ex: circle graph}) from the data for line graphs relative to the endpoints (cf. Example \ref{ex: path graph rel both ends}). E.g. for the determinant, we have
    \begin{multline*}
    \underbrace{\frac{\sinh \beta (N'-1)}{\sinh \beta}}_{\det K_{X',Y}} \cdot \underbrace{\frac{\sinh \beta (N''-1)}{\sinh \beta}}_{\det K_{X'',Y}} \cdot 
    \det \Big(\DN_{Y,X'}+\DN_{Y,X''}-
    \left(
    \begin{array}{cc}
         m^2 & 0  \\
         0 & m^2
    \end{array}
    \right)
    \Big) = \\
    = \underbrace{4\sinh^2 \frac{\beta N}{2}}_{\det K_X}.
    \end{multline*}
    Here the $2\times 2$ matrices $\DN_{Y,X'}$, $\DN_{Y,X''}$ are given by (\ref{DN for path graph rel both ends}), with $N$ replaced by $N',N''$, respectively.
\end{example}

\subsubsection{General cutting/gluing of cobordisms}
Consider the gluing of graph cobordisms (\ref{gluing of cobs}),
$$ Y_1 \xra{X'} Y_2 \xra{X''} Y_3  \quad =\quad  Y_1 \xra{X} Y_3. $$

Let us introduce the following shorthand notations 
\begin{itemize}
    \item DN operators: $\DN_{ij,A}\colon= \DN_{Y_i\sqcup Y_j,A} $, with $i,j\in \{1,2,3\}$ and $A\in \{X',X'',X\}$. Also, by $(\DN_{ij,A})_{kl}$ we will denote $Y_k-Y_l$ block in $\DN_{ij,A}$.
\item ``Interface'' DN operator: $\DNI\colon=(\DN_{Y_1\sqcup Y_2\sqcup Y_3,X})_{Y_2-Y_2\mr{-block}}$.
\item Extension operators: $E_{ij,A}=E_{Y_i\sqcup Y_j,A}$. We will also denote its $A-Y_k$ block by $(E_{ij,A})_k$.
\item Propagators: $G_{A,ij}\colon= G_{A,Y_i\sqcup Y_j}$.
\end{itemize}

One has the following straightforward generalization of Theorem \ref{thm: gluing prop and det} to the case of possibly nonempty $Y_1,Y_3$.
\begin{thm} \label{thm: gluing formulae for composition of graph cobordisms}
The data of the Gaussian theory on the glued cobordism $Y_1 \xra{X} Y_3$ can be computed from the data of the constituent cobordisms $Y_1 \xra{X'} Y_2$, $Y_2 \xra{X''} Y_3$ as follows
\begin{enumerate}[(a)]
    \item Glued DN operator $\DN_{13,X}$:
    {\tiny
    \begin{equation}\label{gluing of DN operators}
        \left(
        \begin{array}{cc}
             (\DN_{12,X'})_{11}- (\DN_{12,X'})_{12}\DNI^{-1} (\DN_{12,X'})_{21}&   -(\DN_{12,X'})_{12} \DNI^{-1}(\DN_{23,X''})_{23} \\
             - (\DN_{23,X''})_{32}\DNI^{-1} (\DN_{12,X'})_{21} &  (\DN_{23,X''})_{33} - (\DN_{23,X''})_{32}\DNI^{-1} (\DN_{23,X''})_{23}
        \end{array}
        \right).
    \end{equation}
    }
    The blocks correspond to vertices of $Y_1$ and $Y_3$.
The interface DN operator here is
\begin{equation}
    \DNI=(\DN_{12,X'})_{22}+(\DN_{23,X''})_{22}-K_{Y_2}.
\end{equation}
\item Extension operator $ E_{13,X}$:
{\tiny
\begin{equation}
    \left( 
    \begin{array}{cc}
         (E_{12,X'})_1- (E_{12,X'})_2 \DNI^{-1} (\DN_{12,X'})_{21}  &  -(E_{12,X'})_2 \DNI^{-1} (\DN_{23,X''})_{23}  \\
         -  \DNI^{-1} (\DN_{12,X'})_{21}  &  - \DNI^{-1} (\DN_{23,X''})_{23} \\
          - (E_{23,X''})_2 \DNI^{-1} (\DN_{12,X'})_{21}  &  (E_{23,X''})_3 -(E_{23,X''})_2 \DNI^{-1} (\DN_{23,X''})_{23} 
    \end{array}
    \right).
\end{equation}
}
Here horizontally, the blocks correspond to vertices of $Y_1$, $Y_3$; vertically -- to vertices of $X'\setminus (Y_1\sqcup Y_2)$, $Y_2$ and $X''\setminus (Y_2\sqcup Y_3)$.
\item Determinant:
\begin{equation}
    \det K_{X,Y_1\sqcup Y_3}=\det K_{X',Y_1\sqcup Y_2}\cdot \det K_{X'',Y_2\sqcup Y_3}\cdot\det \DNI.
\end{equation}
\item Propagator:
\begin{itemize}
    \item For $v_1,v_2\in X'$, 
    \begin{equation}
        G_{X,13}(v_1,v_2)= G_{X',12}(v_1,v_2)+\sum_{u_1,u_2\in Y} E_{12,X'}(v_1,u_1) \DNI^{-1}(u_1,u_2) E_{12,X'}(v_2,u_2)
    \end{equation}
    and similarly for $v_1,v_2\in X''$.
    \item For $v_1\in X'$, $v_2\in X''$,
    \begin{equation}
        G_{X,13}(v_1,v_2)= \sum_{u_1,u_2\in Y} E_{12,X'}(v_1,u_1) \DNI^{-1}(u_1,u_2) E_{23,X''}(v_2,u_2)
    \end{equation}
    and similarly for $v_1\in X''$, $v_2\in X'$.
\end{itemize}
\end{enumerate}
\end{thm}
\subsubsection{Self-gluing and trace formula}
As another generalization of Theorem \ref{thm: gluing prop and det}, one can consider the case of a graph $X$ relative to a subgraph $Y$ that admits a decomposition $Y = Y_1 \sqcup Y_2$ where $Y_1$ and $Y_2$ are isomorphic graphs. Then, specifying a graph isomorphism $f\colon Y_1 \to Y_2$, we can glue $Y_1$ to $Y_2$ using $f$ to form a new graph $\tilde{X}$ with a distinguished subgraph $\tilde{Y}$.\footnote{In the setting of theorem \ref{thm: gluing prop and det}, we have $X = X'\sqcup X''$, and there are no edges between $X'$ and $X ''$. In the following discussion we will suppress $f$ but remark that in principle the glued graphs $\tilde{X}$ and $\tilde{Y}$ do depend on $f$.} We have $\tilde{Y} \cong Y_1 \cong Y_2$ if and only if there are no edges between $Y_1$ and $Y_2$. See Figure \ref{fig: self gluing}.

\begin{figure}[H]
    \centering
\begin{tikzpicture}
\node at (-.75,.5) {$X$};
    \foreach \x in {0,1,2} 
    {
    \foreach \y in {0,1} 
    {
    \draw[fill = black] (\x,\y) circle (2pt);
    }
    }
    \draw (0,0) -- (1,0) -- (2,1); 
    \draw (0,1) -- (1,1) -- (2,0);
    \draw[dashed] (1,0.5) ellipse (0.25cm and 0.75cm);
    \draw[dashed] (2,0.5) ellipse (0.25cm and 0.75cm);
    \node at (3,0.5) {$\leadsto$};
    \node at (1,-.5) {$Y_1$};
    \node at (2,-.5) {$Y_2$};
    \begin{scope}[shift={(4,0)}]
    \foreach \x in {0,1} 
    {
    \foreach \y in {0,1} 
    {
    \draw[fill = black] (\x,\y) circle (2pt);
    }
    }
    \draw (0,0) -- (1,0); 
    \draw (0,1) -- (1,1); 
    \draw (1,0) to[bend left = 30]  (1,1);
    \draw (1,0) to[bend right = 30] (1,1);
    \draw[dashed] (1,0.5) ellipse (0.25cm and 0.75cm);
        \node at (1,-.5) {$\tilde{Y}$};
        \node at (1.75,0.5) {$\tilde{X}$};
    \end{scope}
\end{tikzpicture}\caption{An example of self-gluing.}
    \label{fig: self gluing}
\end{figure}
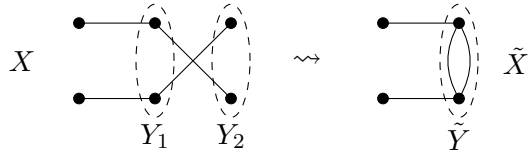

Then one has the following relation between the Dirichlet-to-Neumann operators of $Y$ relative to $X$ and $\tilde{Y}$ relative to $\tilde{X}$:  
\begin{proposition}\label{prop: self gluing}
Let $\phi \in C^0(Y_1) 
\simeq C^0(Y_2) \simeq C^0(\tilde{Y})$, then\footnote{Below we are identifying using $f$ to identify $V(Y_1)$ and $V(Y_2)$, and then also $\phi$ and $(f^{-1})^*\phi$.}    \begin{equation}
        \left( (\phi,\phi), \left(\DN_{Y,X} - \frac12 K_Y\right) \begin{pmatrix} \phi \\ \phi \end{pmatrix}\right) = (\phi, \DN_{\tilde{Y},\tilde{X}} \phi)  -  \left(\phi, \left(\frac12 K_{\tilde{Y}} - \frac{1}{2}K_{Y_1}\right)\phi\right).\label{eq: prop self gluing I}
    \end{equation}
    Equivalently, 
    \begin{equation} \left( (\phi,\phi), \left(\DN_{Y,X} - \frac12 K_{Y_1} - \frac12 K_{Y_2}\right) \begin{pmatrix} \phi \\ \phi \end{pmatrix} \right) = (\phi, \DN_{\tilde{Y},\tilde{X}} \phi). \label{eq: prop self gluing II}
    \end{equation}
    \end{proposition}
\begin{proof}
    We have 
    $$S_X = S_{X,Y}  + S_{Y_1} + S_{Y_2} + S_{Y_1,Y_2},$$
    where the first term contains contributions to the action from vertices in $X\setminus Y$ and edges with at least one vertex in $X \setminus Y$, while the last term contains just contributions from edges between $Y_1$ and $Y_2$. Evaluating on the subspace of fields $F_X^{({\phi,\phi})}$ that agree on $Y_1$ and $Y_2$, we get  
    $$ S_X\big|_{F_X^{({\phi,\phi})}} = S_{X,Y} + 2S_{Y_1}(\phi) + S_{Y_1,Y_2}(\phi)$$ 
    and $S_X - \frac{1}{2}S_Y = S_{X,Y} + S_{Y_1} + \frac{1}{2}S_{Y_1,Y_2}.$
    On the other hand, we have 
    $$
    \begin{aligned}
    S_{\tilde{X}} - \frac12 S_{
    \tilde{Y}}\big|_{F_{\tilde{X}}^\phi} &= S_{\tilde{X},\tilde{Y}} + \frac12S_{\tilde{Y}}(\phi)= S_{X,Y} + \frac{1}{2}S_{Y_1}(\phi) +\frac{1}{2} S_{Y_1,Y_2}(\phi) \\ &= S_X - \frac{1}{2}S_{Y}(\phi) - \frac12 S_{Y_1}.
    \end{aligned}
    $$
    Therefore, 
    $$Z_{X,Y}((\phi,\phi)) = Z_{\tilde{X},\tilde{Y}}(\phi) e^{-\frac{1}{2\hbar}S_{Y_1}(\phi)}. $$
    Noticing that the relative operators agree $K_{X,Y} = K_{\tilde{X},\tilde{Y}}$, and using \eqref{Gaussian Z rel to Y}, we obtain \eqref{eq: prop self gluing I}. 
    To see \eqref{eq: prop self gluing II}, notice that the difference $K_Y - K_{Y_1} - K_{Y_2} = K_{Y_1,Y_2} = K_{\tilde{Y}} - K_{Y_1}$, so adding $\frac{1}{2}K_{Y_1,Y_2}$ to \eqref{eq: prop self gluing I} we obtain \eqref{eq: prop self gluing II}.
\end{proof}
\begin{corollary}
We have the following trace formula
\begin{equation}
    Z_{\tilde{X}} = \int_{F_{Y_1}}[D\phi]\langle \phi | Z_{X,Y_1,Y_2}|\phi \rangle .
\end{equation}
\begin{proof}
    We have 
    \begin{align*}
    \langle \phi |Z_{X,Y_1,Y_2}|\phi\rangle) &= \det K_{X,Y}^{-\frac12}e^{-\frac{1}{2\hbar}((\phi,\phi),(\DN_{Y,X} -\frac12 K_{Y_1} - \frac12 K_{Y_2})(\phi,\phi)} \\ 
    &=\det K_{\tilde{X},\tilde{Y}}^{-\frac12} e^{-\frac{1}{2\hbar}((\phi,\DN_{\tilde{Y},\tilde{X}})\phi) }.
    \end{align*}
    Integrating over $\phi$, we obtain the result. 
\end{proof}
\end{corollary}

\begin{example}[Gluing a circle graph from a line graph]
     For the line graph $L_3$ relative to both endpoints, $$DN_{Y,X} = \frac{m^2+3}{m^2+2}\begin{pmatrix} m^2 +1  & - 1 \\ 
    -1 & m^2 +1         
    \end{pmatrix}. $$
    In this case we have $K_{Y_1} = K_{Y_2} = m^2$ and 
    $$\begin{pmatrix}
        1 & 1 
    \end{pmatrix} \frac{m^2+3}{m^2+2}\begin{pmatrix} m^2 +1  & - 1 \\ 
    -1 & m^2 +1         
    \end{pmatrix} \begin{pmatrix}
        1 \\ 1 
    \end{pmatrix} = \frac{2m^2(m^2+3)}{m^2+2} ,$$
    which implies 
    $$\frac{2m^2(m^2+3)}{m^2 +2} - \underbrace{m^2}_{-\frac12 K_{Y_1} -\frac12 K_{Y_2}} = \frac{m^2(m^2+4)}{m^2 +2}.$$
    Here $m^2 = K_{Y_1}$.
    On the other hand, $\tilde{X} = C_2$ is a circle graph with $\tilde{Y}$ a point, and we have 
    $$K_X = \begin{pmatrix}
        m^2 + 2 & -2 \\ 
        -2 & m^2 +2 
    \end{pmatrix}, \qquad G_X =  \frac{1}{m^2(m^2+4)} \begin{pmatrix}
      m^2 + 2 &  2\\ 
      2 & m^2 +2
    \end{pmatrix},$$
    therefore the corresponding Dirichlet-to-Neumann operator is 
    $$ \DN_{\tilde{Y},\tilde{X}} = \frac{m^2(m^2+4)}{m^2+2} ,$$
    as predicted by Proposition \ref{prop: self gluing}. The relative determinant $K_{Y,X}$ is $m^2 + 2$ so that the trace formula becomes
    $$ \int_{F_{Y_1}}[D\phi]\langle \phi | Z_{X,Y_1,Y_2}|\phi \rangle = (m^2+2)^{-\frac12}\left(\frac{m^2(m^2+4)}{m^2+2}\right)^{-\frac12} = (m^2(m^2+4))^{-\frac12} = Z_{C_2}.$$
     Similarly, for the line graph of length $N$ relative to both endpoints, the Dirichlet-to-Neumann operator is given by \eqref{DN for path graph rel both ends} and we have  
     $$\begin{pmatrix}
        1 & 1 
    \end{pmatrix} \DN_{Y,X} \begin{pmatrix}
        1 \\ 1 
    \end{pmatrix} = \frac{4\sinh\frac{\beta}{2}(\cosh\beta (N-\frac12) 
    - \cosh \frac{\beta}{2})}{\sinh \beta(N-1)} .$$
    On the other hand, the Dirichlet-to-Neumann operator of $\tilde{X}=C_{N-1}$ relative to a single vertex is 
    $$\DN_{\tilde{Y},\tilde{X}} = 2\sinh\beta \tanh\beta\frac{N-1}{2}.$$
    Then one 
    can check that 
    $$  \frac{4\sinh\frac{\beta}{2}(\cosh\beta (N-\frac12) 
    - \cosh \frac{\beta}{2})}{\sinh \beta(N-1)} - m^2 = 2\sinh\beta  \tanh\beta\frac{N-1}{2}.$$
\end{example}
\begin{remark} 
    There is of course also common generalization of Theorem \ref{thm: gluing formulae for composition of graph cobordisms} and Proposition \ref{prop: self gluing}, where we have several boundary components and are allowed sew any two isomorphic components together, we leave this statement to the imagination of the reader. 
\end{remark}

\subsection{Comparison to continuum formulation} \label{ss: gluing in Gaussian theory: comparison to continuum}
In this subsection, we compare of results of subsections \ref{subsec:Gaussian-theory-relative-boundary} and \ref{subsec:Gluing-in-Gaussian-theory} to the continuum counterparts for a free scalar theory on a Riemannian manifold. For details on the latter, we refer to \cite{KMW}. 

Consider the free scalar theory on a closed Riemannian manifold $M$ defined by the action
$$ S(\phi)=\int_M \frac12 d\phi \wedge *d\phi +\frac{m^2}{2}\phi^2 \dvol = \int_M \frac12 \phi (\Delta+m^2)\phi\, \dvol ,$$
where $\phi\in C^\infty(M)$ is the scalar field, $m>0$ is the mass, $*$ is the Hodge star associated with the metric, $\dvol$ is the metric volume form and $\Delta$ is the metric Laplacian.

The partition function is defined to be
$$ Z=``{\int D\phi \,e^{-\frac{1}{\hbar} S(\phi)} \;}"\colon= \left({\det}^\zeta (\Delta+m^2)\right)^{-\frac12} ,$$
where ${\det}^\zeta$ stands for the functional determinant  understood in the sense of zeta-regularization. Correlators are given by Wick's lemma in terms of the Green's function $G(x,y)\in C^\infty(M\times M \setminus \mr{Diag})$ of the operator $\Delta+m^2$.

Next, if $M$ is a compact Riemannian manifold with boundary $\dd M$, one can impose Dirichlet boundary condition $\phi|_{\dd M}= \phi_\dd$ -- a fixed function on $\dd M$ (thus, fluctuations of fields are zero on the boundary). 
The unique solution of the Dirichlet boundary value problem on $M$,
$$ (\Delta+m^2)\phi=0,\quad \phi|_{\dd}=\phi_\dd, $$
can be written as
\begin{equation}
    \phi(x)=\int_{\dd M} \dd^n_y G_D(x,y) \,\phi_\dd(y) \dvol^\dd_y.
\end{equation}
Here:
\begin{itemize}
    \item $\dvol^{\dd}$ is the Riemannian volume form on $\dd M$ (w.r.t. the induced metric from the bulk).
    \item $G_D \in C^\infty(M\times M\setminus \mr{Diag})$ is the Green's function for the operator $\Delta+m^2$ with Dirichlet boundary condition.
    \item $\dd^n$ stands for the normal derivative at the boundary. In particular, for $x\in M$, $y\in \dd M$, 
    \begin{equation}\label{dn dn G}
    \dd^n_y G_D(x,y)  =  \left.\frac{\dd}{\dd t}\right|_{t=0} G_D(x,\til{y}_t), 
    \end{equation}
where $\til{y}_t$, $t\geq 0$ is a curve in $M$ starting at $\til{y}_0=y$ with initial velocity being the inward unit normal to the boundary.  
\end{itemize}

Then on a manifold with boundary one has the partition function
\begin{multline}\label{Z continuum with bdry}
Z(\phi_\dd) =``{\int_{\phi|_\dd=\phi_\dd} D\phi \,e^{-\frac{1}{\hbar} S(\phi)} \;}"\\= \left({\det}^{\zeta}_D(\Delta+m^2)\right)^{-\frac12} e^{-\frac{1}{2\hbar}\int_{\dd M} \phi_\dd \DN(\phi_\dd) \dvol^\dd } \\
= \left({\det}^{\zeta}_D(\Delta+m^2)\right)^{-\frac12} e^{\frac{1}{2\hbar} \int_{\dd M\times \dd M} \phi_\dd(x) \,\dd^n_x \dd^n_y G_D(x,y)\, \phi_\dd(y) d\mr{vol}_x^\dd d\mr{vol}_y^\dd }.
\end{multline}
Here in the determinant in the r.h.s., $\Delta+m^2$ is understood as acting on smooth functions on $M$ vanishing on $\dd M$ (which we indicate by the subscript $D$ for ``Dirichlet boundary condition''); $\DN\colon C^\infty(\dd M)\ra C^\infty(\dd M)$ is the Dirichlet-to-Neumann operator (see footnote \ref{footnote: DN operator in continuum}).
The integral kernel of the DN operator is 
$ - \dd^n_x \dd^n_y G_D(x,y)$.

The integral in the exponential in the last line of (\ref{Z continuum with bdry}) contains a non-integrable singularity on the diagonal and has to be appropriately regularized, cf. Remark 3.4 in \cite{KMW}.

Correlators on a manifold with boundary are:
\begin{itemize}
    \item One-point correlator: 
    $$\langle \phi(x) \rangle_{\phi_\dd} = \int_{\dd M}\dd^n_y G_D(x,y) \phi_\dd(y) \dvol^\dd_y .$$
    \item Centered two-point correlator:
    $$ \langle \delta\phi (x) \delta \phi(y) \rangle_{\phi_\dd} =\hbar\, G_D(x,y) ,$$
    where $\delta \phi(x)\colon= \phi(x)-\langle \phi(x) \rangle_{\phi_\dd}$.
    \item $k$-point centered correlators are given by Wick's lemma.
\end{itemize}

When more detailed notations of the manifolds involved is needed, instead of $G_D$ we will write $G_{M, \dd M}$ (and similarly for ${\det}^\zeta_D$) and instead of $\DN$ we will write $\DN_{\dd M,M}$.

Continuing the dictionary of Remark \ref{rem: graph-continuum QFT dictionary} to free scalar theory on graphs vs. Riemannian manifolds, we have the following.

\begin{table}[H]
\bgroup\def\arraystretch{1.4}
\begin{tabular}{c|c}
  Scalar theory on a graph $X$  &  Scalar theory on a Riemannian manifold $M$  \\
   relative to subgraph $Y$   &  with boundary $\dd M$  \\ \hline 
  propagator $G_{X,Y}$ &  $G_D(x,y)$ with $x,y\in M$ \\
  extension operator $E_{Y,X}$ &  $\dd^n_y G_D(x,y)$ with $x\in M,\; y\in\dd M$ \\
  DN operator $\DN_{Y,X}$ & $\DN$ or $-\dd^n_x \dd^n_y G_D(x,y)$ with $x,y\in \dd M$ \\ 
  $\det K_{X,Y}$ &  ${\det}^\zeta_D (\Delta+m^2)$
\end{tabular}
\egroup
\caption{Comparing the toy model and continuum theory, continued.}
\end{table}
\vspace{0.3cm}
Here in the right column we 
are writing the integral kernels of operators instead of operators themselves.

\subsubsection{Gluing formulae for Green's functions and determinants}\label{subsubsec: gluing continuum}
Assume that $M$ is a closed Riemannian manifold cut be a closed codimension 1 submanifold $\gamma$ into two pieces, $M'$ and $M''$. Then one can recover the Green's function for the operator $\Delta+m^2$ on $M$ from Green's functions on $M'$ and $M''$, both with Dirichlet condition on $\gamma$, as follows.\footnote{
This result is originally due to \cite{Carron}, see also \cite{KMW}.
}
\begin{itemize}
    \item For $x,y\in M'$, one has
    \begin{equation}\label{G gluing continuum eq1}
    G_M(x,y) = G_{M',\gamma}(x,y) + \int_{\gamma} \dvol^\gamma_{w} \int_\gamma \dvol^\gamma_z\, \dd^n_w G_{M',\gamma}(x,w)  \varkappa(w,z) \dd^n_z G_{M',\gamma}(z,y) .
    \end{equation}
    Here $\varkappa(w,z)$ is the integral kernel of the inverse of the total (or ``interface'') DN operator 
    \begin{equation} \label{DN interface continuum}
    \DNI\colon=\DN_{\gamma,M'}+\DN_{\gamma,M''}.
    \end{equation}
    For $x,y\in M''$, one has a similar formula to (\ref{G gluing continuum eq1}).
    \item For $x\in M'$, $y\in M''$, one has
    \begin{equation}\label{G gluing continuum eq2}
        G_M(x,y)= \int_{\gamma} \dvol^\gamma_{w} \int_\gamma \dvol^\gamma_z\, \dd^n_w G_{M',\gamma}(x,w)  \varkappa(w,z) \dd^n_z G_{M'',\gamma}(z,y) ,
    \end{equation}
    and similarly if $x\in M''$, $y\in M'$.
\end{itemize}

In the case $\dim M=2$, the zeta-regularized determinants satisfy a remarkable Mayer-Vietoris type gluing formula due to Burghelea-Friedlander-Kappeler \cite{BFK},
\begin{equation}\label{BFK continuum}
    {\det}^\zeta_{M}(\Delta+m^2)= {\det}^\zeta_{M',\gamma}(\Delta+m^2)\,{\det}^\zeta_{M'',\gamma}(\Delta+m^2)\, {\det}^\zeta_\gamma (\DNI).
\end{equation}
This formula also holds for higher even dimensions provided that the metric near the cut $\gamma$ is of warped product type (this is a result of Lee \cite{Lee}). In odd dimensions, under a similar assumption, the formula is known to hold up to a multiplicative constant known explicitly in terms of the metric on the cut.

Note that formulae (\ref{G gluing continuum eq1}), (\ref{G gluing continuum eq2}) have the exact same structure as formulae (\ref{propagator gluing eq1}), (\ref{propagator gluing eq2}) for gluing of graph propagators.\footnote{
One small remark is that the continuum formula for the interface DN operator (\ref{DN interface continuum}) is similar to  (\ref{DN total}), except for the $-K_Y$ term in the l.h.s. which is specific to the graph setting and disappears in the continuum limit.
} Likewise, the gluing formulae for determinants in the continuum setting (\ref{BFK continuum}) and in graph setting (\ref{determinant gluing}) have the same structure.

One can also allow the manifold $M$ to have extra boundary components disjoint from the cut, i.e., to consider $M$ as a composition of two cobordisms $\gamma_1\xra{M'}\gamma_2,\; \gamma_2\xra{M''}\gamma_3$. One then has the corresponding gluing formulae which have the same structure as the formulae of Theorem \ref{thm: gluing formulae for composition of graph cobordisms}. In particular, one has a gluing formula for continuum DN operators (see \cite{Pickrell}) similar to the formula (\ref{gluing of DN operators}) in the graph setting  .

\subsubsection{Example: continuum limit of line and circle graphs}\label{subsubsec: continuum example}
The action of the continuum theory on an interval $[0,L]$ evaluated on a smooth field $\phi\in C^\infty([0,L])$ can be seen as a limit of Riemann sums
$$ S(\phi)=\lim_{N\ra \infty} \sum_{i=2}^{N} \frac{\epsilon_N}{2}   \left(\frac{\phi( i\epsilon_N )-\phi((i-1)\epsilon_N)}{\epsilon_N}\right)^2 + \sum_{i=1}^N \epsilon_N \frac{m^2}{2}\phi(i\epsilon_N)^2 , $$
where in the r.h.s. we denoted $\epsilon_N = L/N$. The r.h.s. can be seen as the action of the graph theory on a line graph with $N=L/\epsilon$ vertices, where the mass is scaled as $m\mapsto \epsilon m $ and then the kinetic operator is scaled as $K\mapsto \epsilon^{-1}K$ (and thus the propagator scales as $G\ra \epsilon G$), where we consider the limit $\epsilon\ra 0$ (we are approximating the interval by a portion of a 1d lattice and taking the lattice spacing to zero).

Applying the scaling above to the formulae of Example \ref{ex: path graph rel both ends}, we obtain the following for the propagator (\ref{G for path graph rel both ends}):
$$ G_\mr{graph}(x,y) \underset{\epsilon\ra 0}{\sim}  \frac{\cosh m(L-|x-y|)-\cosh m(L-x-y)}{\sinh mL} , $$
where we denoted $x=i\epsilon, y=j\epsilon$ -- we think of $i,j$ as scaling with $\epsilon$ so that $x,y$ remain fixed. The r.h.s. above is the Green's function for the operator $\Delta+m^2$ on an interval $[0,L]$ with Dirichlet boundary conditions at the endpoints.\footnote{
For the formulae pertaining to the continuum theory on an interval, see e.g. \cite[Appendix A.1]{KMW}.
} For the DN operator (\ref{DN for path graph rel both ends}), we obtain 
$$ 
\DN_\mr{graph} \underset{\epsilon\ra 0}{\sim}  \frac{m}{\sinh mL} \left( 
\begin{array}{cc}
    \cosh mL  & -1 \\
    -1 & \cosh mL
\end{array}
\right) .
$$
The r.h.s. is the correct DN operator of the continuum theory on the interval. 

For the determinant (\ref{det for path graph rel both ends}), we have
$$ \det K_\mr{graph}  \underset{\epsilon\ra 0}{\sim}  \epsilon^{-N}\frac{\sinh mL}{m} .$$
For comparison, the zeta-regularized determinant on the interval is
$$ {\det}^\zeta_D (\Delta+m^2)=\frac{2\sinh mL}{m} .$$
It differs from the graph result by a scaling factor $\epsilon^{N}$ and an extra factor $2$ which exhibits a discrepancy between the two regularizations of the functional determinant -- lattice vs. zeta regularization.

\begin{remark}
    One can similarly consider the continuum limit for the line graph of Example \ref{ex: path graph}, without Dirichlet condition at the endpoints. Its continuum counterpart is the theory on an interval $[0,L]$ with \emph{Neumann} boundary conditions at the endpoints, cf. footnote \ref{footnote: path graph -- Neumann bc}. Likewise, in the continuum limit for line graphs relative to one endpoint (Example \ref{ex: path graph rel one end}), one recovers the continuum theory with Dirichlet condition at one endpoint and Neumann condition at the other. 
    
    For example, the zeta-determinant for Neumann condition at both ends is ${\det}_{N-N}^\zeta(\Delta+m^2)=2m\sinh mL$. For Dirichlet condition at one end and Neumann at the other, one has ${\det}^\zeta_{D-N}(\Delta+m^2)=2\cosh mL$. These formulae are related to the continuum limit of the discrete counterparts (\eqref{N-N graph det} for Neumann-Neumann and \eqref{N-D graph determinant} for Neumann-Dirichlet boundary conditions) in the same way as in the Dirichlet-Dirichlet case (by scaling with $\epsilon^N$ and an extra factor of 2). 
\end{remark}

In the same vein, we can consider a circle of length $L$ as a limit of circle graphs (Example \ref{ex: circle graph}) with spacing $\epsilon$. Then in the scaling limit, from (\ref{G circle graph}) we have
$$ G_\mr{graph}(x,y)  \underset{\epsilon\ra 0}{\sim} \frac{\cosh m (\frac{L}{2}-|x-y|)}{2m \sinh \frac{mL}{2}},$$
where the r.h.s. coincides with the continuum Green's function on a circle. For the determinant (\ref{det circle graph}), we have
\begin{equation}
\det K_\mr{graph} \underset{\epsilon\ra 0}{\sim} \epsilon^{-N} 4\sinh^2\frac{mL}{2} . \label{det circle continuum limit}
\end{equation} 
For comparison, the corresponding zeta-regularized functional determinant is
$$ {\det}^\zeta (\Delta+m^2) = 4\sinh^2\frac{mL}{2},$$
which coincides with the r.h.s. of (\ref{det circle continuum limit}) up to the scaling factor $\epsilon^{N}$.

\section{Interacting theory via Feynman diagrams}\label{s: Interacting theory via Feynman diagrams}


Consider scalar field theory on a closed graph $X$ defined by the action (\ref{S}) -- the perturbation of the Gaussian theory by an interaction potential $p(\phi)$:
$$ S_X(\phi)=\underbrace{\frac12 (\phi,K_X\phi)}_{S^{0}_X(\phi)} + \underbrace{\langle \mu, p(\phi) \rangle}_{S^\mr{int}_X(\phi)} . $$
The partition function (\ref{Z}) can be computed by perturbation theory -- the Laplace method for the $\hbar\ra 0$ asymptotics of the integral, with corrections given by Feynman diagrams (see e.g. \cite{Etingof}):
\begin{multline}\label{Z perturbative expansion}
    Z_X=\int_{F_X} D\phi\, e^{-\frac{1}{\hbar} (S^0_X(\phi)+S^\mr{int}_X(\phi)
    )} = \ll e^{-\frac{1}{\hbar} S_X^\mr{int}(\phi)} \gg^0=
    \\
    =\ll \sum_{n\geq 0} \frac{(-1)^n}{\hbar^{n} n!}\sum_{v_1,\ldots,v_n\in V_X}p(\phi(v_1))\cdots p(\phi(v_n))  \gg^0
    \underset{\hbar\ra 0}{\sim} \\
     \underset{\hbar\ra 0}{\sim} 
     \det (K_X)^{-\frac12}\cdot \sum_{\Gamma} \frac{\hbar^{-\chi(\Gamma)}}{|\mr{Aut}(\Gamma)|} \Phi_{\Gamma,X}.
\end{multline}
Here: 
\begin{itemize}
    \item  $\ll\cdots \gg^0$ stands for the non-normalized correlator in the Gaussian theory.
    \item The sum in the r.h.s. is over finite (not necessarily connected) graphs $\Gamma$ (Feynman graphs) with vertices of valence $\geq 3$; $\chi(\Gamma)\leq 0$ is the Euler characteristic of a graph; $|\mr{Aut}(\Gamma)|$ -- the order of the automorphism group of the graph.
    \item The weight of the Feynman graph is\footnote{
    We are using sans-serif font to distinguish vertices of the Feynman graph $\s{u},\s{v}$ from the vertices of the spacetime graph $u,v$.
    }
    \begin{equation}\label{Feynman weight, X closed}
        \Phi_{\Gamma,X}= \sum_{f\colon V_\Gamma\ra V_X} \prod_{\s{v}\in V_\Gamma} (-p_{\mr{val}(\s{v})})\cdot \prod_{(\s{u},\s{v})\in E_\Gamma}G_X(f(\s{u}),f(\s{v})),
    \end{equation}
    where $p_k$ are the coefficients of the potential (\ref{p(phi)}); $G_X$ is the Green's function of the Gaussian theory. The sum over $f$ here -- the sum over $|V_\Gamma|$-tuples of vertices of the spacetime graph $X$ -- is a graph QFT analog of the configuration space integral formula for the weight of a Feynman graph (cf. e.g. \cite{KMW}). 
\end{itemize}

The expression in the r.h.s. of (\ref{Z perturbative expansion}) is a power series in nonnegative powers in $\hbar$, with finitely many graphs contributing at each order.\footnote{This is due to the restriction that valencies in $\Gamma$ are $\geq 3$, which is in turn due to the assumption that
$p(\phi)$ contains only terms of degree $\geq 3$.} We denote the r.h.s. of (\ref{Z perturbative expansion}) by $Z_X^\mr{pert}$ -- the perturbative partition function. It is an asymptotic expansion of the measure-theoretic integral in the l.h.s. of (\ref{Z perturbative expansion}) -- the nonperturbative partition function. 

\subsection{Version relative to a boundary subgraph}
Let $X$ be graph and $Y\subset X$ a ``boundary'' subgraph. We have the following relative version of the perturbative expansion (\ref{Z perturbative expansion}):
\begin{multline}\label{Z relative, perturbative expansion}
    Z_{X,Y}(\phi_Y)=\int_{F_X^{\phi_Y}} [D \phi]^{\phi_Y} e^{-\frac{1}{\hbar}(S_X(\phi)-\frac12 S_Y(\phi_Y))} =\\
   = e^{-\frac{1}{2\hbar}S_Y^\mr{int}(\phi_Y)} \ll e^{-\frac{1}{\hbar}S^\mr{int}_{X\setminus Y}(\phi)} \gg^0_{\phi_Y}
    \\
    \underset{\hbar\ra 0}{\sim}
    \det(K_{X,Y})^{-\frac12} \cdot e^{-\frac{1}{2\hbar} ( (\phi_Y,(\DN_{Y,X}-\frac12 K_Y)\phi_Y) +S_Y^\mr{int}(\phi_Y))} \cdot \sum_\Gamma \frac{\hbar^{-\chi(\Gamma)}}{|\mr{Aut}(\Gamma)|} \Phi_{\Gamma,(X,Y)}(\phi_Y).
\end{multline}
Here the sum is over Feynman graphs $\Gamma$ with vertices split into two subsets -- ``bulk'' vertices $V^\mr{bulk}_\Gamma$ and ``boundary'' vertices $V^\dd_\Gamma$ -- with bulk vertices of valence $\geq 3$ and univalent boundary vertices. In graphs $\Gamma$ we are not allowing edges connecting two boundary vertices (while bulk-bulk and bulk-boundary edges are allowed). The weight of a Feynman graph is a polynomial in the boundary field $\phi_Y$:
\begin{multline}\label{Feynman weight, relative}
    \Phi_{\Gamma,(X,Y)}(\phi_Y)=\\
    = \sum_{f\colon 
    {\tiny
    \begin{array}{c}
         V_\Gamma^\mr{bulk} \ra V_X\setminus V_Y  \\
         V_\Gamma^\dd \ra V_Y
    \end{array}
    }}
    \prod_{\s{v}\in V_\Gamma^\mr{bulk}} (-p_{\mr{val}(\s{v})})\cdot 
    \prod_{\s{u}^\dd\in V_\Gamma^\dd} \phi(f(\s{u}^\dd))
    \cdot 
    \\
    \cdot 
   \prod_{(\s{u},\s{v})\in E_\Gamma^\mr{bulk-bulk}} G_{X,Y}(f(\s{u}),f(\s{v}))\cdot
    \prod_{(\s{u}^\dd,\s{v})\in E_\Gamma^{\mr{bdry-bulk}}} E_{Y,X}(f(\s{u}^\dd),f(\s{v})).
\end{multline}
The sum over $f$ here can be seen as a sum over tuples of bulk and boundary vertices in $X$. Similarly to (\ref{Feynman weight, X closed}), it is a graph QFT analog of a configuration space integral formula for the Feynman diagrams in the interacting scalar field theory on manifolds with boundary (cf. \cite{KMW}), where one is integrating over configurations of $n$ bulk points and $m$ boundary points on the spacetime manifold.

We will denote the r.h.s. of (\ref{Z relative, perturbative expansion}) by $Z^\mr{pert}_{X,Y}(\phi_Y)$.

\begin{example}
    Figure \ref{fig: feynman graph with boundary} is an example of a map $f$ contributing to the Feynman weight (\ref{Feynman weight, relative}):
    \begin{figure}[H]
    \includegraphics[scale=0.7]{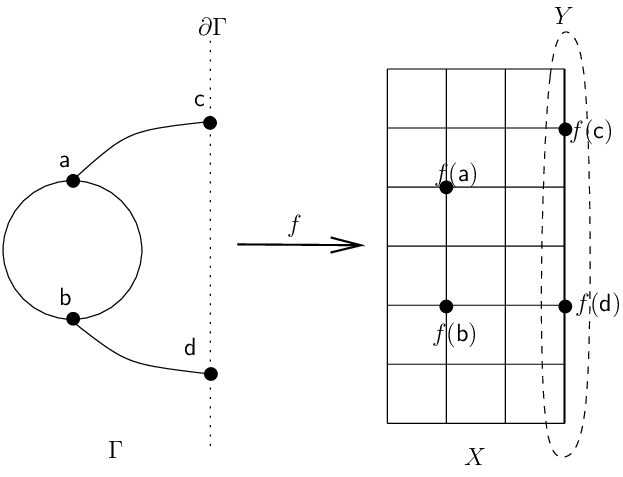}
    \caption{A Feynman graph with boundary vertices and a map contributing to its Feynman weight.}
    \label{fig: feynman graph with boundary}
    \end{figure}
    The full Feynman weight of the graph on the left is:
    $$  
    \Phi_{\Gamma,(X,Y)}(\phi_Y)=\sum_{a,b\in V_X\setminus V_Y,\; c,d \in V_Y}
    (p_3)^2  G_{X,Y}(a,b)^2 E_{Y,X}(c,a) E_{Y,X}(d,b)\, \phi(c) \phi(d),
    $$
    where we denoted $a=f(\s{a}),\, b=f(\s{b}),\, c=f(\s{c}),\, d=f(\s{d})$.
\end{example}

\begin{remark}\label{rem: rems on Z pert rel}
    \begin{enumerate}[(i)]
        \item By the standard argument, due to multiplicativity of Feynman weights w.r.t. disjoint unions of Feynman graphs,  the sum over graphs $\Gamma$ in (\ref{Z perturbative expansion}), (\ref{Z relative, perturbative expansion}) can be written as the exponential of the sum over connected Feynman graphs, $\sum_\Gamma \cdots = e^{\sum_{\Gamma\,\mr{connected}}\cdots}$.
        \item One can rewrite the r.h.s. of (\ref{Z relative, perturbative expansion}) without the DN operator in the exponent in the prefactor, but instead allowing graphs $\Gamma$ with boundary-boundary edges. The latter  contribute extra factors $-\DN_{Y,X}(\s{u}^\dd,\s{v}^\dd)$ in the Feynman weight (\ref{Feynman weight, relative}).
        \item Unlike the case of closed $X$, the sum over $\Gamma$ in the r.h.s. of (\ref{Z relative, perturbative expansion}) generally contributes infinitely many terms to each nonnegative order in $\hbar$ (for instance, in the order $O(\hbar^0)$, one has 1-loop graphs formed by trees connected to a cycle). However, there are finitely many graphs contributing to a given order in $\hbar$, in any fixed polynomial degree in $\phi_Y$. Moreover, one can introduce a rescaled boundary field $\eta_Y$ so that 
        \begin{equation}\label{phi_Y rescaling}
        \phi_Y=\sqrt\hbar\, \eta_Y.
        \end{equation} 
        Then (\ref{Z relative, perturbative expansion}) expressed as a function of $\eta_Y$ is a power series in nonnegative half-integer powers of $\hbar$, with finitely many graphs contributing at each order.\footnote{The power of $\hbar$ accompanying a graph is $\hbar^{|E_\Gamma|-|V_\Gamma^\mr{bulk}|-\frac12 |V_\Gamma^\dd|}$, i.e., one can think that with this normalization of the boundary field, boundary vertices contribute $1/2$ instead of $1$ to the Euler characteristic of a Feynman graph.
        
        We also note that the rescaling (\ref{phi_Y rescaling}) is rather natural, as the expected magnitude of fluctuations of $\phi_Y$ around zero is $O(\sqrt\hbar)$.}
    \end{enumerate}
\end{remark}

\subsection{Cutting/gluing of perturbative partition functions via cutting/gluing of Feynman diagrams}\label{sec: cutting pert Z}
As in Section \ref{sss Gaussian theory, cutting  closed graph}, consider a closed graph $X=X'\cup_Y X''$ obtained from graphs $X',X''$ by gluing along a common subgraph $X'\supset Y \subset X''$ (but now we consider the interacting scalar QFT). 

As we know from Proposition \ref{prop: functoriality}, the nonperturbative partition functions satisfy the gluing formula
$$ Z_{X}= \int_{F_Y} D\phi_Y\, Z_{X',Y}(\phi_Y) Z_{X'',Y}(\phi_Y). $$
Replacing both sides with their expansions (asymptotic series) in $\hbar$, we have the gluing formula for the perturbative partition functions
\begin{equation}\label{gluing of Z^pert}
Z^\mr{pert}_{X}= \int_{F_Y} D\phi_Y\, Z^\mr{pert}_{X',Y}(\phi_Y) Z^\mr{pert}_{X'',Y}(\phi_Y).
\end{equation} 
This latter formula admits an independent proof in the language of Feynman graphs which we will sketch here (adapting the argument of \cite{KMW}).

Consider ``decorations'' of Feynman graphs $\Gamma$ for the theory on $X$ by the following data:
\begin{itemize}
    \item Each vertex $\s{v}$ of $\Gamma$ is decorated by one of three symbols $\{X',X'',Y\}$, meaning that in the Feynman weight $f(\s{v})$ is restricted to be in $V_{X'}\setminus V_Y$, in $V_{X''}\setminus V_Y$, or in $V_Y$, respectively.
    \item Each edge $e=(\s{u},\s{v})$ of $\Gamma$ is decorated by either $u$ or $c$ (``uncut'' or ``cut''), corresponding to the splitting of the Green's function on $X$ in Theorem \ref{thm: gluing prop and det}: $$G_X(f(\s{u}),f(\s{v}))=G_X^u(f(\s{u}),f(\s{v}))+G_X^c(f(\s{u}),f(\s{v})).$$
    Here: 
    \begin{itemize}
        \item If $\s{u}, \s{v}$ are both decorated with $X'$, the ``uncut'' term is  $G_X^u\colon = G_{X',Y}$. Similarly, if $\s{u}, \s{v}$ are both decorated with $X''$, $G_X^u\colon = G_{X'',Y}$. For all other decorations of $\s{u},\s{v}$, $G_X^u\colon=0$. Because of this, we will impose a selection rule: $u$-decoration is only allowed for $X'-X'$ or $X''-X''$ edges.
        \item The ``cut'' term is
        $$G_X^c(f(\s{u}),f(\s{v}))\colon= \sum_{w_1,w_2\in Y} E_{Y,\alpha}(f(\s{u}),w_1) \DN_{Y,X}^{-1}(w_1,w_2) E_{Y,\beta}(f(\s{v}),w_2) ,$$ 
        where $\alpha,\beta$ are the decorations of $\s{u},\s{v}$ (and we understand $E_{Y,Y}$ as identity operator).
    \end{itemize}
\end{itemize}

Let $\mr{Dec}(\Gamma)$ denote the set of all possible decorations of a Feynman graph $\Gamma$.
Theorem \ref{thm: gluing prop and det} implies that  for any Feynman graph $\Gamma$ its weight splits into the contributions of its possible decorations:
$$ 
\Phi_{\Gamma,X}= \sum_{\Gamma^\mr{dec}\in \mr{Dec}(\Gamma)} \Phi_{\Gamma^\mr{dec},X} ,
$$
where in the summand on the r.h.s., we have restrictions on images of vertices of $\Gamma$ as prescribed by the decoration, and we only select either cut or uncut piece of each Green's function. 

Thus, the l.h.s. of (\ref{gluing of Z^pert}) can be written as
\begin{equation}\label{Z_X^pert via decorated graphs}
    Z_X^\mr{pert}= \det (K_X)^{-\frac12}\cdot \sum_{\Gamma^\mr{dec}} \frac{\hbar^{-\chi(\Gamma)}}{|\mr{Aut}(\Gamma^\mr{dec})|}\Phi_{\Gamma^\mr{dec},X}.
\end{equation}
where on the right we are summing over all Feynman graphs with all possible decorations.

The r.h.s. of (\ref{gluing of Z^pert}) is:
\begin{multline}\label{gluing of Z^pert eq2}
    \ll \det(K_{X',Y})^{-\frac12}\det(K_{X'',Y})^{-\frac12} e^{-\frac{1}{\hbar}S_Y^\mr{int}(\phi_Y)} \cdot \\
    \cdot \sum_{\Gamma',\Gamma''} \frac{\hbar^{-\chi(\Gamma'\cup \Gamma'')}}{|\mr{Aut}(\Gamma'\cup \Gamma'')|} \Phi_{\Gamma', (X',Y)}(\phi_Y) \Phi_{\Gamma'', (X'',Y)}(\phi_Y) \gg_Y
    = \\
    \det (K_X)^{-\frac12} 
    \cdot \left\langle 
    e^{-\frac{1}{\hbar}S_Y^\mr{int}(\phi_Y)} 
    \cdot \sum_{\Gamma',\Gamma''} \frac{\hbar^{-\chi(\Gamma'\cup \Gamma'')}}{|\mr{Aut}(\Gamma'\cup \Gamma'')|} \Phi_{\Gamma', (X',Y)}(\phi_Y) \Phi_{\Gamma'', (X'',Y)}(\phi_Y)
    \right\rangle_Y,
\end{multline}
where $\ll \cdots \gg_Y\colon = \int_{F^Y} D\phi_Y\, e^{-\frac{1}{2\hbar}(\phi_Y,\DN_{Y,X}\phi_Y)} \cdots $ is the non-normalized Gaussian average w.r.t. the total DN operator; $\langle\cdots \rangle_Y$ is the corresponding normalized average.

The correspondence between (\ref{Z_X^pert via decorated graphs}) and (\ref{gluing of Z^pert eq2}) is as follows. Consider a decorated graph $\Gamma^\mr{dec}$ and form out of it subgraphs $\Gamma', \Gamma''$ in the following way.  Let us 
cut every cut edge in $\Gamma$ (except $Y-Y$ edges) into two, introducing two new boundary vertices. Then we collapse every edge between a newly formed vertex and a $Y$-vertex.
$\Gamma'$ is the subgraph of $\Gamma$ formed by vertices decorated by $X'$ and uncut edges between them, and those among the newly formed boundary vertices which are connected to an $X'$-vertex by an edge; $\Gamma''$ is formed similarly. 

Then the contribution of $\Gamma^\mr{dec}$ to (\ref{Z_X^pert via decorated graphs}) is equal to the contribution of a particular Wick pairing for the term in (\ref{gluing of Z^pert eq2}) corresponding to the induced pair of graphs $\Gamma',\Gamma''$, and picking a term in the Taylor expansion of $e^{-\frac{1}{\hbar} S^\mr{int}_Y(\phi_Y)}$ corresponding to $Y$-vertices in $\Gamma^\mr{dec}$. 
The sum over all decorated Feynman graphs in (\ref{Z_X^pert via decorated graphs}) recovers the sum over all pairs $\Gamma',\Gamma''$ and all Wick contractions in (\ref{gluing of Z^pert eq2}).
This shows Feynman-graph-wise the equality of (\ref{Z_X^pert via decorated graphs}) and (\ref{gluing of Z^pert eq2}). One can also check that the combinatorial factors work out similarly to the argument in \cite[Lemma 6.10]{KMW}.

\begin{example}
Figure \ref{fig: feynman graph cut} is an example of a decorated Feynman graph on $X$ (on the left; vertex decorations $X',Y,X''$ are according to the labels in the bottom) and the corresponding contribution to  (\ref{gluing of Z^pert eq2}) on the right. 

     \begin{figure}[H]
     \includegraphics[scale=0.75]{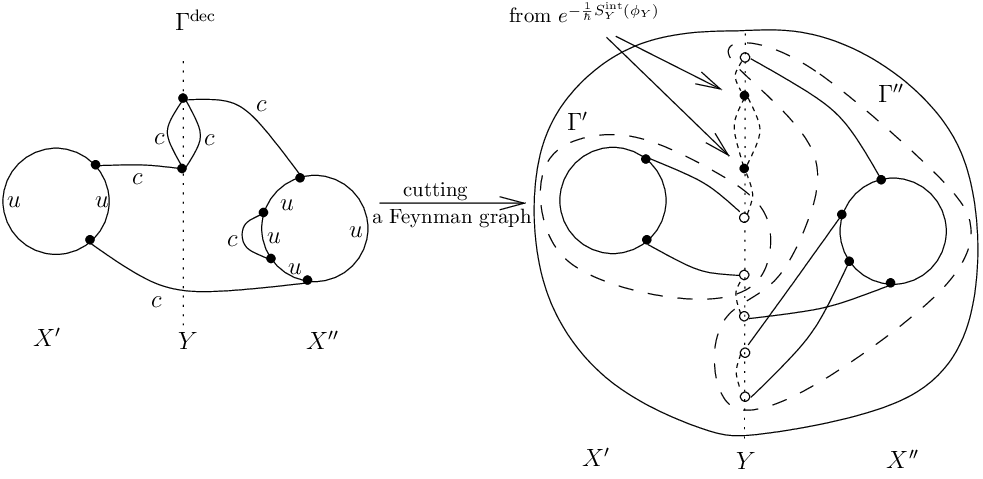}
     \caption{An example of cutting a Feynman graph.}
     \label{fig: feynman graph cut}
     \end{figure}
    Dashed edges on the right denote the Wick pairing for $\langle\cdots \rangle_Y$ and are decorated with $\DN_{Y,X}^{-1}$. Circle vertices are the boundary vertices of graphs $\Gamma',\Gamma''$ or equivalently the vertices formed by cutting the $c$-edges of $\Gamma^\mr{dec}$. 
\end{example}

\section{Path sum formulae for the propagator and determinant (Gaussian theory in the first quantization formalism)}
\label{s: path sum formulae for the propagator and determinant}
\subsection{Quantum mechanics on a graph}\label{ss: QM on a graph}
 Following the logic of Section \ref{sec: first quantization}, we now want to understand the kinetic operator $\Delta_X + m^2$ of the second quantized theory as the Hamiltonian of an auxiliary quantum mechanical system -- a quantum particle on the graph $X$.\footnote{
 This model of quantum mechanics on a graph -- as a model for the interplay between the operator and path integral formalisms -- was considered in \cite{trace,graph_QM},  see also \cite{Graph_Dirac}.
 }  
The space of states $\mathcal H_X$ for graph quantum mechanics on $X$  is $\mathbb C^V$, i.e. the space of $\mathbb C$-valued functions on $V$.
The graph Schr\"odinger equation\footnote{
Here we are talking about the Wick-rotated Schr\"odinger equation (i.e. describing quantum evolution in imaginary time), or equivalently the heat equation.
} on $X$
is
\begin{equation}\label{eq:GraphSchro}
\frac{\partial}{\partial t}\ket{\psi(t,v)}(t,v)= -\left (   \Delta_X+ m^2 \right )\ket{\psi (t,v)},
\end{equation}
where $\ket{\psi(t,v)}$ is a (time-dependent) state, i.e. a vector in $\mathbb C^V$.
The explicit solutions to \eqref{eq:GraphSchro} 
are given by
\begin{equation}
\ket{\psi(t_f)}= e^{-(t_f-t_0)(\Delta_X+m^2)} \ket{\psi(t_0)}. \label{eq:heat QM}
\end{equation}
One can explicitly solve equation \eqref{eq:heat QM} by summing over certain paths on $X$, see equations \eqref{eq: heat kernel hesitant paths}, (\ref{heat kernel path sum for regular graph}), (\ref{heat kernel path sum for irregular graph}) below, in a way reminiscent of Feynman's path integral.\footnote{This analogy is discussed in more detail in \cite{trace, graph_QM}.} 
This graph quantum mechanics is the first step of our
first quantization approach to QFT on a graph. 

\subsection{Path sum formulae on closed graphs} \label{subsec: Path Sum}
\subsubsection{Paths and h-paths in graphs}
 We start with some terminology. A \emph{path} $\gamma$ from a vertex $u$ to a vertex $v$ of a graph $X$ is a sequence $$\gamma=(u=v_0,e_0,v_1,\ldots, e_{k-1},v_k=v)$$ where $v_i$ are vertices of $X$ and $e_i$ is an edge between $v_{i}$ and $v_{i+1}$.\footnote{ 
 For simplicity of the exposition, we assume that the graph $X$ has no short loops. The generalization allowing short loops is straightforward: in the definition of a path and h-path, the edges traversed $e_i$ are not allowed to be short loops (and in the formulae involving the valence of a vertex, it should be replaced with valence excluding the contribution of short loops). This is ultimately due to the fact that short loops do not contribute to the graph Laplacian $\Delta_X$.}
 We denote $V(\gamma)$ the ordered collection $(v_0,v_1, \ldots, v_k)$ of vertices of $\gamma$.  
 We call $l(\gamma) = k$ the length of the path, and denote $P_X^k(u,v)$ the set of paths in $X$ of length $k$ from a vertex $u$ to a vertex $v$. We denote by $P_X(u,v) = \cup_{k=0}^\infty P_X^k(u,v)$ the set of paths of any length from $u$ to $v$. We also denote by $P_X^k =\cup_{u,v \in X} P_X^k(u,v)$, and $P_X = \cup_{k=0}^\infty P_X^k$ the sets of paths between any two vertices of $X$. 
 
 Below we will also need a  variant 
 of this notion that we call \emph{hesitant paths}. Namely, a \emph{hesitant path} 
 (or ``h-path'')
 from a vertex $u$ to a vertex $v$ is a sequence 
 $$\tilde\gamma = (u=v_0,e_0,v_1,\ldots,e_{k-1},v_k=v),$$
 but now we allow the possibility that $v_{i+1} = v_i$, in which case $e_i$ is allowed to be any edge starting at $v_i = v_{i+1}$. In this case we say that $\tilde\gamma$ \emph{hesitates} at step $i$. If $v_{i+1} \neq v_i$, then we say that $\tilde\gamma$ jumps at step $i$.  As before, we say that such a path has length $l(\tilde\gamma) = k$, and we introduce the notion of the \emph{degree} of a h-path as\footnote{Diverting slightly from the notation of \cite{Simone}.} 
\begin{equation}
    \deg(\tilde\gamma) = |\{i|v_i \neq v_{i+1}, 0 \leq i \leq l(\gamma) -1\}|,
\end{equation}
i.e. the degree is the number of jumps of a h-path. We denote by $$h(\tilde\gamma)  = |\{i|v_i = v_{i+1}, 0 \leq i \leq l(\gamma) -1\}|$$
the number of hesitations of $\tilde\gamma$. Obviously $l(\tilde\gamma) = \deg(\tilde\gamma) + h(\tilde\gamma)$. 
We denote the set of h-paths from $u$ to $v$ by $\Pi_X(u,v)$, and the set of length $k$ hesitant paths by $\Pi^k_X(u,v)$. There is an obvious concatenation operation 
\begin{equation}
    \begin{array}{ccc}
    \Pi^k_X(u,v) \times \Pi^l_X(v,w) &\to & \Pi^{k+l}_X(u,w)\\ 
    (\tilde\gamma_1, \tilde\gamma_2) &\mapsto& \tilde\gamma_1 * \tilde\gamma_2
    \end{array}
\end{equation}
Observe that for every h-path $\tilde\gamma$ there is a usual (``non-hesitant'') path $\gamma$ of length $l(\gamma) = \deg(\tilde\gamma)$ given by simply forgetting repeated vertices, giving a map $P\colon\Pi(u,v) \twoheadrightarrow P(u,v)$. See Figure \ref{fig:paths}.

\begin{figure}[H]
    \centering
    \begin{subfigure}[t]{0.45\textwidth}
    \vskip 0pt
    \begin{tikzpicture}[scale=1.5, every edge quotes/.style = {auto}]
        \draw (0,0) -- (1,0) -- (2,0);
        \draw (0,1) -- (1,1) -- (2,1);
        \draw (0,2) -- (1,2) -- (2,2);
        \draw (0,0) -- (0,1) -- (0,2);
        \draw (1,0) -- (1,1) -- (1,2);
        \draw (2,0) -- (2,1) -- (2,2);
        \foreach \i in {0,...,2}
        \foreach \j in {0,...,2}
        \draw[fill=black] (\i,\j) circle (2pt); 
        \node[coordinate, label=above:{$A$}] (A) at (0,2) {$v_0$};
        \node[coordinate, label=above:{$B$}] (B) at (1,2) {$v_1$};
        \node[coordinate, label=above left:{$C$}] (C) at (1,1) {$v_2$};
        \node[coordinate, label=above left:{$D$}] (D) at (1,0) {$v_3$};
        \node[coordinate, label=above right:{$E$}] (E) at (2,0) {$v_4$};
        \draw[ultra thick, red] (A) edge["$e_0$"] (B) (B) edge["$e_1$"] (C) (C) edge["$e_2$"] (D) (D) edge["$e_3$"] (E);
    \end{tikzpicture}
      \end{subfigure}
\begin{subfigure}[t]{0.45\textwidth}
\vskip 0pt
    \begin{tikzpicture}[scale=1.5, every edge quotes/.style = {auto}]
        \draw (0,0) -- (1,0) -- (2,0);
        \draw (0,1) -- (1,1) -- (2,1);
        \draw (0,2) -- (1,2) -- (2,2);
        \draw (0,0) -- (0,1) -- (0,2);
        \draw (1,0) -- (1,1) -- (1,2);
        \draw (2,0) -- (2,1) -- (2,2);
        \foreach \i in {0,...,2}
        \foreach \j in {0,...,2}
        \draw[fill=black] (\i,\j) circle (2pt); 
        \node[coordinate, label=above:{$A$}] (A) at (0,2) {$v_0$};
        \node[coordinate, label=above:{$B$}] (B) at (1,2) {$v_1$};
        \node[coordinate, label=above left:{$C$}] (C) at (1,1) {$v_2$};
        \node[coordinate, label=above left:{$D$}] (D) at (1,0) {$v_3$};
        \node[coordinate, label=above right:{$E$}] (E) at (2,0) {$v_4$};
        \node[coordinate] (F) at (0,1.3) {};
        \node[coordinate] (G) at (1.7,2) {};
        \node[coordinate] (H) at (1.7,1) {};
        \node[coordinate] (I) at (0.3,1) {};
        \node[coordinate] (J) at (0.3,0) {};
        \draw[ultra thick, red] (A) edge["$e_0$"] (B) edge["$f_0$"] (F)
        (B) edge["$e_1$"] (C) edge["$f_1$"] (G)  (C) edge["$e_2$"] (D) edge["$f_2$"] (H) edge["$f_3$"] (I) (D) edge["$e_3$"] (E) edge["$f_4$"] (J);
    \end{tikzpicture}
\end{subfigure}
  \caption{Left: The path $\gamma = (A,e_0,B,e_1,C,e_2,D,e_3,E)$ from $A$ to $E$ with length $l(\gamma) = 4$. Right: the h-path $\tilde\gamma = (A,f_0,A,e_0,B,f_1,B,e_1,C,f_2,C,f_3,C,e_2,D,f_4,D,e_3,E)$ with $l(\tilde\gamma) = 9, \deg(\tilde\gamma) = 4, h(\tilde\gamma) = 5$ and $P(\tilde\gamma) = \gamma$. }    \label{fig:paths}
\end{figure}
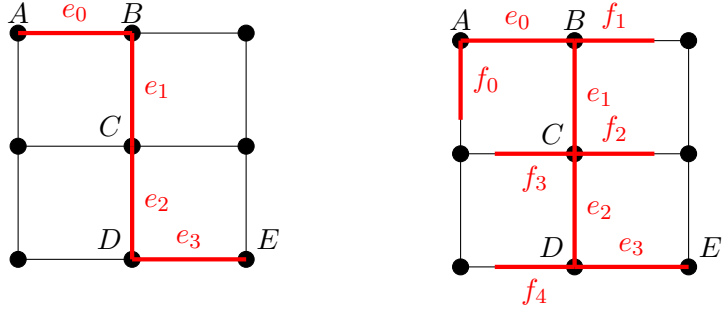
A (hesitant) path is called \emph{closed} if $v_k = v_0$, i.e. the first and last vertex agree. The cyclic group $C_k$ acts on closed paths of length $k$ by shifting the vertices and edges. We call the orbits of this group action \emph{cycles} (i.e. closed paths without a preferred start- or end-point), and denote them by $\Gamma_X$ for equivalence classes of h-paths, and $C_X$ for equivalence classes of regular paths. A cycle $[\tilde\gamma]$ is called \emph{primitive} if its representatives have trivial stabilizer under this group action. Equivalently, this means that there is no $k > 1$ and $\tilde\gamma'$ such that $$\tilde\gamma = \underbrace{\tilde\gamma' * \tilde\gamma' * \ldots * \tilde\gamma'}_{k \text{ times}}, $$
i.e. the cycle is traversed exactly once. In general, the order of the stabilizer of $\tilde\gamma$ is precisely the number of traverses. We will denote this number by $t(\tilde\gamma)$. Obviously, it is well-defined on cycles. 
\begin{example}
    In the $N=3$ circle graph, the two closed paths $\tilde\gamma_1 = (1,(12),2,(23),3,(31),1)$ and $\tilde\gamma_2 = (2,(23),3,(31),1,(12))$ define the same primitive cycle, while the closed path 
    $\tilde\gamma_3 = (1,(31),3,(23),2,(12),1)$ defines a different cycle (since the graph is traversed in a different order). The closed (hesitant) path $\tilde\gamma_4 = (1,(12),1,(12),1)$ is not primitive, since $\gamma_4 = (1,(12),1) * (1,(12),1)$. 
\end{example}
\subsubsection{h-path formulae for heat kernel, propagator and determinant} It is a simple observation that
\begin{equation}
    p_X^k(u,v):=|P_X^k(u,v)| = \langle u | (A_X)^k |v  \rangle,  \label{eq: length k paths}
\end{equation}
where $A_X$ denotes the adjacency matrix of the graph $X$, $| v\rangle$ denotes the state which is 1 at $v$ and vanishes elsewhere, and $$\langle u | A | v \rangle = A_{uv}$$ denotes the $(u,v)$-matrix element of the operator $A$ (in the bra-ket notation for the quantum mechanics on $X$).  
We consider the heat operator 
\begin{equation}
    e^{-t\Delta_X}\colon C^0(X) \to C^0(X),
\end{equation}
which is the propagator of the quantum mechanics on the graph $X$ \eqref{eq:heat QM}.

Suppose that $X$ is regular, i.e. all vertices have  the same valence $n$. Then $\Delta_X = n\cdot I - A_X$ and \eqref{eq: length k paths} implies that the heat kernel $  \langle u | e^{-t\Delta_X} | v \rangle $ is given by 
\begin{equation}\label{heat kernel path sum for regular graph}
    \langle u | e^{-t\Delta_X} | v \rangle = e^{-tn}\sum_{k=0}^\infty\frac{t^k}{k!}p_X^k(u,v) = e^{-tn} \sum_{\gamma \in P_X(u,v)}\frac{t^{l(\gamma)}}{l(\gamma)!} .
\end{equation}
One can think of the r.h.s. as a discrete analog of the Feynman path integral formula where one is integrating over all paths (see \cite{trace}).

For a general graph, one can derive a formula for the heat kernel in terms of h-paths, by using the formula $\Delta = d^T d$. 
Namely, one has (see \cite{Simone})
\begin{equation}
    \langle u | \Delta_X^k | v \rangle = \sum_{\tilde\gamma\in\Pi_X^k(u,v)}(-1)^{\deg(\tilde\gamma)}.
\end{equation}
This implies the following formula for the heat kernel: 
\begin{equation}
    \langle u | e^{-t\Delta_X} | v \rangle = \sum_{k=0}^\infty \frac{t^k}{k!}(-1)^k \sum_{\tilde\gamma \in \Pi^k_X(u,v)}(-1)^{\deg(\tilde\gamma)} = \sum_{\tilde\gamma \in \Pi_X(u,v)}\frac{t^{l(\tilde\gamma)}}{l(\tilde\gamma)!}(-1)^{h(\tilde\gamma)}. \label{eq: heat kernel hesitant paths}
\end{equation}
Here we have used that $l(\tilde\gamma) + \deg(\tilde\gamma) = h(\tilde\gamma) \text{ mod } 2$. 

Then we have the following h-path sum formula for the Green's function: 
\begin{lemma}
The Green's function $G_X$ is given by 
    \begin{equation}
    \left \langle u | G_X |v \right\rangle = m^{-2}\sum_{k=0}^\infty (m^{-2})^k\sum_{\tilde\gamma\in\Pi^k_{X}}(-1)^{h(\tilde\gamma)} = m^{-2}\sum_{\tilde\gamma\in\Pi_X(u,v)}(m^{-2})^{l(\tilde\gamma)}(-1)^{h(\tilde\gamma)}.\label{eq: hesitant path sum G}
\end{equation}
\end{lemma}
\begin{proof}
By expanding $m^2G_X = (m^{-2}K_X)^{-1} = (1 + m^{-2}\Delta_X)^{-1}$ in powers of $m^{-2}$ using the geometric series,\footnote{This series converges absolutely if the operator norm of $m^{-2}\Delta_X$ is less that one, or equivalently $m^2 > \lambda_{max}(\Delta_X)$, i.e. $m^2$ is larger than the largest eigenvalues of $\Delta_X$.} we obtain 
\begin{equation}
    \left\langle u \left| (1+m^{-2}\Delta_X)^{-1} \right|v \right\rangle = \sum_{\tilde\gamma \in \Pi_X(u,v)}(-m^{-2})^{l(\tilde\gamma)}(-1)^{\deg (\tilde\gamma)}, \label{eq: greens hesitant paths}
\end{equation}
which proves \eqref{eq: hesitant path sum G}. Alternatively, one can prove \eqref{eq: hesitant path sum G} by integrating the heat kernel $e^{-tK_X} = e^{-tm^2}e^{-t\Delta_X}$ for $K_X$ over the time parameter $t$ and 
 using the Gamma function identity
$$ \int_0^\infty dt\ e^{-tm^2} \frac{t^k}{k!} = (m^{-2})^{k+1}.$$
\end{proof}
In equation \eqref{eq: hesitant path sum G}, we see two slightly different ways of interpreting the path sum formula. In the middle we see that when expanding in powers of $m^2$, the coefficient of $m^{-2(k+1)}$ is given by a signed count of h-path of length $k$, and that the sign is determined by the number of hesitations. On the right hand side we interpret the propagator as a weighted sum over all h-paths, in accordance with the first quantization picture. 

We have the following formula for the  determinant of the kinetic operator (normalized by $1/m^2$) in terms of closed h-paths or h-cycles:
\begin{lemma}
The determinant of $K_X/m^2$ is given by
    \begin{equation}
    \begin{aligned}
\log\det\left(\frac{K_X}{m^{2}}\right) & = -\sum_{v\in X}\sum_{ \tilde\gamma \in \Pi^{\geq 1}_X(v,v)}\frac{(m^{-2})^{l(\tilde\gamma)}}{l(\tilde\gamma)}(-1)^{h(\tilde\gamma)} \\
&= - \sum_{[\tilde\gamma] \in \Gamma^{\geq 1}_X} \frac{(m^{-2})^{l(\tilde\gamma)}}{ {t(\tilde\gamma) }}(-1)^{h(\tilde\gamma)}. 
\end{aligned}
\label{eq: det hesitant paths}
\end{equation}
\end{lemma}
\begin{proof}
Expand\footnote{Again, this power series converges absolutely for $m^2 > \lambda_{max}(\Delta_X)$. } \begin{align*}\log\det\left(\frac{K_X}{m^2}\right) &= \mathrm{tr}\log (1 + m^{-2}\Delta_X) = - \sum_{v\in X}\sum_{k=0}^\infty \frac{(-m^{-2})^k}{k} \langle v | \Delta_X^k | v\rangle, 
\end{align*}

which implies \eqref{eq: det hesitant paths}.
\end{proof}

 Note that in the expression in the middle of \eqref{eq: det hesitant paths}, we are summing over h-paths of length at least 1 with a fixed starting point. To obtain the right hand side, we sum over orbits of the group action of $C_k$ on closed paths of length $k$, the size of the orbit of $\tilde\gamma$ is exactly $l(\tilde\gamma)/t(\tilde\gamma)$.

\begin{remark}
Both h-paths and paths form monoids  
w.r.t. concatenation,
with $P$ a monoid homomorphism. A map $s$ from a monoid to $\RR$ or $\CC$ is called \emph{multiplicative} if it is a homomorphism of monoids, i.e. 
\begin{equation}
        s(\tilde\gamma_1*\tilde\gamma_2) = s(\tilde\gamma_1)s(\tilde\gamma_2).
    \end{equation}
    Notice that in the path sum expression for the propagator \eqref{eq: hesitant path sum G}, we are summing over h-paths $\tilde\gamma$ with the weight \begin{equation}
        s(\tilde\gamma)  := (m^{-2})^{l(\tilde\gamma)}(-1)^{h(\tilde\gamma)}.
    \end{equation}
    Below it will be important that this weight is in fact multiplicative, which is obvious from the definition. 
\end{remark}
\begin{remark}
Using multiplicativity of $s$, we can resum over iterates of primitive cycles to rewrite the right hand side of \eqref{eq: det hesitant paths}: 
\begin{align*}
    \log\det\left(\frac{K_X}{m^{2}}\right) = -  \sum_{\substack{[\tilde\gamma] \in \Gamma^{\geq 1}_X \\ \tilde\gamma \text { primitive }}} \sum_{k\geq 1} \frac{s(\tilde\gamma)^k}{k} = \sum_{\substack{[\tilde\gamma] \in \Gamma^{\geq 1}_X \\ \tilde\gamma \text { primitive }}} \log\left(1 - m^{-2l(\tilde\gamma)}(-1)^{h(\tilde\gamma)}\right).
\end{align*}
\end{remark}
\subsubsection{Resumming h-paths. Path sum formulae for propagator and determinant}
Summing over the fibers of the map $P \colon \Pi_X(u,v) \twoheadrightarrow P_X(u,v)$, we can rewrite \eqref{eq: greens hesitant paths} as a path sum formula as follows: 

\begin{lemma}
If $m^2 > \val(v)$ for all $v \in X$, we have 
    \begin{equation}
        \left\langle u \left| (1+m^{-2}\Delta_X)^{-1} \right|v \right\rangle = m^2\sum_{\gamma \in P_X(u,v)}\prod_{v_i\in V(\gamma)} \frac{1}{m^2+ \val(v_i)}. \label{eq: prop inverse laplace}
    \end{equation}
\end{lemma}
\begin{proof}
    For a path $\gamma \in P^k_X(u,v)$, the fiber $P^{-1}(\gamma)$ consists of h-paths $\tilde\gamma$ which hesitate an arbitrary number $j_i$ of times at every vertex $v_i$ in $V(\gamma)$. For each vertex $v_i$, there are $\val(v_i)^{j_i}$ possibilities for a path to hesitate $j_i$ times at $v_i$. 
    The length of such a h-path is $l(\tilde\gamma) = k + j_0 + \ldots + j_k$ and its degree is $\deg(\tilde\gamma) = k$, hence we can rewrite equation \eqref{eq: greens hesitant paths} as 
    \begin{align*}
         &\sum_{\tilde\gamma \in \Pi_X(u,v)}(-m^{-2})^{l(\tilde\gamma)}(-1)^{\deg (\tilde\gamma)} \\
         &= \sum_{k=0}^\infty\sum_{\gamma \in P_X^k(u,v)}\sum_{j_0,\ldots,j_{k}=0}^\infty \val(v_0)^{j_0} \cdot \ldots \cdot \val(v_k)^{j_k} (-m^{-2})^{k+j_0 + \ldots + j_k}(-1)^k \\
         &=  \sum_{k=0}^\infty\sum_{\gamma \in P_X^k(u,v)} (m^{-2})^k \sum_{j_0,\ldots,j_k}^\infty \val(v_0)^{j_0} \cdot \ldots \cdot \val(v_k)^{j_k} (-m^{-2})^{j_0 + \ldots + j_k} \\
         &=  \sum_{\gamma \in P_X(u,v)}m^2\prod_{v_i\in V(\gamma)} \frac{m^{-2}}{1+m^{-2}\cdot\val(v_i)}. 
    \end{align*}
\end{proof}
\begin{corollary}
The Green's function of the kinetic operator has the expression 
\begin{equation}
    \langle u | G_X | v\rangle = \sum_{\gamma \in P_X(u,v)}\prod_{v_i\in V(\gamma)} \frac{1}{m^2 + \val(v_i)}. \label{eq: greens corollary}
\end{equation}    
In particular, if $X$ is regular of degree $n$, then 
\begin{equation}
    \langle u | G_X | v\rangle = \sum_{\gamma \in P_X(u,v)} \left(\frac{1}{m^2 + n}\right)^{l(\gamma)+1} = \frac{1}{m^2 + n}\sum_{k=0}^\infty p^k_X(u,v) (m^2 + n)^{-k}. \label{eq: greens corollary 2}
\end{equation}
\end{corollary}

To derive a path sum formula for the determinant, we use a slightly different idea, that also provides an alternative proof of the resummed formula for the propagator. Consider the operator $\Lambda$ which acts on $C^0(X)$ diagonally in the vertex basis and sends $|v\rangle \mapsto (m^2 + \val(v)) |v\rangle$, that is, 
$$\Lambda = \mathrm{diag}(m^2 + \val(v_1),\ldots, m^2 + \val(v_N))$$
in the basis of $C^0(X)$ corresponding to an enumeration $v_1,\ldots,v_N$ of the vertices of $X$. Then, consider the ``normalized'' kinetic operator 
\begin{equation}
    \tilde{K}_X = \Lambda^{-1}K_X = I - \Lambda^{-1}A_X,
\end{equation}
with $A_X$ the adjacency matrix of the graph. Then, we have the simple generalization of the observation that matrix elements of the $k$-th power of the adjacency matrix $A_X$ count paths of length $k$ (see \eqref{eq: length k paths}), namely, matrix elements of $ (\Lambda^{-1}A_X)^k \Lambda^{-1}$ count paths weighted with 
\begin{equation}\label{w}
    w(\gamma):= \prod_{v\in V(\gamma)} \frac{1}{m^2 + \val(v)}. 
\end{equation}
Then, we immediately obtain 
\begin{multline}\label{G path-sum via Lambda}
    \langle u | G_X |v \rangle = \langle u | \tilde{K}_X^{-1}\Lambda^{-1} |v \rangle = \sum_{k=0}^\infty \left\langle u \left| (\Lambda^{-1}A_X)^k\Lambda^{-1} \right| v \right\rangle \\
    = \sum_{k=0}^\infty\sum_{\gamma \in P^k_X(u,v)}w(\gamma) ,
\end{multline}
which is \eqref{eq: greens corollary}. 

For the determinant, we have the following statement: 
\begin{proposition}\label{prop: determinant path sums}
    The determinant of the normalized kinetic 
    operator has the expansions 
    \begin{align}
        \log \det \tilde K_X &= -\sum_{v \in X} \sum_{k=0}^\infty \sum_{\gamma \in P^k_X(v,v)} \frac{w'(\gamma)}{l(\gamma)} \label{eq: det regular paths} \\
        &= -\sum_{[\gamma] \in C^{\geq 1}_X} \frac{w'(\gamma)}{t(\gamma)}, \label{eq: det regular cycles}
    \end{align}
    where for a closed path $\gamma \in P^k_X(v,v)$, $w'(\gamma) = w(\gamma)\cdot (m^2 + \val(v)$.\footnote{Note that this is well-defined on a cycle, we are simply taking the product over all vertices in the path but without repeating the one corresponding to start-and endpoint.}
\end{proposition}
\begin{proof}
    To see \eqref{eq: det regular paths}, we simply observe 
    $$\log\det\tilde{K}_X = \mathrm{tr}\log (1 - \Lambda^{-1}A_X) = - \mathrm{tr} \sum_{k=0}^\infty \frac{(\Lambda^{-1}A_X)^k}{k} = - \sum_{v \in V} \sum_{k=0}^\infty \sum_{\gamma \in P^k_X(v,v)}\frac{w'(\gamma)}{k}.$$
    To see the second formula \eqref{eq: det regular cycles}, one sums over orbits of the cyclic group action on closed paths.
\end{proof}
In particular, for regular graphs we obtain a formula also derived in \cite{Simone}: 
\begin{corollary}
    If $X$ is a regular graph, then 
\begin{equation}\label{det K path-sum for X regular}
    \log\det\tilde K_X = -\sum_{k=1}^\infty\sum_{[\gamma]\in C^k_X}\frac{(m^2 + n)^{-k}}{t(\gamma)}.
\end{equation}
\end{corollary}
Another corollary is the following first quantization formula for the partition function: 
\begin{thm}[First quantization formula for Gaussian theory on closed graphs]
    The partition function of the Gaussian theory on a closed graph can be expressed by 
    \begin{equation}
        \log Z_X = \frac12\left(\sum_{[\gamma]\in C^{\geq 1}_X}\frac{w'(\gamma)}{t(\gamma)} - \sum_{v\in X}\log (m^2 + \val(v))\right).
    \end{equation}
\end{thm}
I.e., the logarithm of the partition function is given, up to the ``normalization'' term $- \sum_{v\in X}\log (m^2 + \val(v))$, by summing over all cycles of length at least 1, dividing by automorphisms coming from orientation reversing and multiple traversals. 
\begin{proof}
    We have $$\log Z_X = -\frac12\log\det K_X = -\frac{1}{2}(\log\det \tilde{K}_X + \log \det \Lambda),$$
    from where the theorem follows by Proposition \ref{prop: determinant path sums}.
\end{proof}
\begin{remark}
    Notice that the weight $w(\gamma)$ of the resummed formula \eqref{eq: greens corollary} is not multiplicative: if $\gamma_1 \in P_X(u,v)$ and $\gamma_2\in P_X(v,w)$ then 
    $$\prod_{v_i\in \gamma_1} \frac{1}{m^2 + \val(v_i)}\prod_{v_i\in \gamma_2} \frac{1}{m^2 + \val(v_i)} = \frac{1}{m^2+\val(v)}\prod_{v_i\in \gamma_1*\gamma_2} \frac{1}{m^2 + \val(v_i)} ,$$
    since on the left hand side the vertex $v$ appears twice. 
\end{remark}
\begin{table}[H]
    \centering
    \bgroup
    \def\arraystretch{1.5}
    \begin{adjustbox}{max width=1.1\textwidth,center}
    \begin{tabular}{c|c|c}
        Object & h-path sum & path sum \\
        \hline $\langle u|G_X | v \rangle $ & $m^{-2}\sum_{\tilde\gamma\in\Pi_X(u,v)}s(\tilde\gamma)$ (Eq. \eqref{eq: greens hesitant paths}) & $\sum_{\gamma \in P_X(u,v)}w(\gamma)$ (Eq. \eqref{eq: greens corollary}\\ 
        $\log \det m^{-2}K_X$ & $- \sum_{[\tilde\gamma] \in \Gamma^{\geq 1}_X} \frac{s(\tilde\gamma)}{ {t(\tilde\gamma) }}$ (Eq. \eqref{eq: det hesitant paths}) & \\
        $\log\det \tilde{K}_X$ &  &  $-\sum_{[\gamma] \in C^{\geq 1}_X} \frac{w'(\gamma)}{t(\gamma)} $ (Eq. \eqref{eq: det regular cycles}) \\ 
    \end{tabular}
    \end{adjustbox}
    \egroup
    \caption{Summary of path sum formulae, closed case.}
    \label{tab: closed path formulas}
\end{table}

\begin{remark}
    The sum over $k$ in (\ref{G path-sum via Lambda}), (\ref{eq: det regular paths}) is absolutely convergent for any $m^2>0$. The reason is that the matrix $a=\Lambda^{-1}A_X$ has spectral radius smaller than $1$ for $m^2>0$. This in turn follows from Perron-Frobenius theorem: Since $a$ is a nonnegative matrix, its spectral radius $\rho(a)$ is equal to its largest eigenvalue (also known as Perron-Frobenius eigenvalue), which in turn is bounded by the maximum of the row sums of $a$.\footnote{For any matrix $A$ with entries $a_{ij}$, $\lambda$ an eigenvalue of $A$ and $x$ an eigenvector for $\lambda$, we have 
    $$ |\lambda| = \frac{||\lambda x||_\infty}{||x||_\infty} \leq \sup_{||y||_\infty = 1} ||Ay||_\infty = \max_{i}\sum_j |a_{ij}|.$$
    Here $||x||_\infty = \max_i |x_i|$ denotes the maximum norm of a vector $x$.}  The sum of entries on the $v$-th row of $a$ is  $\frac{\val(v)}{m^2+\val(v)}<1$, which implies $\rho(a) < \max_v \frac{\val(v)}{m^2 + \val(v)}< 1$. 
    
    In particular, resummation from h-path-sum formula to a path-sum formula extends the absolute convergence region from $m^2>\lambda_{max}(\Delta_X)$ to $m^2>0$. 
\end{remark}

\subsubsection{Aside: path sum formulae for the heat kernel and the propagator -- ``1d gravity'' version.}
There is the following generalization of the path sum formula (\ref{heat kernel path sum for regular graph}) for the heat kernel for a not necessarily regular
graph $X$.
\begin{proposition}
\begin{equation}\label{heat kernel path sum for irregular graph}
\langle u| e^{-t\Delta_X} |v\rangle =\sum_{\gamma\in P_X(u,v)} W(\gamma;t),
\end{equation}
where the $t$-dependent weight for a path  $\gamma$ of length $k$ is given by an integral over a standard $k$-simplex of size $t$:
\begin{equation}\label{heat kernel path sum for irregular graph: weight}
W(\gamma;t)= 
\int_{\scriptsize
\begin{array}{c}
       t_0,\ldots,t_{k}> 0\\
     t_0+\cdots+t_{k}=t
\end{array}
} dt_1\cdots dt_{k}\, e^{-\sum_{i=0}^{k} t_i \val(v_i)},
\end{equation} 
where we denoted $v_0,\ldots,v_{k}$ the vertices along the path.
\end{proposition}
\begin{proof}
To prove this result, note that the Green's function $G_X$ as a function of $m^2$ is the Laplace transform $L$ of the heat kernel $e^{-t\Delta_X}$ as a function of $t$. Thus, one can recover the heat kernel as the inverse Laplace transform $L^{-1}$ of $G_X$.  Applying $L^{-1}$ to (\ref{eq: greens corollary}) termwise, we obtain  (\ref{heat kernel path sum for irregular graph}), (\ref{heat kernel path sum for irregular graph: weight}) (note that the product of functions $\frac{1}{m^2+\val(v)}$ is mapped by $L^{-1}$ to the convolution of functions $L^{-1}(\frac{1}{m^2+\val(v)})=e^{-t\val(v)}$).
\end{proof}
As a function of $t$, the weight (\ref{heat kernel path sum for irregular graph: weight}) is a certain polynomial in $t$ and $e^{-t}$ with rational coefficients (depending on the sequence of valences $\val(v_i)$). If all valences along $\gamma$ are the same (e.g. if $X$ is regular), then the integral over the simplex evaluates  to $W(\gamma;t)=\frac{t^k}{k!}e^{-t\cdot\mr{val}}$ -- same as the weight of a path in (\ref{heat kernel path sum for regular graph}).

Note also that integrating (\ref{heat kernel path sum for irregular graph: weight}) (multiplied by $e^{-m^2 t}$)  in $t$, we obtain an integral expression  for the weight (\ref{w}) of a path in the path sum formula for the Green's function:
\begin{equation}
w(\gamma)=\int_{t_0,\ldots,t_{k}> 0} dt_0\cdots dt_{k}\, e^{-\sum_{i=0}^{k} t_i (\val(v_i)+m^2)}.
\end{equation}
Here unlike (\ref{heat kernel path sum for irregular graph: weight}) the integral is over $\RR_+^{k+1}$, not over a $k$-simplex. 

Observe that the resulting formula for the Green's function
\begin{equation}
\langle u| G_X |v\rangle = \sum_{\gamma\in P_X(u,v)} \int_{t_0,\ldots,t_{k}> 0} dt_0\cdots dt_{k}\, e^{-\sum_{i=0}^{k} t_i (\val(v_i)+m^2)}
\end{equation}
bears close resemblance to the first quantization formula (\ref{G 1q with 1d gravity}), where the  proper times $t_0,\ldots,t_k$ should be though of as parametrizing the worldline metric field $\xi$ (and the path $\gamma$ is the field of the ``1d sigma model'').\footnote{ 
More explicitly, one can think of the worldline as a standard interval $[0,1]$ subdivided into $k$ sub-intervals by points $p_0=0<p_1<\cdots<p_{k-1}<p_k=1$ (we think of $p_i$ as moments when the particle jumps to the next vertex). Then one can think of $\{t_i\}$ as moduli of metrics $\xi$ on $[0,1]$ modulo diffeomorphisms of $[0,1]$ relative to (fixed at) the points $p_0,p_1,\ldots,p_k$.
}
We imagine the particle moving on $X$ along $\gamma$, spending time $t_i$ at the $i$-th vertex and making instantaneous jumps between the vertices, with the ``action functional'' 
\begin{equation}
\overline{S}^{1q}(\gamma,\{t_i\})=\sum_{i=0}^k t_i (\val(v_i)+m^2).
\end{equation}

\subsection{Examples}
\subsubsection{Circle graph, $N=3$.}
Consider again the 
circle graph of Example \ref{ex: circle graph} for $N = 3$ (Figure \ref{fig:N=3circle path examples }). 
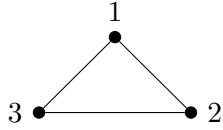
\begin{figure}[H]
    \centering
\begin{tikzpicture}
\node[draw, shape=circle, fill=black,label=above:{1},inner sep = 1.5pt] (1) at (0,1) {};
\node[draw, shape=circle, fill=black,label=left:{3},inner sep = 1.5pt] (3) at (-1,0) {} 
edge (1); 
\node[draw, shape=circle, fill=black,label=right:{2},inner sep = 1.5pt] (2) at (1,0) {} 
edge (3) edge (1); 
\end{tikzpicture}
\caption{The $N=3$ circle graph.}
 \label{fig:N=3circle path examples }
\end{figure}
Counting h-paths from $1$ to $2$, we see that there are no paths of length 0, a unique path $(1,(12),2)$ of length 1, and 5 paths of length 2: 
\begin{multline*}
    (1,(13),3,(23),2),\quad (1,(13),1,(12),2), \quad  (1,(12),1,(12),2), \\ (1,(12),2,(12),2), \quad (1, (12),2,(23),2).
\end{multline*} 
The first one comes with a + sign, since it has no hesitations, the other 4 paths hesitate once either at 1 or 2 and come with a minus sign, the overall count is therefore $-3$. Counting paths beyond that is already quite hard. Looking at the Greens' function, we have 
\begin{align*}
   \langle 1 | G_X | 2 \rangle  &= \frac{1}{m^2(m^2+3)} = m^{-4}\left(\frac{1}{1+3m^{-2}}\right) = m^{-4}\sum_{k\geq 0} (-3m^{-2})^k \\ 
    &= m^{-2}\left(0 \cdot m^0 + 1 \cdot m^{-2} - 3 m^{-4} + 9 m^{-6} + \ldots \right).
\end{align*}
Since the circle graph is regular, we can count paths from $u$ to $v$ by expanding in the parameter $\alpha^{-1} = \frac{1}{m^2+2}$. Here we observe that 
\begin{align*}
   \langle 1 |  G_X | 2 \rangle &= \frac{1}{(\alpha - 2)(\alpha + 1)} = \frac{1}{3\alpha}\left(\frac{1}{1 -2\alpha^{-1}} - \frac{1}{1 + \alpha^{-1}}\right) 
   \\ &= \frac{1}{3\alpha} \sum_{k\geq 0} (2^k - (-1)^k)\alpha^{-k} \\ 
    &= \alpha^{-1}\left(0 \cdot \alpha^0 + 1 \cdot \alpha^{-1} + 1 \alpha^{-2} + 3\alpha^{-3} + 5\alpha^{-4} + 11\alpha^{-5} + \ldots \right),
\end{align*}
and one can count explicitly that there is no path of length zero, a unique path (12) of length 1, a unique path (132) of length 2, 3 paths (1212),(1312),(1232) of length 3, 5 paths (12312), (13212), (13132), (12132), (13232) of length 4, and so on.\footnote{For brevity, here we just denote a path by its ordered collection of vertices, which determines the edges that are traversed.}
Similarly, we could expand 
\begin{align*}
    \langle 1 | G_X | 1 \rangle  &= \frac{m^2+1}{m^2(m^2+3)} = m^{-2}\left(1 + m^{-2}\right)\left(\frac{1}{1+3m^{-2}}\right) 
    \\ &= m^{-2}\left(1 + (-2)\sum_{k\geq 1}(-3)^{k-1}m^{-2k}\right) \\ 
    &= m^{-2} \left(1 \cdot m^{0} + (-2) \cdot m^{-2} + 6 \cdot m^{-4}  + (-18) \cdot m^{-6} + \ldots\right),
\end{align*}
which counts h-paths from vertex 1 to itself: a single paths of length 0, 2 length 1 paths which hesitate once at 1, two length 2 paths with 0 hesitations and 4 length 2 paths with 2 hesitations, and so on.  In terms of $\alpha = m^2 +2$, we get 
\begin{align*}
    \langle 1 | G_X | 1 \rangle  &= \frac{\alpha -1 }{(\alpha -2) (\alpha +1 )} \\ 
    &= \alpha^{-1}\sum_{k \geq 0}\frac{2^k + 2(-1)^k}{3}\alpha^{-k}  \\
    =& \alpha^{-1}\left(1\cdot \alpha^0 + 0 \cdot \alpha^{-1} + 2\cdot \alpha^{-2} + 2\cdot \alpha^{-3} + 6\cdot \alpha^{-4} + \ldots \right)  ,
\end{align*}
where we recognize the path counts from 1 to itself: A unique path (1) of length 0, no paths of length 1, two paths (121),(131) of length 2, 2 paths (1231),(1321) of length 3, and so on.   \\ 
The determinant is $\det K_X = m^2(m^2 + 3)^2$, so we have 
\begin{align*}
    \log\det m^{-2}K_X &= 2\log (1 + 3m^{-2}) = -2 \sum_{k\geq 1} \frac{(-3m^{-2})^k}{k} = \\
    & =
    6m^{-2} - 9m^{-4} + 18m^{-6} - \frac{81}{2}m^{-8} + \ldots 
\end{align*}
and we can see that rational numbers appear, because we are either counting paths with $\frac{1}{l(\tilde\gamma)}$, or cycles with $\frac{1}{t(\tilde\gamma)}$. Let us verify the cycle count for the first two powers of $m^2$. Indeed, there is a total of 6 cycles of length 1 that hesitate once, of the form $(1,(12),1)$, and similar. At length 2, there are 3 closed cycles that do no hesitate, of the form $(1,(12),2,(12),1)$. Then, there are three cycles that hesitate twice and are of the form $(1,(12),1,(31),1)$ (they visit both edges starting at a vertex). Moreover, at every vertex we have the cycles of the form $(1,(12),1,(12),1)$. There are a total of 6 such cycles, however, they come with a factor of 1/2 because those are traversed twice!  Overall we obtain $3 + 3 + \frac12 \cdot 6 = 9$ cycles (they all come with the same $+$ sign). \\ 
Finally, we can count cycles in $X$ by expanding the logarithm of the determinant in powers of $\alpha$: 
\begin{equation}
\begin{aligned}
        \log\det \tilde{K}_X &= \log \frac{(\alpha -2)(\alpha +1)^2}{\alpha^3} = \log 1 - 2 \alpha^{-1} + 2 \log 1+\alpha^{-1} \\
        &= -\sum_{k \geq 1} \frac{(2\alpha^{-1})^k}{k} - 2\sum_{k \geq 1}\frac{(-\alpha)^{-k}}{k} \\
    &= - \sum_{k \geq 1}\frac{2^k + 2(-1)^k}{k}\alpha^{-k} \\
   & =-\left( 0 \cdot \alpha^{-1} + 3\cdot \alpha^{-2} + 2\alpha^{-3} + \frac{9}{2}\alpha^{-4} + \ldots \right) .
\end{aligned}
\end{equation}
Counting cycles we see there are 0 cycles of length 1, 3 cycles of length 2, namely (121),(131),(232), 2 cycles of length 3, namely (1231), (1321). There are 3 primitive cycles of length 4 (those of the form (12131) and similar), and 3 cycles which are traversed twice ((12121) and similar), which gives $3 +\frac32 = \frac{9}{2}$. 
\subsubsection{Line graph, $N=3$}
Consider again the line graph of example \ref{ex: path graph N=3}. For instance, we have 
\begin{align*}
    \langle 1 | G_X | 3 \rangle &= \frac{1}{m^2(1+m^2)(3+m^2)} = m^{-2}\frac{1}{2}\left(\frac{1}{1+m^2}-\frac{1}{3+m^2}\right) \\
    &=\frac{m^{-4}}{2}\sum_{k=0}^\infty (-m^{-2})^k - (-3m^{-2})^{-k} = \frac{m^{-4}}{2}\sum_{k=0}^\infty((-1)^k - (-3)^k)m^{-2k} \\ 
    &= m^{-2}\left(0\cdot m^0 + 0 \cdot m^{-2} + 1 \cdot m^{-4} - 4 \cdot m^{-6} + 13 m^{-8} + \ldots\right)
\end{align*}
and indeed we can observe there are no h-paths from 1 to 3 of length 0 and 1, and there is a unique path $\gamma$ of length 2. At length 3, there are 4 different h-paths whose underlying path is $\gamma$ and who hesitate exactly once, there are a total of 1 + 2 + 1 possibilities to do so. At the next order, there are a total of 11 possibilities for $\gamma$ to hesitate twice, and two new paths of length 4 appear, explaining the coefficient 13.

The path sum (\ref{eq: greens corollary}) becomes
$$
\langle 1| G_X |3\rangle = \frac{1}{(1+m^2)^2 (2+m^2)}+ \frac{2}{(1+m^2)^3 (2+m^2)^2}+\cdots
$$
Here the numerator $1$ corresponds to the single path of length $2$, $(123)$; the numerator $2$ corresponds to the two paths of length $4$, $(12123), (12323)$. In fact, there are exactly $2^{l-1}$ paths $1\ra 3$ of length $2l$ for each $l\geq 1$, and along these paths  the 1-valent vertices (endpoints) alternate with the 2-valent (middle) vertex, resulting in
$$
\langle 1| G_X |3\rangle = \sum_{l\geq 1}\frac{2^{l-1}}{(1+m^2)^{l+1} (2+m^2)^l}.
$$

For the determinant $\det K_X = {m^2(m^2+1)(m^2+3)}$, we can give the hesitant cycles expansion 
\begin{align*} \log\det m^{-2}K_X &= \log (1 + m^{-2}) + \log (1 + 3m^{-2}) = - \sum_{k \geq 1} \frac{(-1)^k + (-3)^k}{k}m^{-2k}  \\
&= -\left(-4 m^{-2} + 5m^{-4} + \frac{28}{3}m^{-6} + \ldots \right).
\end{align*}
Here the first 4 is given by the four hesitant cycles of length 1. At length 2, we have the 4 iterates of length 1 hesitant cycles, contributing 2, a new hesitant cycle that hesitates twice at 2 (in different directions), and 2 regular cycles of length 2, for a total of $4\cdot\frac12 + 1 + 2 = 5.$ 
For the path sum we have 
\begin{multline*}
    -\log \det \tilde{K}_X = -\log \frac{m^2(m^2+3)}{(m^2+1)(m^2+2)} \\
    = -\log \left(1 - \frac{2}{(m^2+1)(m^2+2)}\right) = \sum_{k \geq 1} \frac{2^k}{k}(m^2+1)^{-k}(m^2+2)^{-k} ,
\end{multline*}
which means there are $2^k/k$ cycles (counted with $1/t(\gamma)$) of length $2k$. For instance, there are 2 cycles of length 2, namely (121) and (232). There is a unique primitive length 4 cycle, namely (12321), and the two non-primitive cycles (12121),(23232), which contribute $k = \frac12$, so we obtain $1 + 2\frac12 = 2$. There are 2 primitive length 6 cycles, namely (1232321) and (1212321), and the two non-primitive cycles (1212121), (2323232), contributing $\frac13$ each, for a total of $2 + \frac{2}{3}= \frac83$. At length 8 there are 3 new primitive cycles, the iterate of the length 4 cycle and the iterates of the 2 length 2 cycles for a total of $3 + \frac12 + 2 \cdot \frac 14 = 4 = 2^4/4$.
\subsection{Relative versions} \label{sec: path sums relative}
In this section we will study path-sum formulae for a graph $X$ relative to a boundary subgraph $Y$.
We will then give a path-sum proof of the gluing formula (Theorem \ref{thm: gluing prop and det}) in the case of a closed graph presented as a gluing  of subgraphs over $Y$. The extension to gluing of cobordisms is straightforward but notationally tedious. 

\subsubsection{h-path formulae for Dirichlet propagator, extension operator, Dirichlet-to-Neumann operator}
In this section we consider the path sum versions of the objects introduced in Section \ref{subsec:Gaussian-theory-relative-boundary}. 
Remember that, for a graph $X$ and a subgraph $Y$, we have the notations \eqref{K_X inverse as a block matrix}:
\begin{equation*}
    (K_X)^{-1} = \left( 
    \begin{array}{c|c}
         A& B \\ \hline
         C& D
    \end{array}
    \right)
\end{equation*}
and \eqref{K_X  as a block matrix}:
\begin{equation*}
    K_X = \left( 
    \begin{array}{c|c}
         \wh{A}=K_{X,Y}& \wh{B} \\ \hline
         \wh{C}& \wh{D}
    \end{array}
    \right).
\end{equation*}
We are interested in the following objects: 
\begin{itemize} 
\item The propagator with Dirichlet boundary conditions on $Y$, $G_{X,Y} = K_{X,Y}^{-1}$(cf. Section \ref{sec:rel part funct}).
\item The determinant of the kinetic operator $K_{X,Y}$ with Dirichlet boundary on $Y$ (cf. Section  \ref{sec:rel part funct}).
\item The combinatorial Dirichlet-to-Neumann operator $\DN_{Y,X} = D^{-1} \colon F_Y \to F_Y$ (cf. Section \ref{sec:DtoN}).
\item The extension operator $E_{Y,X} = BD^{-1} \colon F_Y \to F_X$ (cf. equation \eqref{E}).
\end{itemize}

\textbf{Propagator with Dirichlet boundary conditions.} 
For $u,v$ two vertices of $X\setminus Y$,
let us denote by $\Pi_{X,Y}(u,v)$ 
the set of h-paths from $u$ to $v$ that contain no vertices in $Y$ (but they may contain edges between $X\setminus Y$ and $Y$), and $\Pi^k_{X,Y}(u,v)$ the subset of such paths that have length $k$. Then we have the formula (\cite{Simone})
\begin{equation}
    \langle u | \Delta^k_{X,Y} | v \rangle = \sum_{\tilde\gamma \in \Pi_{X,Y}^k(u,v) }(-1)^{\deg(\tilde\gamma)}.
\end{equation}
In exactly the same manner as in the previous subsection, we can then prove 
\begin{equation}
\langle u | (1+m^{-2}\Delta_{X,Y})^{-1}| v \rangle = \sum_{\tilde\gamma \in \Pi_{X,Y}(u,v)} s(\tilde{\gamma}) 
\end{equation}
and therefore 
\begin{equation}
    \langle u | G_{X,Y} |v \rangle = m^{-2} \sum_{\tilde\gamma\in\Pi_{X,Y}(u,v)}s(\tilde{\gamma}). 
    \label{eq: G rel hesitant paths}
\end{equation}

\textbf{Determinant of relative kinetic operator.}
In the same fashion, we obtain the formula 
\begin{equation}
    \log\det\left(\frac{K_{X,Y}}{m^{2}}\right) = -\sum_{\tilde\gamma\in C^{\geq 1}_{X,Y} } \frac{s(\tilde{\gamma})}{t(\tilde\gamma)},
    \label{eq: det K rel hesitant paths}
\end{equation}
where we have introduced the notation $C^{\geq 1}_{X,Y}$ for cycles corresponding to closed h-paths in $X\setminus Y$ that may use edges between $X\setminus Y$ and $Y$.

\textbf{Dirichlet-to-Neumann operator.}
Notice  also that as a submatrix of $K_X^{-1}$, we have the following path sums for $D$ (here $u,v\in Y$): 
\begin{equation}
   \langle u | D |v \rangle = m^{-2}\sum_{\tilde\gamma \in \Pi_X(u,v)} s(\tilde{\gamma}). 
    \label{eq: D path sums}
\end{equation}

For $u,v \in Y$, we introduce the notation $\Pi''_{X,Y}(u,v)$ to be those h-paths from $u$ to $v$ containing exactly two vertices in $Y$, i.e. the start- and end-points. We define the operator $D'\colon C^0(Y) \to C^0(Y)$ given by summing over such paths (see Figure \ref{fig: path Y to Y})
\begin{equation}
    \langle u | D' | v \rangle \colon = \sum_{\tilde\gamma \in \Pi''_{X,Y}(u,v)} s(\tilde{\gamma}). 
    \label{eq: D prime hesitant paths }
\end{equation}
\begin{figure}[H]
    \centering
    \begin{subfigure}{0.45\textwidth}
            \includegraphics[scale=0.6]{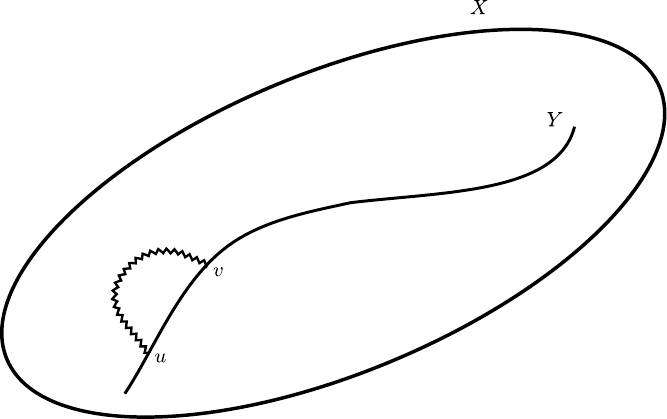}
\caption{h-paths in $\Pi''_X(u,v)$, contributing to $ \langle u | D' | v\rangle$. }
    \label{fig: path Y to Y}
    \end{subfigure}
       \begin{subfigure}{0.45\textwidth}
            \includegraphics[scale=0.6]{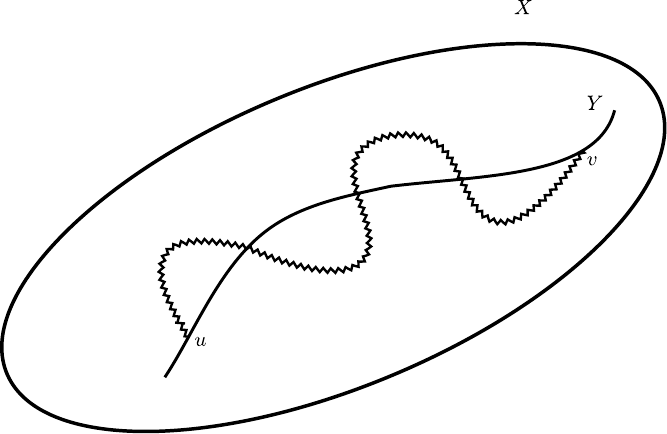}
\caption{h-paths contributing to $\langle u | (D')^k| v\rangle$. }
    \label{fig: path Y to Y several}
    \end{subfigure}
\end{figure}
Notice that $\langle u | (D')^k |v \rangle $ is given by summing over paths which cross the interface $Y$ exactly $k-1$ times between the start- and the end-point (see Figure \ref{fig: path Y to Y several}). 
Since the summand is multiplicative, we can therefore rewrite $D$ as 
\begin{equation}
    D = m^{-2}\sum_{k\geq 0} (D')^k = m^{-2}(I - D')^{-1}.
\end{equation}
Therefore the Dirichlet-to-Neumann operator is given by the formula 
\begin{equation}
    \DN_{Y,X} = D^{-1} = m^{2}(I - D'). 
    \label{eq: D to N rel h paths}
\end{equation}

\textbf{Extension operator.} 
Finally, we give a path sum formula for the extension operator. To do so we introduce the notation $\Pi'_{X,Y}(u,v)$ for h-paths that start at a vertex $u \in X \setminus Y$, end at a vertex $v \in Y$, and contain only a single vertex on $Y$, i.e. the end-point. 

 \begin{lemma}
     The extension operator can be expressed as 
  \begin{eqnarray}
       BD^{-1}(u,v) = E_{Y,X}(u,v) = \sum_{\tilde\gamma \in \Pi'_{X,Y}(u,v)}s(\tilde{\gamma}) .\label{eq: h path sum extension operator}
     \end{eqnarray}
 \end{lemma} 
 \begin{proof}
    We will prove that composing with $D$ we obtain $B$. Indeed, denote the right hand side of equation \eqref{eq: h path sum extension operator} by $\tilde{B}$. Then, using the h-path sum expression for $D$  \eqref{eq: D path sums} we obtain 
    \begin{equation*}
        \tilde{B}D = m^{-2}\sum_{v \in Y} \left(\sum_{\tilde\gamma \in \Pi'_{X,Y}(u,v)}s(\tilde{\gamma})\right)\left( \sum_{\tilde\gamma \in \Pi_X(v,w)}s(\tilde{\gamma})\right).
    \end{equation*}
    Using multiplicativity, we can rewrite this as 
    \begin{equation*}
        m^{-2}\sum_{(\tilde\gamma_1,\tilde\gamma_2) \in \sqcup_{v\in Y} \Pi'_{X,Y}(u,v) \times \Pi_X(v,w)}s(\tilde\gamma_1 * \tilde\gamma_2).
    \end{equation*}
    Now the argument finishes by observing that any h-path $\tilde\gamma$ from a vertex $u$ in $X\setminus Y$ to a vertex $w$ in $Y$ can be decomposed as follows. Let $v \in Y$ be the first vertex of $Y$ that appears in $\tilde\gamma$ and denote $\tilde\gamma_1$ the part of the path before $v$, and $\tilde\gamma_2$ the rest. Then $\tilde\gamma = \tilde\gamma_1 * \tilde\gamma_2$ and $\tilde\gamma_1 \in \Pi'_X(u,v)$. This decomposition is the inverse of the composition map
    \begin{equation*}
    \begin{array}{ccc}
       \sqcup_{v\in Y} \Pi'_{X,Y}(u,v) \times \Pi_X(v,w) &\to& \Pi_X(u,w) \\
       (\tilde\gamma_1,\tilde\gamma_2) &\mapsto& \tilde\gamma_1 * \tilde\gamma_2
       \end{array}
    \end{equation*}
    which is therefore a bijection. In particular, we can rewrite the expression above as 
    \begin{equation*}
        m^{-2}\sum_{\tilde\gamma\in\Pi_X(u,w)}s(\tilde{\gamma}) = B(u,w).
    \end{equation*}
    We conclude that $\tilde{B} = BD^{-1}$.
 \end{proof}
 \begin{figure}[H]
     \centering
     \begin{subfigure}{0.3\linewidth}
     \includegraphics[scale=0.4]{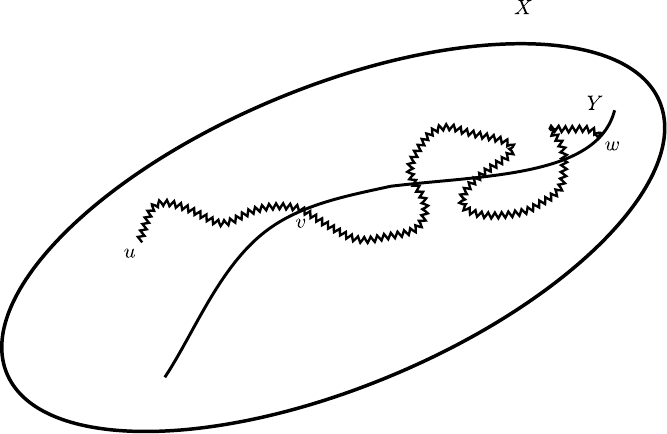}
         \caption{A h-path in $\Pi(u,w)$ contributing to $\langle u | B| w\rangle$.}
     \end{subfigure}
     \qquad
     \begin{subfigure}{0.3\linewidth}
     \includegraphics[scale=0.4]{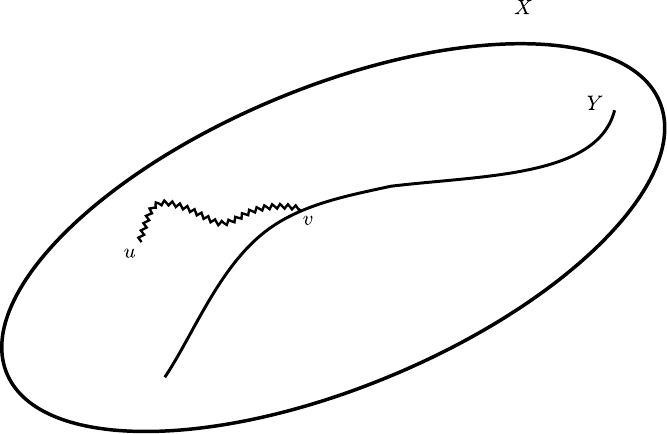}
     \caption{A h-path in $\Pi'(u,v)$ contributing to $\langle u | E_{Y,X} | v \rangle$.}
     \end{subfigure}
       \begin{subfigure}{0.3\linewidth}
     \includegraphics[scale=0.4]{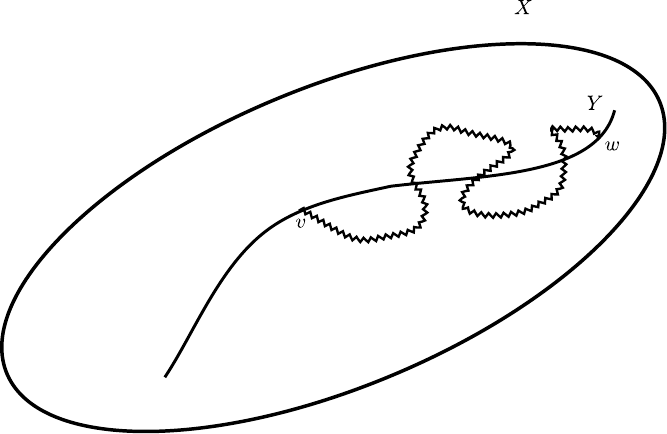}
     \caption{A h-path in $\Pi(v,w)$ contributing to $\langle v | D | w \rangle$.}
     \end{subfigure}
     \caption{Paths contributing to $B$ (left) can be decomposed into paths contributing to $E_{Y,X}$ (middle) and paths contributing to $D$ (right),
     proving that $B = E_{Y,X}D$. }
     \label{fig:enter-label}
 \end{figure}

 \subsubsection{Resumming h-paths}
 In the relative case, for any path $\gamma$ we use the notation 
 \begin{equation}
     w_{X,Y}(\gamma) = \prod_{v \in V(\gamma) \setminus V(Y)} \frac{1}{m^2 + \val_X(v)},
 \end{equation}
 where for a vertex $v\in X \setminus Y$, we put the subscript $X$ on $\val_X(v)$ to emphasize we are considering its valence in $X$, i.e. we are counting all edges in $X$ incident to $v$ regardless if they end on $Y$ or not. Then we have the following path sum formulae for the relative objects: 
\begin{proposition}\label{prop: rel gadgets path sums}
    The propagator with Dirichlet boundary condition can be expressed as 
    \begin{equation}
        \langle u | G_{X,Y} |v \rangle = \sum_{\gamma \in P_{X\setminus Y}(u,v)}w_{X,Y}(\gamma). \label{eq: path sum rel prop}
    \end{equation}
Here the sum is over paths involving only vertices in $X \setminus Y$.  \footnote{Notice that if instead we were using the path weight $w_{X\setminus Y}(\gamma)$, we would obtain the Green's function $G_{X \setminus Y}$ of the closed graph $X \setminus Y$, not the relative Green's function $G_{X,Y}$.}

Similarly, for the extension operator we have 
\begin{equation}
    \langle u | E_{Y,X} |v\rangle = \sum_{\gamma \in P'_{X,Y}(u,v)} w_{X,Y}(\gamma),  
    \label{eq: path sum extension}
\end{equation}
where $P'_{X,Y}(u,v)$ denotes paths in $X\setminus Y$ from $u$ to $v$ with exactly one vertex (i.e. the endpoint) in $Y$. 
Finally, the operator $D'$ appearing in the Dirichlet-to-Neumann operator can be written as 
\begin{equation}
    \langle u | D' |v \rangle  = -m^{-2}\val(v)\delta_{uv} + m^{-2}\sum_{\gamma \in P''_{X,Y}(u,v)} w_{X,Y}(\gamma), 
    \label{eq: path sum D prime}
\end{equation}
where $P''_{X,Y}(u,v)$ denotes paths in $X$ with exactly two (i.e. start- and endpoint) vertices in $Y$. 
In particular, the Dirichlet-to-Neumann operator is 
\begin{equation}
    \langle u | \DN_{Y,X} |v \rangle  = (m^2 + \val(v))\delta_{uv} - \sum_{P''_{X,Y}(u,v)}w_{X,Y}(\gamma). \label{eq: path sum D to N}
\end{equation}
\end{proposition}
\begin{proof}
    Equation \eqref{eq: path sum rel prop} is proved with a straightforward generalization of the arguments in the previous section. For equation \eqref{eq: path sum extension}, notice that because of the final jump there is an additional factor of $m^{-2}$. For the Dirichlet-to-Neumann operator, we have the initial and final jumps contributing a factor of $-m^{-2}$. In the case $u = v$, the contribution of the h-paths which simply hesitate once at $v$ have to be taken into account separately and result in the first term in \eqref{eq: path sum D prime}. Finally, \eqref{eq: path sum D to N} follows from \eqref{eq: path sum D prime} and $\DN_{Y,X} = m^2 (I - D')$. 
\end{proof}
We also have a similar statement for the determinant. For this, we introduce the normalized relative kinetic operator $$\tilde{K}_{X,Y} = \Lambda_{X,Y}^{-1}K_{X,Y} = I - \Lambda_{X,Y}^{-1}A_{X\setminus Y} ,$$
where $\Lambda_{X,Y}$ is the diagonal matrix whose entries are $m^2 + \val_X(v)$. For a closed path $\gamma \in P_{X\setminus Y}(v,v)$, we introduce the notation 
$$w'_{X,Y}(\gamma) = (m^2+\val_X)(v)\prod_{w 
\in V(\gamma)}\frac{1}{m^2+\val_X(w)}. $$
\begin{proposition}
    The determinant of the normalized relative kinetic operator is 
    \begin{equation}\label{eq: det K rel paths}
        \log\det \tilde{K}_{X,Y} = - \sum_{v \in X \setminus Y} \sum_{k=1}^\infty \sum_{\gamma \in P^k_{X\setminus Y}(v,v)}\frac{w'_{X,Y}(\gamma)}{k} = - \sum_{[\gamma] \in C^{\geq 1}_{X
        \setminus Y}} \frac{w'_{X,Y}(\gamma)}{t(\gamma)}.
    \end{equation}
\end{proposition}
\begin{proof}
    Again, simply notice that 
    $$\log\det \tilde{K}_{X,Y} = \mathrm{tr}\log (1 - \Lambda^{-1}_{X,Y}A_{X \setminus Y}).$$
    Then, the argument is the same as in the proof of Proposition \ref{prop: determinant path sums} above.
\end{proof}
In the relative case, we are counting paths in $X\setminus Y$, but weighted according to the valence of vertices in $X$. This motivates the following definition. 
\begin{definition}
    We say that pair $(X,Y)$ of a graph $X$ and a subgraph $Y$ is quasi-regular  of degree $n$ if all vertices $v \in X \setminus Y$ have the same valence $n$ in $X$, i.e. 
$$\val_X(v) =n, \forall v \in  V(X\setminus Y).$$ 
\end{definition}
If $X$ is regular, the pair $(X,Y)$ is quasi-regular for any subgraph $Y \subset X$. An important class of examples are the line graphs $X$ of example \ref{ex: path 2 graph rel bdry} with $Y$ both boundary vertices, or more generally rectangular graphs or their higher-dimensional counterparts with $Y$ given by the collection of boundary vertices. See Figure \ref{fig: quasi regular graph}. 

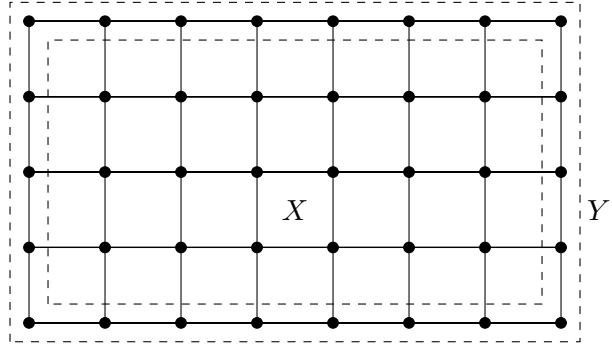
\begin{figure}[H]
    \centering
    \begin{tikzpicture}
    \foreach \x in {0,1,...,7}
    { 
    \draw (\x,0 ) -- (\x,4);
    \foreach \y in {0,1,...,4}
    {
    \draw (0,\y) -- (7,\y); 
    \draw[fill=black] (\x,\y) circle (2pt); 
    }
    }
    \draw[dashed] (-.25,-.25) -- (-.25,4.25) -- (7.25,4.25) -- (7.25,-.25) -- (-.25,-.25);
    \draw[dashed] (.25,.25) -- (.25,3.75) -- (6.75,3.75) -- (6.75,.25) -- (.25,.25);
    \node at (3.5,1.5) {$X$};
    \node at (7.5,1.5) {$Y$};
    \end{tikzpicture}
    \caption{A quasi-regular graph pair $(X,Y)$ with $n=4$.}
    \label{fig: quasi regular graph}
\end{figure}

For quasi-regular graphs, the path sums of Proposition \ref{prop: rel gadgets path sums} simplify to power series in $(m^2 + n)^{-1}$, with $n$ the degree of $(X,Y)$: 
\begin{corollary}
    Suppose $(X,Y)$ is quasi-regular, then we have the following power series expansions for the relative propagator, extension operator, Dirichlet-to-Neumann operator and determinant:
    \begin{align}
        \langle u | G_{X,Y} |v \rangle &= \frac{1}{m^2+n}\sum_{k=0}^\infty p^k_{X\setminus Y}(u,v)(m^2 + n)^{-k} , \label{eq: path sum rel prop quasi reg} \\
         \langle u | E_{Y,X} |v\rangle &= \sum_{k=1}^\infty  ({p'}^k_{X,Y})(u,v) (m^2 +n)^{-k} ,\label{eq: path sum extension quasi reg} \\
    \langle u | \DN_{Y,X} |v \rangle  &= (m^{2}+\val(v))\delta_{uv} - \sum_{k=2}^\infty ({p''}^k_{X,Y})(m^2 + n)^{-k+1} , \label{eq: path sum D to N quasi reg} \\
        \log\det \tilde{K}_{X,Y} &= - \sum_{k=1}^\infty \sum_{[\gamma]\in C^k_X}\frac{(m^2+n)^{-k}}{t(\gamma)} .
    \end{align} 
\end{corollary}
\begin{figure}[H]
    \centering
    \includegraphics[scale=0.8]{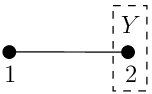}
    \caption{The 2--vertex line graph.}
    \label{fig:2 vertex examples paths}
\end{figure}

Again, we can collect our findings in the following first quantization formula for the partition function: 
\begin{thm}
    The logarithm of the partition function of the Gaussian theory relative to a subgraph $Y$ is 
    \begin{multline}
       \hbar \log Z_{X,Y}(\phi_Y) = 
       -\frac12\sum_{u,v\in Y}\phi_Y(u)\phi_Y(v) \cdot\\
       \cdot\bigg(\left(\frac{m^2}{2} + \val_X(v)-\frac12\val_Y(v)\right)\delta_{uv} + \frac12(A_Y)_{uv}\  - \sum_{\gamma \in P''_{X,Y}(u,v)}w_{X,Y}(\gamma) \bigg) 
       \\
       +\frac{\hbar}{2}\left(\sum_{[\gamma]\in C^{\geq 1}_{X\setminus Y}}\frac{w'_{X,Y}(\gamma)}{t(\gamma)} - \sum_{v\in X}\log (m^2 + \val(v))\right). \label{eq: 1q formula rel Z}
    \end{multline}
\end{thm}
In \eqref{eq: 1q formula rel Z} we are summing over all connected Feynman diagrams with no bulk vertices: boundary-boundary edges in the last term of the second line of the r.h.s. at order $\hbar^0$ (together with the diagonal terms  and $\frac{1}{2}(A_Y)_{uv}$, they sum up to $\DN_{Y,X} - \frac{1}{2}K_Y$) and ``1-loop graphs'' (cycles) on the third line at order $\hbar^1$.  

\subsubsection{Examples}
\begin{example}
Consider the graph $X$ in Figure \ref{fig:2 vertex examples paths}, with $Y$ the subgraph consisting of the single vertex on the right. Then, the set $\Pi_{X,Y}$ consists exclusively of iterates of the path which hesitates once along the single edge at 1, $\tilde\gamma = (1,(12),1)$. Therefore, we obtain 
$$\langle 1 | G_{X,Y} | 1 \rangle = m^{-2}\sum_{k=0}^\infty (-m^{-2})^k = \frac{m^{-2}}{1+m^{-2}} = \frac{1}{1+m^2}.$$ 
Alternatively, we can obtain this from the path sum formula \eqref{eq: path sum rel prop} by noticing there is a single (constant) path from 1 to 1 in $X \setminus Y$.
For the determinant, we obtain 
$$\log\det K_{X,Y}/m^2 = - \sum_{k\geq 1} \frac{(-m^{-2})^k}{k}  = \log (1+m^{-2}) = \log \frac{1+m^2}{m^2}. $$ 
h-paths  in $\Pi''_{X,Y}(2,2)$ are either $(2,(12),2)$ or of the form $(2,(12),1,(12),1,\ldots, 1,(12),2)$ -i.e. jump from 2 to 1, hesitate $k$ times and jump back - and therefore the operator $D'$ is given by 
$$ D'= -m^{-2} + \sum_{k \geq 0} (m^{-2})^{k+1}(-1)^{k} = -m^{-2} + \frac{m^{-2}}{1+m^{-2}} = \frac{-1}{m^2 +1}.$$ 
Alternatively, one can just notice there is a unique path in $P''_{X,Y}(2,2)$, namely (212), and use formula \eqref{eq: path sum D to N}. 
Therefore the Dirchlet-to-Neumann operator is 
$$\DN_{Y,X} = m^2\left(1 - \frac{-1}{m^2+1}\right) =\frac{m^2(2+m^2)}{1+m^2}.$$
Finally, h-paths in $\Pi'_{X,Y}(1,2)$ are only those that hesitate $k$ times at 1 before eventually jumping to 2, and therefore the extension operator is 
$$\langle 1 | E_{Y,X} | 2 \rangle = m^{-2}\sum_{k = 0}^\infty (-m)^{-2} = m^{-2} \frac{1}{1+m^{-2}} = \frac{1}{1+m^2},$$
alternatively, this follows directly from formula \eqref{eq: path sum extension}, because $P'_{X,Y}(1,2) = \{(12)\}$. 
\end{example}
\begin{example}
    Consider $X$ the $N=4$ line graph with $Y$ both endpoints (1 and 4). Then $(X,Y)$ is quasi-regular of degree 2 and we can count paths $X\setminus Y$ easily, namely, we have $$p^k_{X\setminus Y}(2,2) = p^k_{X \setminus Y}(3,3) = \begin{cases}
        1 & k \text{ even} \\
        0 & k \text{ odd}
    \end{cases}$$ and 
    $$ 
    p^k_{X\setminus Y}(2,3) = p^k_{X \setminus Y}(3,2) = \begin{cases}
        0 & k \text{ even} \\
        1 & k \text{ odd}
    \end{cases}$$
\end{example}
Therefore, the relative Green's function is 
$$ G_{X,Y}(2,2) = \frac{1}{m^2+2} \sum_{k=0}^\infty \frac{1}{(m^2+2)^{2k}} = \frac{1}{m^2+2}\cdot\frac{1}{1-\frac{1}{(m^2+2)^2}} = \frac{m^2+2}{(m^2+1)(m^2+3)}$$ 
and 
$$G_{X,Y}(2,3) = \frac{1}{m^2+2}\sum_{k=0}^\infty\frac{1}{(m^2+2)^{2k+1}} = \frac{1}{(m^2+1)(m^2+3)} ,$$
in agreement with \eqref{G for path graph rel both ends}. 

As for the determinant, notice there is a unique cycle of length 2, all other cycles are iterates of this one, therefore, the logarithm of the normalized determinant is given by 
$$\log\det \tilde{K}_{X,Y} = -\sum_{k=1}^\infty \frac{(m^2+2)^{-2k}}{k} = \log (1 - (m^2+2)^{-2}) $$ 
and the determinant is then $$\det K_{X,Y} = (m^2+1)(m^2+3) ,$$ 
in agreement with \eqref{det for path graph rel both ends}. 

For an example of the extension operator, notice that 
$(p')^k_{X, Y}(2,1)$ is 1 for odd $k$ and 0 for even $k$, and therefore 
$$ \langle 2 | E_{Y,X} | 1 \rangle = \sum_{k=0}^\infty (m^2+2)^{-(2k+1)} = \frac{m^2+2}{(m^2+1)(m^2+3)}$$ 
and similarly $(p')^k_{X,Y}(3,1) =1 $ for $k \geq 2$ even and 0 for odd $k$, and therefore 
$$\langle 3 |E_{Y,X} | 1\rangle = 
\sum_{k=0}^\infty (m^2+2)^{-(2k+2)} =  \frac{1}{(m^2+1)(m^2+3)},$$ 
in agreement with \eqref{E for path graph rel both ends}. Finally, we can compute the matrix elements of the Dirichlet-to-Neumann operator: we have $(p'')^k(1,1) = 1$ for even $k \geq 2$ and it vanishes for odd $k$, therefore 
$$\langle 1 | \DN_{Y,X} | 1 \rangle = m^2+1 - \sum_{k=1}^\infty (m^2 +2)^{-2k+1} = m^2 + 1 - \frac{m^2+2}{(m^2+1)(m^2+3)}.
$$
Similarly, $(p'')^k(1,4)$ vanishes for even $k$ and is 1 for odd $k \geq 3$, and therefore
$$ \langle 1 | \DN_{Y,X} | 4 \rangle = -\sum_{k=1}^\infty \frac{1}{(m^2+2)^{2k}} = -\frac{1}{(m^2+1)(m^2+3)}.$$ 
These formulae agree with \eqref{DN for path graph rel both ends}. 
\begin{table}[H]
    \bgroup
\def\arraystretch{1.5}
\begin{adjustbox}{max width=1.1\textwidth,center}
\begin{tabular}{c|c|c}
        Object & h-path sum & path sum \\
        \hline $\langle u | G_{X,Y} | v \rangle$ & $m^{-2} \sum_{\tilde\gamma\in\Pi_{X,Y}(u,v)}s(\tilde{\gamma})$ (Eq. \eqref{eq: G rel hesitant paths})& $\sum_{\gamma \in P_{X\setminus Y}(u,v)}w_{X,Y}(\gamma)$ (Eq. \eqref{eq: path sum rel prop}\\ 
        $\log\det \frac{K_{X,Y}}{m^2} $ & $ - \sum_{[\tilde\gamma] \in \Gamma^{\geq 1}_{X,Y}} \frac{s(\tilde\gamma)}{ {t(\tilde\gamma) }} $ (Eq. \eqref{eq: det K rel hesitant paths})& \\
        $\log\det \tilde{K}_{X,Y}$ &  &  $- \sum_{[\gamma] \in C^{\geq 1}_{X,Y}} \frac{w'_{X,Y}(\gamma)}{t(\gamma)}$ (Eq. \eqref{eq: det K rel paths})\\ 
        $\langle u| E_{Y,X} | v \rangle$ & $\sum_{\tilde\gamma\in\Pi'_{X,Y}(u,v)}s(\tilde{\gamma})$ (Eq. \eqref{eq: h path sum extension operator} & $\sum_{\gamma \in P'_{X,Y}(u,v)} w_{X,Y}(\gamma)$ \eqref{eq: path sum extension}\\
        $\langle u|\DN_{Y,X} | v \rangle$ & $m^2\delta_{uv}- $  & $ (m^2 + \val(v))\delta_{uv}-$ \\
        & $- m^2\sum_{\tilde\gamma \in \Pi''_{X,Y}(u,v)} s(\tilde{\gamma})$ (Eq. \eqref{eq: D to N rel h paths}) &  $ - \sum_{P''_{X,Y}(u,v)}w_{X,Y}(\gamma)$ (Eq. \eqref{eq: path sum D to N})
    \end{tabular}
    \end{adjustbox}
    \egroup
    \caption{Summary of path sum formulae, relative case.}
    \label{tab:my_label}
\end{table}
\subsection{Gluing formulae from path sums}\label{sec: gluing proof path sums}
In this section we prove Theorem \ref{thm: gluing prop and det} from the path sum formulae presented in this chapter. The main observation in this proof is a decomposition of h-paths in $X$ with respect to a subgraph $Y$. 
\begin{lemma}
    Let $u,v \in X$, then we have a bijection 
    \begin{multline}
\Pi_X(u,v) \leftrightarrow   \\
\Pi_{X,Y}(u,v) \bigsqcup \sqcup_{w_1,w_2\in Y} \Pi'_{X,Y}(u,w_1) \times \Pi_X(w_1,w_2) \times \Pi'_{X,Y}(w_2,v) ,      \label{eq: decomposition hesitant paths}
    \end{multline}
    where $\Pi_{X,Y}(u,v)$ denotes h-paths in $X$ that contain no vertices in $Y$ (but they may contain edges between $X\setminus Y$ and $Y$) and $\Pi'_{X,Y}(u,w)$, for either $u$ or $w$ in $Y$, denote h-paths containing exactly one vertex in $Y$, namely the initial or final one.\footnote{It is possible to have $u=w \in Y$, in which case there $\Pi'_X(w,w)$ contains only the 1-element path.}
    \end{lemma}
    \begin{proof}
        One may decompose $\Pi_X(u,v)$ into paths containing no vertex in $Y$ and those containing at least one vertex in $Y$. The former are precisely $\Pi_{X,Y}(u,v)$. If $\tilde\gamma$ is an element of the latter, let $w_1$ be the first vertex in  $\tilde\gamma$ in $Y$ and $w_2$ the last vertex in $\tilde\gamma$ in $Y$. Splitting $\tilde\gamma$ at $w_1$ and $w_2$ gives the map from left to right.  The inverse map is given by composition of h-paths. See also Figure \ref{fig:pathdecomp}
    \end{proof}
    \begin{figure}[H]
        \centering
        \includegraphics{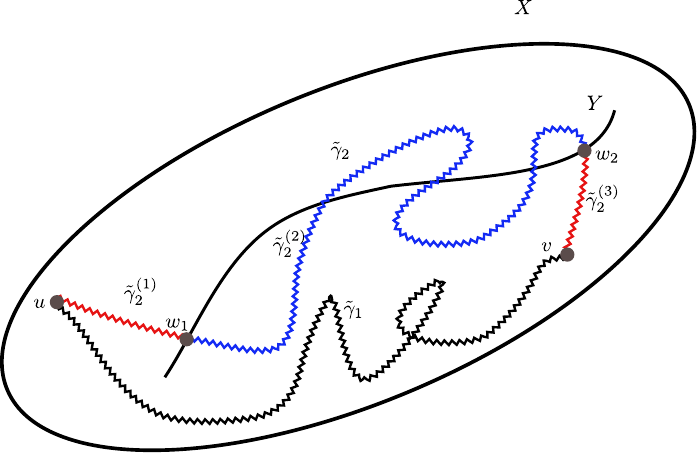}
        \caption{h-paths from $u$ to $v$ fall in two categories. Those not containing vertices in $Y$ (e.g. $\tilde\gamma_1$) are paths in $\Pi_{X,Y}(u,v)$, while those intersecting $Y$ (e.g. $\tilde\gamma_2$) can be decomposed in paths in $\tilde\gamma_2^{(1)}\in\Pi'_{X,Y}(u,w_1)$ and $\tilde\gamma_2^{(3)}\in\Pi'_{X,Y}(w_2,v)$ (red) and a path $\tilde\gamma_2^{(2)} \in\Pi_X(w_1,w_2)$, where $w_1,w_2$ are the first and last vertices in $\tilde\gamma_2$ contained in $Y$.}
        \label{fig:pathdecomp}
    \end{figure}

    For the gluing formula for the determinant, we will also require the following observation on counting of closed paths. 
    \begin{lemma}\label{lem: path decomp closed}
        Denote $\Gamma_{X,Y}^{\geq 1, (k)}$ the set of h-cycles $[\tilde\gamma]$ in $X$ of length $l(\tilde\gamma)\geq 1$ that intersect $Y$ exactly $k$ times, with $k \geq 1$. Then concatenation of paths 
        \begin{equation}
           \bigsqcup_{w_1,w_2,\ldots,w_k \in Y} \Pi_X''(w_1,w_2) \times \Pi_X''(w_2,w_3) \times \ldots \times \Pi_X''(w_k,w_1) \to \Gamma_{X,Y}^{\geq 1, (k)} \label{eq: lemma path decomp closed}
        \end{equation}
        is surjective, and a cycle $[\tilde\gamma]$ has precisely $k/t(\tilde\gamma)$ preimages. 
    \end{lemma}
    \begin{proof}
        For a cycle $\tilde\gamma \in \Gamma_{X,Y}^{\geq 1, (k)}$, denote $w_1, \ldots w_k$ the intersection points with $Y$ and $\tilde\gamma^{(i)}$ the segment of $\tilde\gamma$ between $w_{i+1}$ and $w_i$ (here we set $w_{k+1} = w_1$). See Figure \ref{fig:pathdecompclosed}. Then obviously $\tilde\gamma$ is the concatenation of the $\tilde\gamma^{(i)}$, so concatenation is surjective. On the other hand a $k$-tuple of paths concatenates to the same closed path if and only if they are related to each other by a cyclic shift (this corresponds to a cyclic shift of the labeling of the intersection points). They are precisely $k/t(\tilde\gamma)$ such shifts.
    \end{proof}
    Recall that  $D'$ is the operator given by summing the weight $s(\tilde\gamma)=(m^{-2})^{l(\tilde\gamma)}(-1)^{h(\tilde\gamma)}$ over paths starting and ending on $Y$ without intersecting $Y$ in between (Eq. \eqref{eq: D prime hesitant paths }.)
    \begin{corollary}\label{cor: trace D prime k}
        We have that 
        \begin{equation}
            \mathrm{tr}(D')^k = k \sum_{\tilde\gamma\in 
            \Gamma_{X,Y}^{\geq 1,(k)}
            } 
            \frac{(m^{-2})^{l(\tilde\gamma)}(-1)^{h(\tilde\gamma)}}{t(\tilde\gamma)} .
        \end{equation}
    \end{corollary}
    \begin{proof} 
    The statement follows by summing the weight $s(\tilde\gamma)$ over the l.h.s. and  r.h.s. of \eqref{eq: lemma path decomp closed} in Lemma \ref{lem: path decomp closed}, using multiplicativity of $s(\tilde\gamma)$ in the l.h.s. and with multiplicity $k/t(\tilde\gamma)$ in the r.h.s. (corresponding to the count of preimages of the map (\ref{eq: lemma path decomp closed})).
    \end{proof}
    \begin{figure}[H]
        \centering
        \includegraphics{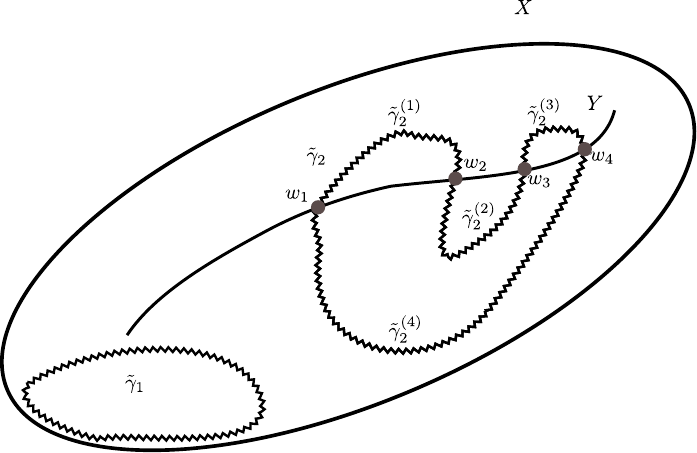}
        \caption{Cycles in $X$ either do not intersect $Y$ (like $\tilde\gamma_1$) and such that intersect $Y$ $k$ times (in the case of $\tilde\gamma_2$, $k=4$. Such paths can be decomposed into $k$ h-paths $\tilde\gamma^{(i)}$ in $\Pi''_{X,Y}(w_i,w_{i+1})$ in $k$ different ways, corresponding to cyclic shift of the labels of $w_i$'s.}
        \label{fig:pathdecompclosed}
    \end{figure}
\begin{proof}[h-path sum proof of Theorem \ref{thm: gluing prop and det}]
    We first prove the gluing formula 
    $$\langle u | G_X| v\rangle = \langle u | G_{X,Y} | v \rangle + \sum_{v_1,v_2 \in Y} \langle u | E_{Y,X} | v_1\rangle  \langle v_1 | \DN_{Y,X}^{-1} | v_2\rangle   \langle v_2 | E_{Y,X} | v\rangle.$$ 
    Applying the decomposition of $\Pi_X(u,v)$  \eqref{eq: decomposition hesitant paths}, and using multiplicativity of the weight $s(\tilde\gamma) = (m^{-2})^{l(\tilde\gamma)}(-1)^{h(\tilde\gamma)}$, we get 

\begin{multline}
G_X(u,v) = m^{-2}\sum_{\tilde\gamma \in \Pi_{X,Y}(u,v)}s(\tilde\gamma) + \\
\sum_{w_1,w_2\in Y}\left(\sum_{\tilde\gamma_1\in\Pi'_{X,Y}(u,w_1)}s(\tilde\gamma)\right)\left(m^{-2}\sum_{\tilde\gamma_2\in\Pi_{X}(w_1,w_2)}s(\tilde\gamma)\right)\left(\sum_{\tilde\gamma_1\in\Pi'_{X,Y}(w_2,v)}s(\tilde\gamma)\right).
\end{multline}
The first term is $G_{X,Y}$ by equation \eqref{eq: G rel hesitant paths}. In the second term, we recognize the path sum expressions \eqref{eq: h path sum extension operator} for the extension operator and \eqref{eq: D path sums} for the operator $D$, which is the inverse of the total Dirichlet-to-Neumann operator. This completes the proof of the gluing formula for the propagator. 

Next, we prove the gluing formula for the determinant 
$$\det(K_X) = \det(K_{X,Y})\det(\DN_{Y,X}). $$
Dividing both sides by $m^{2N}$, where $N$ is the number of vertices in $X$, this is equivalent to 
$$ \det(m^{-2}K_X) = \det(m^{-2}K_{X,Y})\det(m^{-2}\DN_{Y,X}). $$
Taking negative logarithms and using $ \log \det = \mathrm{tr} \log$, this is equivalent to 
\begin{equation} -\log\det (1 + m^{-2}\Delta_X) =  -\log\det (1 + m^{-2}\Delta_{X,Y})  - \mathrm{tr} \log (I - D'),\label{eq: proof path sums det}
\end{equation}
where we have used that $\DN_{Y,X} = m^2(I - D').$

We claim that equation \eqref{eq: proof path sums det} can be proven by summing over paths. Indeed, the left hand side is given by summing over closed h-paths in $X$. We decompose them into paths which do not intersect $Y$, and those that do.
From the former we obtain $-\log \det K_{X,Y}/m^2$ by equation \eqref{eq: det K rel hesitant paths}. Decompose the latter set into paths  that intersect $Y$ exactly $k$ times, previously denoted $C_{X,Y}^{\geq 1, (k)}$. By Corollary \ref{lem: path decomp closed}, when summing over those paths we obtain precisely $\mathrm{tr} (D')^k/k$.
Summing over $k$ we obtain $\mathrm{tr} \sum_{k\geq 1} (D')^k/k = -\mathrm{tr}\log (I-D')$, which proves the gluing formula for the determinant. 
\end{proof}
\begin{proof}[Path sum proof of Theorem \ref{thm: gluing prop and det}]
\label{rem: gluing path sums}
    In the proof above we used the h-path expansions, but of course one could have equally well used the formulae in terms of paths. To prove the gluing formula for the Green's function
    $$\langle u | G_X| v\rangle = \langle u | G_{X,Y} | v \rangle + \sum_{v_1,v_2 \in Y} \langle u | E_{Y,X} | v_1\rangle  \langle v_1 | \DN_{Y,X}^{-1} | v_2\rangle   \langle v_2 | E_{Y,X} | v\rangle$$ 
    in terms of path counts, notice that a path crossing $Y$ can again be decomposed into a path from $X$ to $Y$, then a path from $Y$ to $Y$ and another path from $Y$ to $X$. The weight $w(\gamma) = \prod_{v \in V(\gamma)}(m^2+\val(v))^{-1}$ is distributed by among those three paths by taking the vertices on $Y$ to the $Y-Y$ path: In this way, when summing over all paths from the $Y-Y$ paths we obtain precisely the operator $D = \DN_{Y,X}^{-1}$ (this is a submatrix of $G_X$ and hence the weights of paths \emph{include} start and end vertices) while from the other parts we obtain the extension operator $E_{Y,X}$ (where weights of paths \emph{do not} include the vertex on $Y$). 
    
Next, we consider the gluing formula for the determinant, 
$$ \det K_X = \det K_{X,Y}\det \DN_{Y,X}.$$ 
Dividing both sides by 
\begin{multline*}
\det \Lambda_X = \prod_{v \in X} (m^2 + \val(v))= \prod_{v \in X \setminus Y}(m^2 + \val_X(v))\prod_{v\in Y}(m^2 + \val_X(v)) \\ = \det \Lambda_{X,Y} \det \Lambda_Y,
\end{multline*}
this is equivalent to 
$$ \det \tilde{K}_X = \det \tilde{K}_{X,Y} \det \Lambda_Y^{-1} \DN_{Y,X}.$$ 
Taking logarithms and using the formulae \eqref{eq: det regular paths} and \eqref{eq: det K rel paths} for logarithms of determinants of kinetic operators, we get 
\begin{equation}
\label{eq: proof det gluing paths 1}
\log \det \tilde{K}_X - \log \det \tilde{K}_{X,Y} = - \sum_{[\gamma] \in C^{\geq 1}_{X}, V(\gamma) \cap Y \neq \varnothing} \frac{w'_{X}(\gamma)}{t(\gamma)},
\end{equation}
where on the right hand side we are summing over cycles in $X$ that intersect $Y$. We therefore want to show that the sum on the r.h.s of \eqref{eq: proof det gluing paths 1} equals $\log\det \Lambda^{-1}_Y \DN_{Y,X}$. 
From \eqref{eq: path sum D to N}, we have that $\DN_{Y,X} = \Lambda_Y - D''$, where we introduced the auxiliary operator $D'' \colon C^0(Y) \to C^0(Y)$ with matrix elements 
$$\langle u | D'' | v \rangle = \sum_{\gamma \in P''_{X}(u,v)} w_{X,Y}(\gamma).$$ 
Then 
\begin{equation}
\label{eq: proof det gluing paths 2}-\log\det \Lambda^{-1}_Y \DN_{Y,X} = -\operatorname{tr} \log (I - \Lambda^{-1}_YD'') = \sum_{k \geq 1} \operatorname{tr} \frac{(\Lambda^{-1}_YD'')^k}{k}.
\end{equation}
Notice that $\operatorname{tr}(\Lambda^{-1}_YD'')^k$ is given by summing over closed paths $\gamma$ that intersect $Y$ exactly $k$ times, with the weight $w'(\gamma)$: the factor $(m^2 + \val(v))^{-1}$ in $w'(\gamma)$ for vertices not on $Y$ comes from $D''$ (recall that $w_{X,Y}(\gamma)$ does not contain factors for vertices on $Y$), and from $\Lambda^{-1}_Y$, for $v \in Y$. By a combinatorial argument analogous to Lemma \ref{lem: path decomp closed}, every cycle appears in this way exactly $\frac{k}{t(\gamma)}$ times. Therefore the sum on the r.h.s. of \eqref{eq: proof det gluing paths 1} equals the sum on the r.h.s of equation \eqref{eq: proof det gluing paths 2}, which finishes the proof. 
\end{proof}


\section{Interacting theory: first quantization formalism} 
\label{s: Interacting theory: first quantization formalism}
In this section, we extend the path sum formulae to the interacting theory. In this language, weights of Feynman graphs are given by summing over all possible maps from a Feynman graph to a spacetime graph where edges are mapped to paths. We also analyze the gluing formula in terms of path sums. 
\subsection{Closed graphs}
We first consider the case of closed graphs.
\subsubsection{Edge-to-path maps}
Let $\Gamma$ and $X$ be graphs. Recall that by $P_X$ we denote the set of all paths in $X$, and by $\Pi_X$ the set of h-paths in $X$. 
\begin{definition}
    An \emph{edge-to-path map} $F = (F_V,F_P)$ from $\Gamma$ to $X$ is a pair of maps 
    $F_V \colon V_\Gamma \to V_X$ and $F_P\colon E(\Gamma) \to P_X$
    such that for every edge $e = (\s{u},\s{v})$ in $\Gamma$ we have 
    $$ F_P(e) \in P_X(F_V(\s{u}),F_V(\s{v})).$$ 
    The set of edge-to-path maps is denoted $P_X^\Gamma$.
    \end{definition}
    Equivalently, an edge-to-path map is a lift of a map $F_V \colon V_\Gamma \to V_X$ to the fibrations 
    $E_\Gamma \to V_\Gamma \times V_\Gamma$ and $P_X \to V_X \times V_X$. 
    \[
\begin{tikzcd}
E_\Gamma \arrow[r, "F_P"] \arrow[d]                 & P_X \arrow[d]  \\
V_\Gamma \times V_\Gamma \arrow[r, "F_V\times F_V"] & V_X \times V_X
\end{tikzcd}
    \]
    Similarly, we define an \emph{edge-to-h-path map} as a lift of a map $F_V \colon V_\Gamma \to V_X$ to the fibrations $E_\Gamma \to V_\Gamma \times V_\Gamma$, $\Pi_X \to V_X \times V_X$. 
    \[ 
\begin{tikzcd}
E_\Gamma \arrow[r, "F_\Pi"] \arrow[d]                 & \Pi_X \arrow[d]  \\
V_\Gamma \times V_\Gamma \arrow[r, "F_V\times F_V"] & V_X \times V_X
\end{tikzcd}
    \] 
    The set of such maps is denoted $\Pi_X^\Gamma$.
    Alternatively, an edge-to-path map can be thought of as labeling of $\Gamma$ where we label vertices in $\Gamma$ by vertices of $X$ and edges in $\Gamma$ by \emph{paths} in $X$. 
\subsubsection{Feynman weights}
Suppose that $\Gamma$ is a Feynman graph appearing in the perturbative partition function on a closed graph, with weight given by \eqref{Feynman weight, X closed}. By the results of the previous section, we have the following first quantization formula, a combinatorial analog of the first quantization formula \eqref{Feynman weight 1q formalism}: 
\begin{proposition}
    The weight of the Feynman graph $\Gamma$ has the path sum expression
    \begin{align}
    \Phi_{\Gamma,X} &= \prod_{\s{v} \in V_\Gamma}(-p_{\val(\s{v})})\sum_{F \in P_X^\Gamma}\prod_{e\in E_\Gamma} w(F_P(e)) \label{eq: Feynman weight closed paths}\\
    &=  \prod_{\s{v} \in V_\Gamma}(-p_{\val(\s{v})})\sum_{F \in \Pi_X^\Gamma}\prod_{e\in E_\Gamma}s(F_\Pi(e)), \label{eq: Feynman weight closed hesitant paths}
    \end{align}
    where in \eqref{eq: Feynman weight closed paths} we are summing over all edge-to-path maps from $\Gamma$ to $X$, and in \eqref{eq: Feynman weight closed hesitant paths} we are summing over all edge-to-h-path maps.
\end{proposition}

Figure \ref{fig: edge to path} contains an example of an edge-to-path map from $\Gamma$ the 
$\Theta$-graph to a grid $X$. 
    \begin{figure}[H]
        \centering
\scalebox{0.8}{
\begin{tikzpicture}
\begin{scope}[shift={(-6,3)}]
  \draw (0,0) -- (2,0); 
  \draw (0,0) to[bend left = 60] (2,0);
  \draw (0,0) to[bend right = 60] (2,0);
\node at (-0.25,0) {$\s{u}$};
\node at (2.25,0) {$\s{v}$};
   \draw[fill = black] (0,0) circle (2pt);
    \draw[fill = black] (2,0) circle (2pt);
\draw[-latex] (3,0) to[bend left = 30] (5,0);
\node at (4,0.6) {$F$};
\node at (1,-2) {$\Gamma$}; 
\end{scope}
\foreach \x in {0,1,...,7} 
{
\draw (\x,-0.5) -- (\x,6.5);
}
\foreach \y in {0,1,...,6} 
{\draw (-.5,\y) -- (7.5,\y);}
\draw[red, fill = red] (1,4) circle (3pt);
\draw[red, fill = red] (5,2) circle (3pt); 
\draw[path] (1,4) -- (1,5) -- (2,5) -- (3,5) -- (3,6) -- (4,6) -- (5,6) -- (6,6) -- (6,2) -- (5,2);
\draw[path] (1,4) -- (4,4) -- (4,2) -- (5,2);
\draw[path] (1,4) -- (1,1) -- (2,1) -- (2,2) -- (3,2) -- (3,1) -- (4,1) -- (4,0) -- (5,0) -- (5,2);
\node at (0.5,4.2) {$F_V(\s{u})$}; 
\node at (5.5,1.7) {$F_V(\s{v})$};
\node at (3,-1) {$X$};
\end{tikzpicture} 
}
\caption{An example of an edge-to-path map.}
        \label{fig: edge to path}
    \end{figure}
  We then have the following expression of the perturbative partition function: 
    \begin{corollary}
        The perturbative partition function of $X$ is given in terms of edge-to-paths maps as
        \begin{equation}
            Z^\mr{pert}_X = \det (K_X)^{-\frac12}\sum_\Gamma\frac{\hbar^{-\chi(\Gamma)}}{\mr{Aut}(\Gamma)}\sum_{F \in P_X^\Gamma}\prod_{\s{v}\in V_\Gamma}(-p_{\val(\s{v})}) \prod_{e \in E_\Gamma}w(F_P(e)).
        \end{equation}
    \end{corollary}

We can reformulate this as the following ``first quantization formula.''
\begin{corollary}
    The logarithm of the perturbative partition function has the expression 
    \begin{multline}
       \label{eq: Z pert 1st quant}
        \log\, Z^\mr{pert}_X = \\
        =\frac12\sum_{[\gamma]\in C^{\geq 1}_X}\frac{w'(\gamma)}{t(\gamma)} +\sum_{\Gamma^\mr{conn}} \frac{\hbar^{-\chi(\Gamma)}}{\mr{Aut}(\Gamma)} \sum_{F \in P^\Gamma_X}\prod_{\s{v}\in V_\Gamma}(-p_{\val(\s{v})})\prod_{e \in E_\Gamma}w(F_P(e)) \\
       -\frac12\sum_{v\in X}\log (m^2 + \val(v)).
    \end{multline}
    Here $\Gamma^\mr{conn}$ stands for connected Feynman graphs.
\end{corollary}
We remark that in the second line of \eqref{eq: Z pert 1st quant}, one can interpret the first term as coming from an analog of edge-to-path maps for the circle, divided by automorphisms of such maps (the factor of 2 comes from orientation reversal). In this sense, the second line can be interpreted as the partition function of a 1d sigma model with target $X$. The term in the third line should be interpreted as a normalizing constant.
\subsection{Relative version}
Now we let $X$ be a graph and $Y$ a subgraph, and consider the interacting theory on $X$ relative to $Y$. 
Recall that in the relative case, Feynman graphs $\Gamma$ have vertices split into bulk and boundary vertices, with Feynman weight given by \eqref{Feynman weight, relative}. Bulk vertices have valence at least 3, while boundary vertices are univalent. 
Again, we do not want to allow boundary-boundary edges. Edge-to-path maps now additionally have to respect the type of edge: bulk-bulk edges are mapped to paths in $P_{X\setminus Y}$ and bulk-boundary edges are mapped to paths in $P'_{X,Y}$. 
We collect this in following technical definition: 
\begin{definition}
    Let $\Gamma$ be a graph with $V(\Gamma) = V_\Gamma^{\mathrm{bulk}}\sqcup V_\Gamma^\partial$, such that $\val(v) \geq 3$ for all $v \in V_\Gamma^{\mathrm{bulk}}$ and $\val(v^\partial) = 1 $ for all $v^\partial \in V_\Gamma^\partial$. Denote by the induced decomposition of edges by 
    $$E(\Gamma) = E_\Gamma^{\mr{bulk}-\mr{bulk}} \sqcup E_\Gamma^{\mr{bulk}-\mr{bdry}} \sqcup E_\Gamma^{\mr{bdry}-\mr{bdry}}.$$ 
    Let $X$ be a graph and $Y \subset X$ be a subgraph. Then a \emph{relative edge-to-path map} (resp. \emph{relative edge-to-h-path map}) is a pair $F=(F_V,F_P)$ (resp. $F=(F_V,F_\Pi)$) where $F_V\colon V(\Gamma) \to V(X)$ and $F_P\colon E(\Gamma) \to P_X$ (resp. $F_\Pi \colon E(\Gamma) \to \Pi_X$) such that 
    \begin{itemize}
        \item $F_V$ respects the vertex decompositions, i.e. $F_V(V_\Gamma^{\mr{bulk}})\subset V(X) \setminus V(Y)$ and $F_V(V_\Gamma^{\mr{bdry}}) \subset V(Y)$, 
        \item $F_E$ (resp. $F_\Pi$) is a lift of $F_V$ i.e. for all edges $e = (\s{u},\s{v}) \in E(\Gamma)$ we have $F_P(e) \in P_X(F_V(\s{u}),F_V(\s{v}))$ (resp. $F_\Pi(e) \in \Pi_X(F_V(\s{u}),F_V(\s{v}))$),
        \item $F_P$ (resp. $F_\Pi$) respects the edge decompositions, i.e. $F_P(E_\Gamma^{\mr{bulk}-\mr{bulk}}) \subset P_{X\setminus Y}$, $F_P(E_\Gamma^{\mr{bulk}-\mr{bdry}})\subset P'_{X,Y}$, 
        and similarly for $F_\Pi$. 
    \end{itemize}
    The set of relative edge-to-(h-)path maps is denoted $P_{X,Y}^\Gamma$ (resp. $\Pi_{X,Y}^\Gamma$). 
\end{definition}

Figure \ref{fig: relative edge to path} contains an example of a relative edge-to-path map from $\Gamma$ a Feynman graph with boundary vertices to a grid $X$ relative to a subgraph $Y$. 
    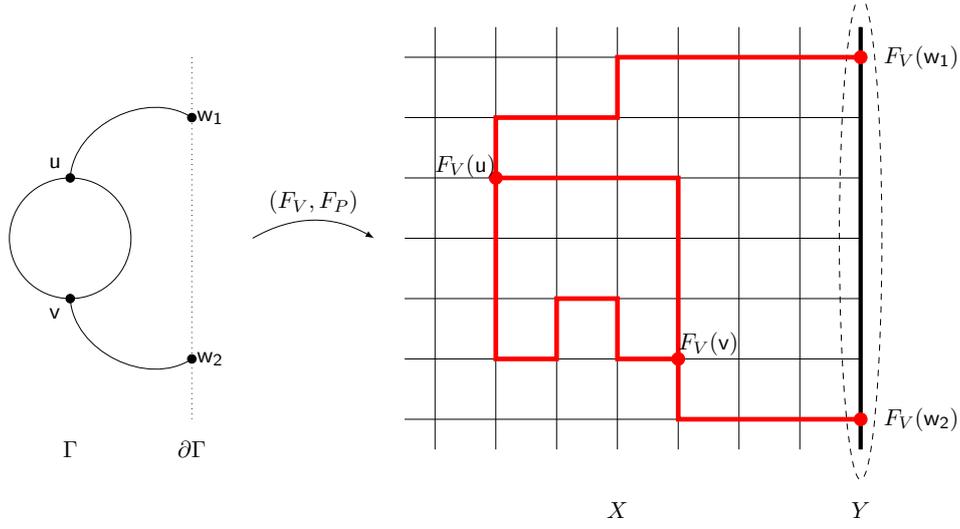
\begin{figure}[H]
        \centering
\scalebox{0.8}{
\begin{tikzpicture}
\begin{scope}[shift={(-6,3)}]
\draw[dotted] (2,3) -- (2,-3);
  \draw (0,0) circle (1cm); 
  \draw (0,1) to[bend left = 60] (2,2);
  \draw (0,-1) to[bend right = 60] (2,-2);
\node at (-0.25,1.25) {$\s{u}$};
\node at (-.25,-1.25) {$\s{v}$};
   \draw[fill = black] (0,1) circle (2pt);
    \draw[fill = black] (2,2) circle (2pt);
   \draw[fill = black] (2,-2) circle (2pt);
    \draw[fill = black] (0,-1) circle (2pt);
\draw[-latex] (3,0) to[bend left = 30] (5,0);
\node at (4,0.6) {$(F_V,F_P)$};
\node at (0,-3.5) {$\Gamma$}; 
\node at (2,-3.5) {$\partial\Gamma$}; 
\node at (2.3,2) {$\s{w_1}$};
\node at (2.3,-2) {$\s{w_2}$};
\end{scope}
\foreach \x in {0,1,...,7} 
{
\draw (\x,-0.5) -- (\x,6.5);
}
\foreach \y in {0,1,...,6} 
{\draw (-.5,\y) -- (7,\y);}
\draw[line width=2pt] (7,-.5) -- (7,6.5);
\draw[dashed] (7,3) ellipse (0.35cm and 4cm);
\draw[red, fill = red] (1,4) circle (3pt);
\draw[red, fill = red] (4,1) circle (3pt); 
\draw[red, fill = red] (7,6) circle (3pt); 
\draw[red, fill = red] (7,0) circle (3pt); 

\draw[path] (1,4) -- (1,5) -- (2,5) -- (3,5) -- (3,6) -- (7,6);
\draw[path] (1,4) -- (4,4) -- (4,1) ;
\draw[path] (1,4) -- (1,1) -- (2,1) -- (2,2) -- (3,2) -- (3,1) -- (4,1); 
\draw[path] (4,1) -- (4,0) -- (7,0); 
\node at (0.5,4.2) {$F_V(\s{u})$}; 
\node at (4.5,1.25) {$F_V(\s{v})$};
\node at (8,6) {$F_V(\s{w_1})$};
\node at (8,0) {$F_V(\s{w_2})$};
\node at (3,-1.5) {$X$};
\node at (7,-1.5) {$Y$};
\end{tikzpicture} 
}
\caption{An example of a relative  edge-to-path map.}
        \label{fig: relative edge to path}
    \end{figure}

We can now express the 
weight of a Feynman graph with boundary vertices as a sum over relative edge-to-path maps -- the combinatorial analog of the first quantization formula \eqref{eq: relative Feynman weight 1q formalism}: 
\begin{proposition}
    Suppose that $\Gamma$ is a Feynman graph with boundary vertices and $\phi \in C^0(Y)$. Then, the Feynman weight $\Phi_{\Gamma,(X,Y)}(\phi_Y)$ can be expressed by summing over relative edge-to-path maps as 
    \begin{multline}
        \Phi_{\Gamma,(X,Y)}(\phi_Y) = \\
        =\sum_{F \in P_{X,Y}^\Gamma} \prod_{\s{v}\in V^\mr{bulk}_\Gamma}(-p_{\val(\s{v})}) \prod_{\s{v}^\partial \in V^\partial_\Gamma}\phi_Y(F_V(\s{v}^\partial))
        \cdot \prod_{e \in E_\Gamma} w_{X,Y}(F_P(e)) . 
        \label{eq: Feynman weight relative paths}
    \end{multline} 
    In terms of h-paths, the expression is 
    \begin{multline}
        \Phi_{\Gamma,(X,Y)}(\phi_Y) =\\
        =(m^2)^{ - \#E^{\mr{bulk}-\mr{bulk}}} \sum_{F \in \Pi_{X,Y}^\Gamma} \prod_{\s{v}\in V^\mr{bulk}_\Gamma}(-p_{\val(\s{v})}) \prod_{\s{v}^\partial \in V^\partial_\Gamma}\phi_Y(F_V(\s{v}^\partial))
        \cdot \prod_{e \in E_\Gamma} s(F_\Pi(e)).
        \label{eq: Feynman weight relative h paths}
    \end{multline} 
\end{proposition}
\begin{proof}
In \eqref{eq: Feynman weight relative paths}  we are using  the path sum formulae \eqref{eq: path sum rel prop}, \eqref{eq: path sum extension}
Similarly, to see \eqref{eq: Feynman weight relative h paths} we are using the relative h-path sums \eqref{eq: G rel hesitant paths}, \eqref{eq: h path sum extension operator} and  notice that every bulk-bulk Green's function comes with an additional power of $m^{-2}$.
\end{proof}
We immediately obtain the following formula for the partition function: 
\begin{proposition}
    The relative perturbative partition function can be expressed as 
    \begin{multline}
        Z^\mr{pert}_{X,Y}(\phi) = \det(K_{X,Y})^{-\frac12} \cdot e^{-\frac{1}{2\hbar} ( (\phi_Y,(\DN_{Y,X}-\frac12 K_Y)\phi_Y)
        - 
        S_Y^\mr{int}(\phi_Y))} \cdot \\
        \cdot \sum_\Gamma \frac{\hbar^{-\chi(\Gamma)}}{|\mr{Aut}(\Gamma)|}\sum_{F \in P_{X,Y}^\Gamma} \prod_{\s{v}\in V^\mr{bulk}_\Gamma}(-p_{\val(\s{v})}) \prod_{\s{v}^\partial \in V^\partial_\Gamma}\phi_Y(F_V(\s{v}^\partial))
        \cdot \prod_{e \in E_\Gamma} w_{X,Y}(F_P(e)). \label{eq: Z pert path sum}
    \end{multline}
\end{proposition}
\begin{remark}
    As in Remark \ref{rem: rems on Z pert rel}, the Dirichlet-to-Neumann operator in the exponent of \eqref{eq: Z pert path sum} could be expanded in terms of Feynman diagrams with boundary-boundary edges. An edge-to-path map $F$ should map such a boundary-boundary edge $e = (\s{u}^\partial,\s{v}^\partial)$ either to a path  $\gamma \in P''_X$ (which is weighted with $w_{X,Y}(\gamma)$) or, in the case where $F_V(\s{u}^\partial) = F_V(\s{v}^\partial)$, possibly to the constant path $(\s{v}^\partial)$ (which is then weighted with $-(m^2 + \val(\s{v}^\partial))$.

    Equivalently, one has the following expression for the logarithm of the relative perturbative partition function:
    \begin{multline}
        \log\, Z^\mr{pert}_{X,Y}(\phi)=\\
        =\frac{1}{2\hbar}\Bigg(
         \sum_{u,v\in Y}\phi_Y(u)\phi_Y(v) 
       \cdot\bigg(-\left(m^2 + \val_X(v)\right)\delta_{uv} +  \sum_{\gamma \in P''_{X,Y}(u,v)}w_{X,Y}(\gamma) \bigg) +\\
       +\frac12 (\phi_Y,K_Y \phi_Y)+ \sum_{v\in Y} p(\phi_Y(v))
        \Bigg) \\
        +
        \frac{1}{2}\left(\sum_{[\gamma]\in C^{\geq 1}_{X\setminus Y}}\frac{w'_{X,Y}(\gamma)}{t(\gamma)} - \sum_{v\in X}\log (m^2 + \val(v))\right)\\
        +\sum_{\Gamma^{\mr{conn}}} \frac{\hbar^{-\chi(\Gamma)}}{|\mr{Aut}(\Gamma)|}\sum_{F \in P_{X,Y}^\Gamma} \prod_{\s{v}\in V^\mr{bulk}_\Gamma}(-p_{\val(\s{v})}) \prod_{\s{v}^\partial \in V^\partial_\Gamma}\phi_Y(F_V(\s{v}^\partial))
        \cdot \prod_{e \in E_\Gamma} w_{X,Y}(F_P(e)).
    \end{multline}
    This generalizes the results (\ref{eq: Z pert 1st quant}) and (\ref{eq: 1q formula rel Z}) to relative interacting case.
    
\end{remark}
\subsection{Cutting and gluing}
\label{ss: 1st quantization cutting-gluing a Feynman graph}
The goal of this section is to provide a sketch of a proof of the gluing of perturbative partition functions \eqref{gluing of Z^pert} by counting paths. Suppose that $X = X' \cup_Y X''$ and $F\in P_X^\Gamma$ is an edge-to-path map from a Feynman graph $\Gamma$ to $X$.\footnote{Again, for notational simplicity we consider only the case where $X$ is closed, with the generalization to cobordisms notationally tedious but straightforward.} Then, the decomposition $X = X' \cup_Y X''$ induces a decoration of $\Gamma$, as in Section \ref{sec: cutting pert Z}. Namely, we decorate a vertex $\s{v} \in V_\Gamma$ with $\alpha \in \{X',Y,X''\}$ if $F_V(\s{v}) \in \alpha$, and we decorate an edge $e$ with $c$ if and only if the path $F_P(e)$ contains a vertex in $Y$. See Figure \ref{fig: edge to path compatible}.


 \begin{figure}[H]
     
        \centering
        \scalebox{0.8}{
\begin{tikzpicture}
\begin{scope}[shift={(-6,4)}]
  \draw (0,0) to (2,0); 
  \draw (0,0) to[bend left = 30]   (.5,.5);
    \draw (.5,.5) to[bend left = 60]   (1.5,.5);
\draw (.5,.5) to[bend right = 60]   (1.5,.5);
\draw (1.5,.5) to[bend left = 30] (2,0);
  \draw (0,0) to[bend right = 60]  (2,0);
\node at (-0.25,0) {$\s{u}$};
\node at (2.25,0) {$\s{v}$};
\node at (0.2,.55) {$\s{w}$};
\node at (1.8,.55) {$\s{z}$};
   \draw[fill = black] (0,0) circle (2pt);
    \draw[fill = black] (2,0) circle (2pt);
    \draw[fill = black] (.5,.5) circle (2pt);
    \draw[fill = black] (1.5,.5) circle (2pt);
\draw[-latex] (3,0) to[bend left = 30] (5,0);
\node at (4,0.6) {$(F_V,F_P)$};

\node at (-1,1) {$\Gamma$}; 
\end{scope}
\begin{scope}[shift={(-6,1)}]
  \draw (0,0) to[bend left = 30] node[above left] {$c$} (1,.75);
  \draw (1,.75) to[bend left = 60] node[above] {$c$} (2,.75);
    \draw (1,.75) to[bend right = 60] node[below] {$c$} (2,.75);
 \draw (2,.75) to[bend left =30] node[right] {$u$} (2,0);
  \draw (0,0) to[bend left = 0] node[above left] {$c$} (2,0);
  \draw (0,0) to[bend right = 60] node[below right] {$c$} (2,0);
  \draw[dashed] (1,1.5) -- (1,-1.5);
\node at (-0.25,0) {$\s{u}$};
\node at (2.25,0) {$\s{v}$};
\node at (0.77,.85) {$\s{w}$};
\node at (2.2,.85) {$\s{z}$};
   \draw[fill = black] (0,0) circle (2pt);
   \draw[fill = black] (1,0.75) circle (2pt);
    \draw[fill = black] (2,0) circle (2pt);
        \draw[fill = black] (2,0.75) circle (2pt);

\draw[decoration=snake,latex-] (3,0) to[bend left = 20] (5,0);
\node at (4,0.6) {infer decoration};

\node at (-.5,1.5) {$\Gamma^{dec}$}; 
\node at (0,-1) {$X'$}; 
\node at (2,-1) {$X''$}; 
\node at (1,-1.5) {$Y$};
\end{scope}
\foreach \x in {0,1,...,7} 
{
\draw (\x,-0.5) -- (\x,7);
\draw ($(\x,-0.5) + (.5,0)$) -- ($(\x,7) + (.5,0)$);
}
\foreach \y in {0,1,...,6} 
{\draw (-.5,\y) -- (8,\y);
\draw ($(-.5,\y) + (0,.5)$) -- ($(8,\y) + (0,0.5)$);
}
\draw[line width = 2pt] (4,-.5) -- (4,7);
\draw[red, fill = red] (1,4) circle (3pt);
\draw[red, fill = red] (4,6) circle (3pt);
\draw[red, fill = red] (6,4.5) circle (3pt);
\draw[red, fill = red] (5,2) circle (3pt);

\draw[dashed] (4,3.25) ellipse (0.25cm and 4cm);
\draw[path] (1,4) -- (1,5) -- (2,5) -- (3,5) -- (3,6) -- (4,6) -- (6,6) -- (6,4.5);
\draw[path] (4,6) -- (4,5.5) -- (3.5,5.5) -- (3.5,4.5) -- (6,4.5);
\draw[path] (6,4.5) -- (7,4.5) -- (7,3.5) -- (5,3.5) -- (5,2);
\draw[path] (1,4) -- (4.5,4) -- (4.5,3.5) -- (3,3.5) -- (3,3) -- (4.5,3) -- (4.5,2.5) -- (3.5,2.5) -- (3.5,2) -- (4,2) -- (5,2);
\draw[path] (1,4) -- (1,1) -- (2,1) -- (2,2) -- (3,2) -- (3,1) -- (4,1) -- (4,0) -- (5,0) -- (5,2);
\node at (3.5,6.2) {$F_V(\s{w})$}; 
\node at (0.5,4.2) {$F_V(\s{u})$}; 
\node at (5.5,1.7) {$F_V(\s{v})$};
\node at (6.5,4.8) {$F_V(\s{z})$}; 

\node at (1,-1.3) {$X'$};
\node at (4,-1.3) {$Y$};
\node at (6,-1.3) {$X''$};
\end{tikzpicture} 
}
\caption{Inferring a decoration of $\Gamma$ from cutting an edge-to-path map $\Gamma \to X' \cup_Y X''$. }
        \label{fig: edge to path compatible}
        \end{figure}
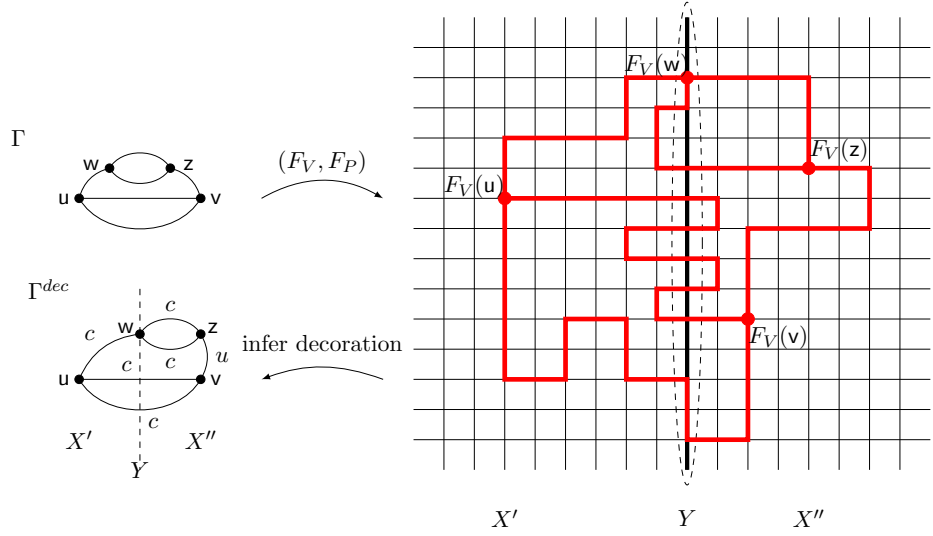

Recall that from a decorated graph, we can form two new graphs $X'$ and $X''$ with boundary vertices. Given and edge-to-path map $F$ and its induced decoration of $\Gamma$, we can define two new relative edge-to-path maps $(F'_V,F'_P)$ and $(F''_V,F''_P)$ for the new graphs $X'$ and $X''$ as follows. The map $F'_V$ is simply the restriction of $F_V$ to vertices colored $X'$. For edges labeled $u$, $F'_P(e) = F_P(e)$. For a bulk-boundary edge in $X'$,  $F'_P(e)$ is the segment of the path $F_P(\tilde{e})$ of the corresponding edge $\tilde{e}$ in $\Gamma$ (that was necessarily labeled $c$) up to (and including) the first vertex in $Y$. The construction of $(F''_V,F_P'')$ is similar, as is the extension to edge-to-hesitant-path maps. The definition of $\Gamma', \Gamma''$ ensures that $(F'_V,F'_P)$ and  $(F''_V,F_P'')$ are well-defined relative edge-to-path maps.   For example, from the edge-to-path map in Figure \ref{fig: edge to path compatible}, one obtains the two edge-to-path maps in Figure \ref{fig: edge to path cut}.
  \begin{figure} [H]
        \centering
        \scalebox{0.8}{
\begin{tikzpicture}
\begin{scope}[shift={(-6,3)}]
  \draw (0,0) -- (1,0); 
  \draw (0,0) to[bend left = 60] (1,1);
  \draw (0,0) to[bend right = 60] (1,-1);
    \draw (3,0) -- (2,0); 
  \draw (2,1) to[bend left = 0] (3,1);
  \draw (2,1.75) to[bend left = 60] (3,1);
  \draw (2,-1) to[bend right = 60] (3,0);
  \draw (3,1) to[bend left = 30] (3,0);
  \draw[dashed] (1,2) -- (1,-2);
\draw[dashed] (2,2) -- (2,-2);

\node at (-0.25,0) {$\s{u}$};
\node at (3.25,0) {$\s{v}$};
\node at (3.25,1) {$\s{z}$};
 \draw[fill = black] (0,0) circle (2pt);
 \draw[fill = black] (3,0) circle (2pt);
\draw[fill = black] (3,1) circle (2pt);
\draw (2,1.75) circle (2pt);
    \foreach \x in {1,2} 
    {
    \foreach \y in {1,0,-1} 
    {\draw (\x,\y) circle (2pt);
    }
    }
\draw[-latex] (1,2.2) to[bend left = 30] (5,2.2);
\draw[-latex] (3,-2.5) to[bend right = 30] (12,-3.7);
\node at (3,3.1) {$(F'_V,F'_P)$};
\node at (4,-4) {$(F''_V,F''_P)$};
\node at (0,-2.3) {$\Gamma'$}; 
\node at (1,-2.3) {$\partial\Gamma'$}; 
\node at (3,-2.3) {$\Gamma''$}; 
\node at (2,-2.3) {$\partial\Gamma''$}; 
\end{scope}
\foreach \x in {0,1,...,4} 
{
\draw ($(\x,-0.5) +(-.5,0)$) -- ($(\x,7) - (.5,0)$);

\draw (\x,-0.5) -- (\x,7);
}
\foreach \x in {5,6,7,8} 
{
\draw ($(\x,-0.5) +(.5,0)$) -- ($(\x,7) + (.5,0)$);

\draw (\x,-0.5) -- (\x,7);
}
\foreach \y in {0,1,...,6} 
{\draw (-.5,\y) -- (4,\y);
\draw ($(-.5,\y) + (0,.5)$) -- ($(4,\y)+ (0,.5)$);
\draw (5,\y) -- (9,\y);
\draw ($(5,\y) + (0,.5)$) -- ($(9,\y)+ (0,.5)$);}
\draw[line width = 2pt] (4,-.5) -- (4,7);
\draw[line width = 2pt] (5,-.5) -- (5,7);

\draw[red, fill = red] (1,4) circle (3pt);
\draw[red, fill = red] (6,2) circle (3pt); 
\draw[red, fill = red] (4,6) circle (3pt);
\draw[red, fill = red] (4,4) circle (3pt); 
\draw[red, fill = red] (4,1) circle (3pt); 
\draw[red, fill = red] (5,6) circle (3pt); 
\draw[red, fill = red] (5,4.5) circle (3pt); 
\draw[red, fill = red] (5,2) circle (3pt); 
\draw[red, fill = red] (5,0) circle (3pt); 
\draw[red, fill = red] (7,4.5) circle (3pt);

\draw[dashed] (4,3.25) ellipse (0.3cm and 4cm);
\draw[dashed] (5,3.25) ellipse (0.3cm and 4cm);
\draw[path] (1,4) -- (1,5) -- (2,5) -- (3,5) -- (3,6) -- (4,6); 
\draw[path] (5,6) -- (7,6) -- (7,4.5);
\draw[path] (1,4) -- (4,4);
\draw[path] (5,2) -- (6,2);
\draw[path] (1,4) -- (1,1) -- (2,1) -- (2,2) -- (3,2) -- (3,1) -- (4,1);
\draw[path] (5,0) -- (6,0) -- (6,2);
\draw[path] (5,4.5) -- (7,4.5);
\draw[path] (7,4.5) -- (8,4.5) -- (8,3.5) -- (6,3.5) -- (6,2);

\node at (0.5,4.2) {$F'_V(\s{u})$}; 
\node at (6.5,1.7) {$F''_V(\s{v})$};
\node at (7.5,4.75) {$F''_V(\s{z})$};
\node at (3,-1) {$X'$};
\node at (4,-1) {$Y$};
\node at (5,-1.3) {$Y$};

\node at (6,-1.3) {$X''$};

\end{tikzpicture} 
}
\caption{Relative edge-to-path maps $(F'_V,F'_P)$ from $\Gamma'$ to $(X',Y)$ and $(F''_V,F''_P)$ fron $\Gamma''$ to $(X'',Y)$ arising from the cutting of the Feynman graph $\Gamma^{dec}$ in Figure \ref{fig: edge to path compatible}. }
        \label{fig: edge to path cut}
    \end{figure}
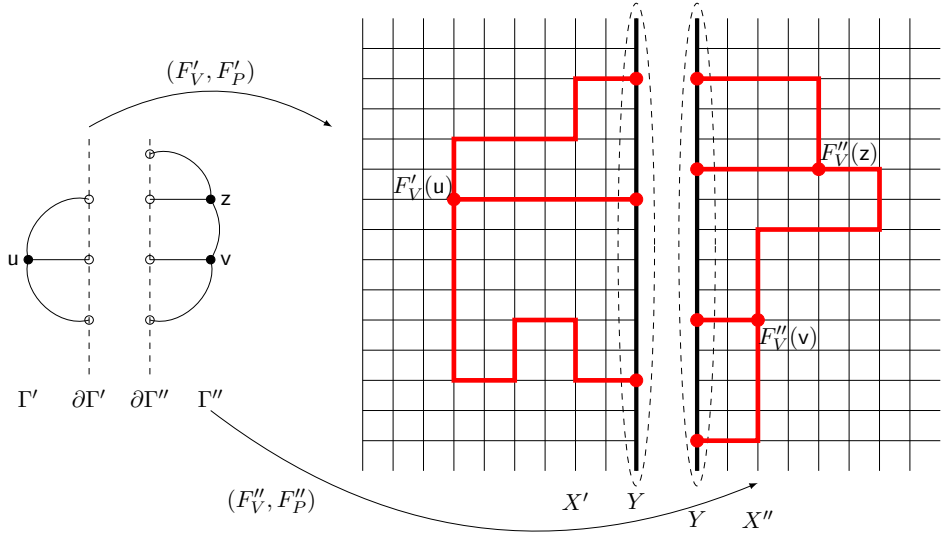
    
Notice that in the process of creating the cut edge-to-path-maps we are forgetting about the parts of the paths between the first and the last crossing of $Y$, as well as the vertices labelled with $Y$. This information is encoded in the Dirichlet-to-Neumann operator and the interacting term $S^\mr{int}_Y$ respectively. 
Integrating the product of a pair of relative edge-to-path maps appearing in the product $Z^\mr{pert}_{X',Y}(\phi_Y)Z^\mr{pert}_{X'',Y}(\phi_Y)$ over $\phi_Y$, two things happen: 
\begin{itemize}
    \item An arbitrary number of vertices on $Y$ is created (due to the factor of $e^{-\frac{1}{\hbar}S^\mr{int}_Y}(\phi_Y)$).
    \item All vertices on $Y$ (the new boundary vertices and those coming from the relative edge-to-path maps) are connected by the inverse $D$ of total Dirichlet-to-Neumann operators.
\end{itemize} In this way, we obtain all edge-to-path maps that give rise to this pair of relative edge-to-path maps. This provides a sketch of an alternative proof of the gluing formula for perturbative partition functions using the first quantization formalism, i.e. path sums. 
\section{Conclusion and outlook
}
In this paper we analyzed a combinatorial toy model for massive scalar QFT, where the spacetime manifold is a graph. We focused on incarnations of locality -- the behaviour under cutting and gluing -- and the interplay with the first quantization formalism. In particular, we showed that the convergent 
functional integrals naturally define a functor with source a graph cobordism category and target the category of Hilbert spaces, and we proposed an extended version with values in commutative unital algebras. We discussed the perturbative theory -- the $\hbar \to 0$ limit -- and its behaviour under cutting and gluing. Finally, we analyzed the  theory in the first quantization formalism, where all objects have expressions in terms of sums over paths (or h-paths) in the spacetime graph. We showed that cutting and gluing interacts naturally with those path sums.  Below we outline several promising directions for future research.
\begin{itemize}
    \item \emph{Continuum limit, renormalization, extended QFTs.} In this paper we discussed the behaviour of the theory in the continuum limit for line graphs only. However, our toy model is in principle dimension agnostic and makes sense on lattice graphs of any dimension, and one can take a similar continuum limit there. However, in the interacting theory in dimension $d \geq 2$, one has to take into account the issue of divergencies and renormalization.\footnote{In $d=2$, the only divergent subgraphs are short loops, but already those interact nontrivially with cutting and gluing - see \cite{KMW}.} It will be interesting to see how this problem manifests itself in the continuum limit of a higher-dimensional lattice graph, and whether this approach will be helpful in defining a renormalized massive scalar QFT with cutting and gluing. Furthermore, the extended QFT proposed in this paper could provide an insight in defining an extension of such a functorial QFT to higher codimensions.\footnote{
    One approach to constructing a QFT with corners involves geometric quantization of the BV-BFV phase spaces attached to corners, \cite{CMR1,IM}, see also \cite{Safronov}.
    } 

     Another interesting question in the continuum limit is to recover
    the first quantization path measure (in particular, the action functionals (\ref{S^1q}), (\ref{S^1q bar})) from a limit of our weight system on paths on a dense lattice graph.
    \item 
    \emph{Massless limit.} Another interesting problem is the study of the limit $m \to 0$. In this limit, the kinetic operator becomes degenerate if no boundary conditions are imposed, and extra work is needed to make sense of theory.\footnote{ 
    In this case it is natural to formulate the perturbative quantum answers in terms of effective actions of the zero-mode of the field $\phi$; it might be natural here to employ  the BV-BFV formalism \cite{CMRpert} combining effective actions with cutting-gluing.
    }
    This will be particularly interesting in the case of two-dimensional lattice graphs, where the massless limit of the  continuum theory is a conformal field theory  (in the free case, $p(\phi)=0$), thus the massless limit of our toy model is a discrete model for this CFT. 
    We also remark that while the h-path formulae do not interact well with the $m \to 0$ limit, since they are expansions in $m^{-2}$, for path sums the weight of a path at $m=0$ is 
    $$w(\gamma) = \prod_{v \in V(\Gamma)} \frac{1}{\val(v)},$$
    i.e. the weight of the path is the probability of a random walk on the graph where at every vertex, the walk can continue along all adjacent edges with probability $1/\val(v)$. 

    \item \emph{Gauge theories on cell complexes.} 
     Finally, it will be interesting to study in a similar fashion (including first quantization formalism) gauge theories (e.g. $p$-form electrodynamics, Yang-Mills or AKSZ theories) 
    on a cell complex, with gauge fields (and ghosts, higher ghosts and antifields) becoming cellular cochains.\footnote{  Second quantization formalism for (abelian and non-abelian) $BF$ theory on a cell complex is developed in \cite{CMR}. What we propose here is a generalization to other models, possibly involving metric on cochains, and a focus on the path-sum approach.} 
\end{itemize}


\begin{thebibliography}{99}

\bibitem{Anderson-Driver} L. Anderson, B. Driver, ``Finite dimensional approximations to Wiener measure and path integral
formulas on manifolds,'' J. Funct. Anal. 165, 430--498 (1999). 

\bibitem{Baez-Dolan} J. Baez, J. Dolan, ``Higher-dimensional algebra and topological quantum field theory,'' Journal of Mathematical Physics 36 (1998).
\bibitem{BFK} D. Burghelea, L. Friedlander, T. Kappeler, ``Meyer-Vietoris type formula for determinants of elliptic differential operators,'' J. Funct. Anal., 107.1 (1992) 34--65.
\bibitem{Carron}
G. Carron, ``D\'eterminant relatif et la fonction Xi,'' Amer. J. Math., 124.2 (2002).
\bibitem {Graph_Dirac} B. Casiday, I. Contreras, T. Meyer, S. Mi, E. Spingarn, ``Laplace and Dirac Operators on Graphs," Linear and Multilinear Algebra, published online (2022).

\bibitem{CMR1} A. S. Cattaneo, P. Mnev, N. Reshetikhin, ``Classical BV theories on manifolds with boundary,'' Communications in Mathematical Physics 332.2 (2014) 535--603.

\bibitem{CMRpert} A. S.  Cattaneo, P. Mnev, N. Reshetikhin, ``Perturbative quantum gauge theories on manifolds with boundary,'' Communications in Mathematical Physics 357.2 (2018) 631--730.

\bibitem{CMR} A. S. Cattaneo, P. Mnev, N. Reshetikhin, ``A Cellular Topological Field Theory,'' Communications in Mathematical Physics, 374,  (2020): 1229--1320.
\bibitem{Simone} S. Del Vecchio, ``Path sum formulae for propagators on graphs, gluing and continuum limit,'' ETH master thesis (2012).

\bibitem{Dijkgraaf97} R. Dijkgraaf, ``Les Houches lectures on fields, strings and duality,'' arXiv:hep-th/9703136 (1997).

\bibitem{Etingof} P. Etingof, ``Mathematical ideas and notions of quantum field theory,'' 2002.

\bibitem{Feynman-Hibbs}
R. P. Feynman, A. R. Hibbs,  ``Quantum Mechanics and Path Integrals,'' McGraw-Hill, New York (1965).

\bibitem{Glimm-Jaffe} J. Glimm, A. Jaffe, ``Quantum physics. A functional integral point of view.'' Springer-Verlag,
New York, second edition, 1987.
\bibitem{IM}  R. Iraso,  P. Mnev, "Two-dimensional Yang–Mills theory on surfaces with corners in Batalin–Vilkovisky formalism," Communications in Mathematical Physics 370 (2019): 637--702.
\bibitem{K} S. Kandel, ``Functorial quantum field theory in the
Riemannian setting,'' Advances in Theoretical and Mathematical Physics 20.6 (2016): 1443--1471.
\bibitem{KMW} S. Kandel, P. Mnev, K. Wernli, ``Two-dimensional perturbative scalar QFT and Atiyah–Segal gluing,'' Advances in Theoretical and Mathematical Physics 25.7 (2021): 1847--1952.
\bibitem{Lee}  Y. Lee, ``Mayer-Vietoris formula for the determinant of a Laplace operator on an even-dimensional manifold,'' Proc. Amer. Math. Soc. 123.6 (1995) 1933--1940.
\bibitem{Lurie} J. Lurie, ``On the classification of topological field theories,'' Current Developments in Mathematics (2009).
\bibitem{trace} P. Mnev, ``Discrete path integral approach to the Selberg trace formula for regular graphs,'' Communications in Mathematical Physics 274.1 (2007): 233--241.
\bibitem{graph_QM} P. Mnev, ``Graph quantum mechanics,'' contribution to \textit{Jahrbuch der Max-Planck-Gesellschaft} (2016).
\bibitem{Pickrell} D. Pickrell, ``$P(\phi)_2$ Quantum Field Theories and Segal’s Axioms, '' Communications in Mathematical Physics 280 (2008):403--425.
\bibitem{R} N. Reshetikhin, ``Lectures on quantization of gauge systems,'' \textit{New Paths Towards Quantum Gravity}. Berlin, Heidelberg: Springer Berlin Heidelberg, 2010. 125--190.
\bibitem{RV} N. Reshetikhin, B. Vertman, ``Combinatorial quantum field theory and gluing formula for determinants,'' Letters in Mathematical Physics 105 (2015): 309--340.

\bibitem{Safronov} P. Safronov,  ``Shifted geometric quantization,'' arXiv preprint arXiv:2011.05730 (2020).
\bibitem{Simon} B. Simon, ``The $P(\phi)_2$ Euclidean (quantum) field theory,'' Princeton University Press, Princeton, N.J., 1974. Princeton Series in Physics.
\bibitem{Vassilevich} D. V. Vassilevich, ``Heat kernel expansion: user's manual,'' Physics reports 388.5--6 (2003): 279--360.

\bibitem{Woodhouse} N. M. J. Woodhouse,``Geometric quantization,''
The Clarendon Press, Oxford University Press, New York, 1992. 


\end{thebibliography}
\end{document}